\documentclass[twoside,11pt]{article}
\usepackage{amssymb,amsmath,amsthm}           
\usepackage{hyperref}
\usepackage{eucal}
\usepackage{a4wide}
\usepackage[svgnames]{xcolor}
\usepackage{mathptmx} 
\usepackage{graphicx}
\usepackage{tikz-cd}
 \usepackage{subfigure}
 \usepackage{fancyhdr}
 \usepackage{MnSymbol}

\title{The geometry of nonholonomic Chaplygin systems revisited.}

\author{ Luis C.~Garc\'ia-Naranjo \& Juan C.~Marrero}

\numberwithin{equation}{section}
\numberwithin{table}{section}
\numberwithin{figure}{section}

\hypersetup{colorlinks=true,linkcolor=violet,citecolor=purple}

\newtheorem{theorem}{Theorem}[section]
\newtheorem{lemma}[theorem]{Lemma}
\newtheorem{proposition}[theorem]{Proposition}
\newtheorem{corollary}[theorem]{Corollary}

{\theoremstyle{definition}
\newtheorem{definition}[theorem]{Definition}
\newtheorem{remark}[theorem]{Remark}
\newtheorem*{remarks*}{Remarks}

}

\newtheorem{innercustomgeneric}{\customgenericname}
\providecommand{\customgenericname}{}
\newcommand{\newcustomtheorem}[2]{%
  \newenvironment{#1}[1]
  {%
   \renewcommand\customgenericname{#2}%
   \renewcommand\theinnercustomgeneric{##1}%
   \innercustomgeneric
  }
  {\endinnercustomgeneric}
}

\newcustomtheorem{customhyp}{}

\newcommand{\defn}[1]{{\bfseries\itshape{#1}}}

\def\headcolour{\color{Grey}}
\def\restr#1{\,\vrule height1.2ex width.4pt
  depth0.8ex\lower0.4ex\hbox{\scriptsize $\,#1$}}

\newcommand{\R}{\mathbb{R}}

\newcommand{\I}{\mathbb{I}}

\newcommand{\g}{\mathfrak{g}}
\newcommand{\so}{\mathfrak{so}}

\newcommand{\SO}{\mathrm{SO}}
\newcommand{\Ad}{\mathrm{Ad}}
\newcommand{\Ss}{\mathrm{S}}

\newcommand{\tr}{\mathop\mathrm{tr}\nolimits}

\begin{document}

\maketitle

\begin{abstract}
We consider nonholonomic Chaplygin systems and associate to them a 
 $(1,2)$ tensor field on the shape space, that we term the gyroscopic tensor, and
 that measures the interplay between the non-integrability of the 
 constraint distribution and the kinetic energy metric. We 
show how this tensor may be naturally used to derive an almost symplectic description of the reduced dynamics.
 Moreover, we express sufficient conditions for measure preservation and Hamiltonisation via Chaplygin's reducing
 multiplier method in terms of the  properties of this tensor. The theory is used to give a new proof of the remarkable Hamiltonisation
 of the multi-dimensional Veselova system obtained by Fedorov and Jovanovi\'c in~\cite{FedJov, FedJov2}.
  \end{abstract}

{\small
\tableofcontents
}

\section{Introduction}

In recent years, a great deal of   research in nonholonomic mechanics  has been  concerned with identifying a geometric structure of the equations
of motion that can guide the dynamical investigation of these systems. In particular, it is of great interest to determine geometric conditions,
usually related with symmetry, that
lead to the existence of   invariants of the flow, such as first integrals, volume forms and Poisson or symplectic structures (see e.g.~\cite{Koi,BKMM, CaCoLeMa,
ZeBo, FedJov, 
EhlersKoiller, FassoJGM,  BalseiroGN,BolsBorMam2015,FedGNMa2015} and the references therein).

This article is concerned with the geometric structure, and certain dynamical consequences that may be derived from  it, 
 of a particular kind of nonholonomic systems with symmetry: the so-called {\em Chaplygin}, $G$-{\em Chaplygin}, or {\em generalised
 Chaplygin} systems\footnote{also termed {\em the principal or purely kinematic case} in \cite{BKMM}.}$^,$\footnote{reference~\cite{BorMamHam}
gives a different meaning to the terminology ``generalised Chaplygin systems".}.
   For these
systems, the Lie group $G$ acts on the configuration space
$Q$, and the corresponding  
 tangent  lift leaves both the Lagrangian and the constraint distribution $D$
invariant. Moreover, for all $q\in Q$ one has a splitting:
\begin{equation}
\label{eq:chap-condition-intro}
T_qQ=D_q \oplus ( \g \cdot q),
\end{equation}
where $ \g$ denotes the Lie algebra of $G$ and $\g \cdot q$  the tangent space to the $G$-orbit through $q$.

 These systems were first considered by
Chaplygin and Hamel around the year 1900, and their geometric features have since been investigated by a number of authors
e.g. \cite{Baksa,Iliev1985,Stanchenko, Koi,BKMM,CaCoLeMa,FedJov,EhlersKoiller,BorMamHam,HochGN,BalseiroFdz,FedGNMa2015} and references therein. Their key feature is that the reduced equations of motion may be formulated as
an unconstrained,  forced mechanical system on the \defn{shape space} $S=Q/G$. 
The dimension $r$ of $S$ is termed the \defn{number of degrees of freedom} and,
because of \eqref{eq:chap-condition-intro}, it coincides with the rank of the
constraint distribution $D$ (in particular, the non-integrability of $D$ implies that $r\geq 2$).

Our contribution  to the subject is to highlight the relevance of a  tensor field $\mathcal{T}$ defined on $S$, that we 
term the \emph{gyroscopic tensor}, in the structure  of the reduced equations of motion and their properties. Moreover, using this tensor,
we are able to single out a very special class of systems, that we term {\em $\phi$-simple}, 
which possess an invariant measure and allow a Hamiltonisation, and for which there
exist non-trivial examples. 

\vspace{0.3cm}
\noindent {\bf The gyroscopic tensor}

The gyroscopic tensor  $\mathcal{T}$ is introduced in Definition \ref{D:defGyrTensor}. It is a $(1,2)$ skew-symmetric
 tensor field on the shape space $S$,  that  to a pair of vector fields $Y, Z$ on $S$,  assigns a third  vector field $\mathcal{T}(Y,Z)$ on $S$. The assignment  is done in a
manner that measures the  
interplay between the nonintegrability of the noholonomic constraint distribution and the kinetic energy of the system. In particular, in the case of holonomic
constraints, where the constraint distribution is integrable, we have $\mathcal{T}=0$. We mention
that   there is a close relation
between  $\mathcal{T}$ and the  geometric formulation of nonholonomic systems in terms of linear almost Poisson brackets on vector bundles~\cite{Grab,LeMaMa} (see also \cite{FedGNMa2015}).

Although the tensor $\mathcal{T}$ appears in the previous works of Koiller~\cite{Koi} and Cantrijn et al.~\cite{CaCoLeMa} (with an alternative definition
than the one that we present here), its dynamical relevance 
had not been fully appreciated until the recent work 
Garc\'ia-Naranjo~\cite{LGN18} where sufficient conditions for Hamiltonisation were given in terms of the coordinate
representation of $\mathcal{T}$. This work continues the research started in~\cite{LGN18} by providing a coordinate-free definition
of the gyroscopic tensor (Definition~\ref{D:defGyrTensor}), and studying in  depth its role in the almost symplectic structure of the equations of motion,
the conditions for the existence of an invariant measure, and the Hamiltonisation of Chaplygin systems.
In particular, the gyroscopic tensor allows us to define in a straightforward manner the  $\phi$-simple Chaplygin systems  (see below)
 which always possess an invariant measure and  allow
a Hamiltonisation via a time reparametrisation.
 Our results in all of these aspects 
are summarised below. Our point of view is that the gyroscopic tensor is the fundamental geometric object that should be considered
in the study of nonholonomic Chaplygin systems.

\vspace{0.3cm}
\noindent {\bf Almost symplectic structure of of the reduced equations}

The reduced phase space of a $G$-Chaplygin system is isomorphic to  the cotangent bundle $T^*S$ where the reduced equations
may be formulated as:
\begin{equation}
\label{eq:almostsymplectic}
{\bf i}_{\overline X_{nh}}\Omega_{nh} =dH.
\end{equation}
Here $H:T^*S\to \R$ is the (reduced) Hamiltonian, $\overline X_{nh}$ is the vector field on $T^*S$ describing the reduced dynamics,
 and $\Omega_{nh}$ is an 
\defn{almost symplectic 2-form}, namely, it is a non-degenerate 2-form that in general fails to be closed.
This structure of the equations was first noticed by Stanchenko~\cite{Stanchenko} for abelian $G$ and by Cantrijn et al.~\cite{CaCoLeMa}
in the general case.

The almost symplectic 2-form $\Omega_{nh}=\Omega_{can}+\Omega_\mathcal{T}$,
 where $\Omega_{can}$ is the canonical symplectic form on $T^*S$ and $\Omega_\mathcal{T}$ is a semi-basic
 2-form that encodes the nonholonomic reaction forces.\footnote{ 
An alternative  construction of $\Omega_\mathcal{T}$ is given in
 Ehlers et al.~\cite{KoRiEh,EhlersKoiller} (see also~\cite{HochGN})
 as the ``$\langle J , K\rangle$"-term
 which is obtained as a  pairing of the momentum map 
  of  the  $G$ action lifted to $T^*Q$
and the curvature of the  constraint distribution.  At the end of Section~\ref{SS:history} 
we prove that $\Omega_\mathcal{T}$ as defined by \eqref{eq:Om_T_intro} coincides (up to a sign) with  the  ``$\langle J , K\rangle$"-term.} In Section~\ref{SS:Almostsymplectic} we prove that the  2-form $\Omega_\mathcal{T}$
allows the following natural construction in terms of the gyroscopic tensor $\mathcal{T}$:
\begin{equation}
\label{eq:Om_T_intro}
\Omega_{\mathcal T}(\alpha) (U, V) = \alpha \left({\mathcal T}((T_\alpha \tau_S)(U), (T_\alpha \tau_S)(V))\right),
\end{equation}
for $\alpha \in T^*S$ and $U, V \in T_\alpha(T^*S)$, with $\tau_S: T^*S \to S$ the canonical projection. 
Our proof that the formulation \eqref{eq:almostsymplectic} is valid with $\Omega_{\mathcal T}$  as
above is given in Theorem~\ref{T:almostsymplectic}.

\vspace{0.3cm}
\noindent {\bf Existence of a smooth invariant measure}

 It was shown in Cantrijn et al.~\cite{CaCoLeMa}    (see also~\cite{FedGNMa2015})
 that  the reduced equations of motion 
 of a Chaplygin system possess a basic invariant measure if and only if a certain 
1-form $\Theta$ on $S$  is exact. 
In Section~\ref{SS:measures} we show that $\Theta$ is given by  the  following  ordinary contraction of the gyroscopic
tensor $\mathcal{T}$:
\begin{equation}
\label{eq:Theta-intro}
\Theta(Y)= \sum_{j=1}^r \langle X^j \, , \, \mathcal{T}(X_j,Y) \rangle ,
\end{equation}
where   $\{ X_1, \dots X_r\}$ is a basis of vector fields of $S$,   $\{ X^1, \dots X^r\}$ is the dual basis, and  $\langle \cdot \, ,  \, \cdot \rangle$ 
denotes the pairing of covectors and vectors on $S$. 

The relationship between the exactness of $\Theta$ and the existence of 
an invariant measure for the reduced equations of motion of a Chaplygin system is 
stated precisely  in  Theorem~\ref{T:measures-1-form}, for which  we give an intrinsic proof.

We note that the formulation of conditions for the existence of an invariant measure for nonholonomic systems
 in terms of the
exactness of a 1-form goes back to Blackall~\cite{Blackall}\footnote{We thank one of the anonymous referees
for indicating this reference to us.}, where, however,  the treatment is done  in local coordinates and 
not taking into account the geometric data of the problem.

\vspace{0.3cm}
\noindent {\bf  Chaplygin Hamiltonisation}

Let $g:S\to \R$ be a positive function. In view of Equation~\eqref{eq:almostsymplectic} we have
\begin{equation*}
{\bf i}_{ g \overline X_{nh}}(g^{-1} \Omega_{nh}) =dH.
\end{equation*}
If the function $g$ is such that $g^{-1} \Omega_{nh}$ is closed, then it is symplectic
and  the rescaled  vector field $g \overline X_{nh}$ is Hamiltonian. 
The rescaling by $g$ is commonly interpreted as a time reparametrisation and
one says that the system allows a \defn{Hamiltonisation}. The process described above may be understood as a geometric instance of \defn{Chaplygin's reducing multiplier method}
where one searches for a time reparametrisation that eliminates the {\em gyroscopic} reaction forces arising from the nonholonomic constraints.

It turns out that  a positive function $g\in C^\infty(S)$ satisfying that $g^{-1}\Omega_{nh}$ is closed can exist
if and only if the nonholonomic Chaplygin system is 
{\em $\phi$-simple} (defined below).

\vspace{0.3cm}
\noindent {\bf $\phi$-simple Chaplygin systems}

A main contribution of this paper is to  identify the class of  \defn{$\phi$-simple Chaplygin systems}.  This is a rather special class of Chaplygin 
systems for which there exists a smooth function $\phi:S\to \R$ such that the gyroscopic tensor satisfies
\begin{equation}
\label{eq:phi-simple-intro}
{\mathcal T}(Y, Z) =Z[\phi]Y - Y[\phi]Z, 
\end{equation}
for all vector fields $ Y, Z \in {\frak X}(S)$.

Our main result, contained in Theorem~\ref{T:measure-gyro},  shows that a nonholonomic
Chaplygin system is  $\phi$-simple if and only if the 2-form 
$\exp(\phi\circ \tau)\Omega_{nh}$ is closed. In particular  $\phi$-simple
systems always possess an invariant measure.
Moreover, if the number of degrees of freedom $r=2$,  the existence of  a basic invariant measure
is equivalent to the $\phi$-simplicity of the system.

\vspace{0.3cm}
\noindent {\bf  Weak Noetherianity  of our results}

An interesting observation is that the definition of the gyroscopic tensor  $\mathcal{T}$ only depends on the kinetic energy and on the constraints.
Therefore, any dynamical feature of the system that is derived as a consequence of the properties of $\mathcal{T}$, continues to hold in the 
presence of an arbitrary  potential  that is $G$-invariant. 
Following the terminology of Fass\`o et al~\cite{Fasso2, Fasso3}, we shall say that such
dynamical features  are   \defn{weakly Noetherian}.\footnote{Fass\`o et al used the terminology
 ``weakly Noetherian" to refer to those  first integrals of a nonholonomic system with symmetry
 that persist under the addition of an invariant potential.  }
 In particular, our treatment shows that the  Chaplygin Hamiltonisation of 
a   $\phi$-simple Chaplygin system, and the  preservation of a basic measure by  a Chaplygin system,  
are weakly Noetherian (Corollaries~\ref{C:Noetherianity-measure} and \ref{C:ChapThm} item $(iii)$).

\vspace{0.3cm}
\noindent {\bf  Chaplygin Hamiltonisation of the multi-dimensional Veselova problem}

Fedorov and Jovanovi\'c~\cite{FedJov, FedJov2} provided the first example of a nonholonomic Chaplygin system 
with arbitrary number of degrees of freedom that allows a Chaplygin Hamiltonisation. Their example is a 
multi-dimensional generalisation of the Veselova problem~\cite{Veselova, Veselov-Veselova-1988} with a  \emph{special type of inertia tensor}.

In Section~\ref{SS:Veselova} we prove that the problem considered by Fedorov and Jovanovi\'c~\cite{FedJov, FedJov2} 
is $\phi$-simple (Theorem~\ref{T:Veselova-phi-simple}) so the Chaplygin Hamiltonisation of the 
problem may be understood within our geometric framework.

 Other  examples of  $\phi$-simple Chaplygin systems with an arbitrary number of degrees of freedom
are the  multi-dimensional generalisation of the problem of a symmetric rigid body with a flat face that rolls 
without slipping or spinning over a sphere~\cite{LGN18} and the multi-dimensional rubber Routh sphere~\cite{LGNTAM-2019}.
We conjecture that other multi-dimensional Hamiltonisable Chaplygin nonholonomic systems considered by 
Jovanovi\'c~\cite{JovaRubber,Jova18} are also $\phi$-simple. If our conjecture holds, then 
 the notion of $\phi$-simplicity  places all of these
examples within a comprehensive geometric framework.

\vspace{0.3cm}
\noindent {\bf  Structure of the paper}

We begin by presenting a  brief introduction to nonholonomic systems, that recalls known results,   in Section~\ref{S:Prelim}. This 
serves to introduce the notation and makes the paper self-contained. The core of the paper is Section~\ref{S:Chaplygin} that focuses on 
the geometric study of $G$-Chaplygin systems. In this section we give a coordinate-free definition of the  gyroscopic tensor  $\mathcal{T}$ 
and prove the  results described above.
The relationship of $\mathcal{T}$ with previous constructions in the literature is described in Section~\ref{SS:history}. In Section
\ref{S:Example} we treat the examples. Apart from the treatment of the multi-dimensional Veselova problem described above, we 
illustrate our geometric constructions for the nonholonomic particle. We finish the paper by indicating some open problems in 
Section~\ref{S:FutureWork}, and with a couple of  appendices. The first appendix reviews some
geometric constructions that are necessary to give intrinsic proofs within the text and the second 
contains the proof of  a pair of technical
 lemmas that are used in Section~\ref{SS:Veselova}.

In order to help the reader to keep track of the results in the text we recall again the main results and indicate
their appearance on the text:
\begin{enumerate}
\item[$\bullet$] Definition~\ref{D:defGyrTensor} defines the 
gyroscopic tensor $\mathcal{T}$ and  Proposition~\ref{P:GyrTensor-well-defined} shows that 
the definition is unambiguous. 

\item[$\bullet$] Theorem \ref{T:almostsymplectic} proves that the reduced equations of motion
may be formulated as the  almost symplectic system \eqref{eq:almostsymplectic} in terms of the almost symplectic form $\Omega_{nh}=\Omega_{can}
+ \Omega_\mathcal{T}$, with $\Omega_\mathcal{T}$  given by \eqref{eq:Om_T_intro}.
(This is an alternative construction of $\Omega_{nh}$ with respect to previous references
\cite{Stanchenko, CaCoLeMa,EhlersKoiller}).

 \item[$\bullet$] The conditions for measure preservation in terms of the 1-form $\Theta$ given by \eqref{eq:Theta-intro} are stated in Theorem~\ref{T:measures-1-form}. This is a reformulation of the main result in
 Cantrijn et al \cite{CaCoLeMa}.
\item[$\bullet$] The main original results of the paper that state the properties of $\phi$-simple systems defined by condition \eqref{eq:phi-simple-intro} are presented in 
Theorem \ref{T:measure-gyro}.
\item[$\bullet$] Finally,  Theorem \ref{T:Veselova-phi-simple} states that the multi-dimensional Veselova problem 
considered by Fedorov and
Jovanovi\'c~\cite{FedJov, FedJov2}  is $\phi$-simple.
\end{enumerate}

\subsection*{Glossary of symbols}
\noindent

\begin{table}[h!]
{\bf Nonholonomic Chaplygin systems} \newline
\begin{tabular}{l l } 
$Q$&  $n$-dimensional configuration manifold, \\
$D\subset TQ$ & rank $r<n$ vector subbundle  defined by nonholonomic constraints, \\
$\llangle \cdot , \cdot \rrangle$ &  kinetic energy metric on $Q$, \\
$\mathcal{P}:TQ \to D$ & \begin{tabular}{l } 
bundle projector associated to the  orthogonal decomposition $TQ=D\oplus D^\perp$  \\
 defined by the kinetic energy metric $\llangle \cdot , \cdot \rrangle$,  \end{tabular} \\
 $G$ & symmetry group of dimension $n-r$ acting freely and properly on $Q$, \\ 
 $\mathfrak{g}$ & Lie algebra of  $G$, \\ 
 $S=Q/G$ & shape space (differentiable manifold of dimension $r$), \\ 
$\pi:Q\to S$  & principal bundle projection,  \\
$\mbox{hor}$  & horizontal lift associated to the principal connection $D$,   \\
$\llangle \cdot \, ,  \, \cdot   \rrangle^{-}$ & induced kinetic energy on $S$, defined by \eqref{reduced-metric}, \\
$\mathcal{T}$ & gyroscopic tensor ((1,2) tensor field on $S$), defined by \eqref{gyroscopic-tensor}, \\ 
$\tau_S:T^*S\to S$ & canonical bundle projection, \\ 
$\lambda_S\in \Omega^1(T^*S)$ & Liouville 1-form on $T^*S$ defined by \eqref{Liouville-1-form}, \end{tabular}
\end{table}

\begin{table}[h!]
\begin{tabular}{l l } 

$\Omega_{can}= -d\lambda_S$ & canonical  symplectic form on $T^*S$,  \\ 
$\nu=\Omega^r_{can}$ & Liouville volume form on $T^*S$, \\
$\Omega_{\mathcal{T}}$ & gyroscopic 2-form on $T^*S$, defined by \eqref{eq:defgyrotwoform}, \\ 
$\Omega_{nh}= \Omega_{can} +\Omega_{\mathcal{T}} $ & almost symplectic structure on $T^*S$,  \\ 
$H\in C^\infty(T^*S) $ & reduced Hamiltonian,  \\ 
$\overline{X}_{nh}\in \mathfrak{X}(T^*S)$ &   \begin{tabular}{l }  vector field describing the reduced dynamics \\ characterised by
${\bf i}_{\overline{X}_{nh}} \Omega_{nh}=dH$  (see Theorem \ref{T:almostsymplectic}),  \end{tabular}   \\ 
$\Theta \in \Omega^1(S)$ &    \begin{tabular}{l }  1-form that is exact if and only if $ X_{nh}$ preserves a basic volume form  \\
  (defined in \eqref{eq:1-form-theta-def}), (see Theorem~\ref{T:measures-1-form}),  \end{tabular} \\
 $ \Psi: TQ \times_Q TQ \to {\frak g}$ & curvature form of the principal connection $D$ (defined by \eqref{curvature-connection}).
\end{tabular}
\end{table}

\begin{table}[h!]
{\bf Geometry} \newline
\begin{tabular}{l l } 
$\Omega^1(M)$ & space of 1-forms on the manifold $M$, \\
$\mathfrak{X}(M)$ & space of vector fields on the manifold $M$, \\
$\Gamma(E)$ & space of sections of the vector bundle $E$, \\
$Y^{\ell}\in C^\infty (E^*)$ &  linear function induced by $Y\in \Gamma(E)$, where $E$ is a vector bundle (defined by \eqref{eq:linear-function}), \\
$\beta^{\sharp} \in {\frak X}(S)$  & metric dual of the 1-form $\beta \in \Omega^1(S)$ (defined by \eqref{eq:metric-dual}), \\
$Y^\flat \in \Omega^1(S)$ & metric dual of the vector field $Y\in  {\frak X}(S)$  (defined by \eqref{eq:metric-dual}),
\\
$\gamma^{\bf v}\in \mathfrak{X}(T^*S)$  &   vertical lift of $\gamma \in \Omega^1(S)$ (defined by  \eqref{vertical-lift-dual}), \\
$Y^{*c}\in \mathfrak{X}(T^*S)$ & complete lift  of $Y\in\mathfrak{X}(S)$    (defined by  \eqref{complete-lift-dual}), \\
 $\Xi^{\mathfrak{q}}$ &    \begin{tabular}{l }   quadratic function  on $T^*S$ associated to the  \\
 type $(2, 0)$ tensor $\Xi$ on $S$,  (defined by  \eqref{eq:quad-function}).  \end{tabular}\\
  \begin{tabular}{l }   ${\sharp}_m:T_m^*M\to T_mM$  \\
$ {\flat}_m:T_mM\to T_m^*M$ \end{tabular} & \begin{tabular}{l} musical isomorphisms on the Riemannian manifold $M$  \\ 
at the point $m\in M$ (defined by  \eqref{eq:metric-dual}, \eqref{eq:musical-isomorphisms}). \end{tabular}
 \end{tabular}
\end{table}

\newpage

\section{Preliminaries}
\label{S:Prelim}

We present  a rather synthetic review of known results in the theory of nonholonomic systems that sets the notation and basic 
notions that will be used in Section~\ref{S:Chaplygin} ahead.

\subsection{Nonholonomic systems}
A nonholonomic system consists of a triple $(Q,D,L)$ where $Q$ is the configuration space, $D\subset TQ$ is a vector sub-bundle  whose fibres define a  
non-integrable constraint distribution on $Q$
and $L$ is the Lagrangian.  The configuration space $Q$ is an $n$-dimensional smooth manifold, $D$ has rank $2\leq r<n$ and models $n-r$ independent linear constraints on the velocities  of the system, and the Lagrangian $L:TQ\to \R$ is of mechanical type, namely
\begin{equation*}
L=K-U,
\end{equation*}
 where the kinetic energy $K$ defines a Riemannian metric $\llangle \cdot , \cdot \rrangle$ on 
 $Q$ and $U:Q\to \R$ is the potential energy.

The vector sub-bundle  $D\subset TQ$ is the velocity phase space of the system  and the dynamics are described by a second order vector
field on $D$  which is determined by the Lagrange-D'Alembert principle of ideal constraints. Classical
references on the subject are
\cite{NeiFu,VF}. For an intrinsic definition of the  vector field describing the dynamics 
see, for instance, \cite{LeMa}.

An equivalent formulation of the dynamics may be given in the momentum phase space $D^*$ (the dual  bundle of $D$) that is a rank $r$ vector bundle over $Q$ which is isomorphic to $D$. The space $D^*$ is equipped with an {\em almost Poisson structure} which codifies the  
reaction forces in a geometric manner. The dynamics on $D^*$ is described by a vector
field which is obtained as the contraction of the almost Poisson structure on $D^*$ and the {\em Hamiltonian function}, which is the energy of the system.
This formulation has its origins in \cite{vdSM, Marle,Ibort}. In the following section we review  this construction, which is useful for our purposes. Our
description follows closely the exposition in \cite{Grab,LeMaMa} (see also \cite{FedGNMa2015}).

\subsection{Almost Poisson formulation of nonholonomic systems}
\label{SS:APoissonNonho}

Let $i_D:D\hookrightarrow TQ$ be the canonical bundle inclusion and 
$\mathcal{P}:TQ \to D$ 
the bundle projection associated to the orthogonal decomposition $TQ=D\oplus D^\perp$ defined by the kinetic energy metric.
Passing to the dual spaces we respectively get the bundle projection and the bundle inclusion
\begin{equation*}
i_D^*:T^*Q \rightarrow D^*, \qquad \mathcal{P}^*: D^* \hookrightarrow T^*Q.
\end{equation*}
 
The \defn{nonholonomic bracket}
 of the functions $\varphi, \psi \in C^\infty (D^*)$ is the smooth function on $D^*$  defined by
\begin{equation}\label{non-holo-bracket}
\{\varphi, \psi \}_{D^*} := \{\varphi \circ i_D^*, \psi \circ i_D^* \}_{T^*Q}  \circ \mathcal{P}^*,
\end{equation}
 where $\{\cdot , \cdot  \}_{T^*Q}$ denotes the canonical Poisson bracket on the cotangent bundle $T^*Q$ \cite{LeMaMa}. 
 The nonholonomic bracket is skew-symmetric and satisfies Leibniz rule. On the other hand, the Jacobi
 identity is satisfied if and only if $D$ is integrable and the constraints are holonomic (see Theorem~\ref{Integra-almost-Poisson} below). For nonholonomic constraints, the failure of the Jacobi 
 identity leads to the notion of an \defn{almost Poisson bracket}.

The Lagrangian $L$ passes via the usual Legendre transform to the Hamiltonian function $h \in C^\infty(T^*Q)$, which is the sum of the kinetic and
potential energy of the system.  Explicitly, for $\alpha \in T_q^*Q$ we have
\begin{equation*}
h(\alpha)=\frac{1}{2}\| \alpha \|^2 + U(q),
\end{equation*}
where $\| \cdot \|$  denotes the norm on the fibres of $T^*Q$ induced by the kinetic energy Riemannian metric $K$ on $Q$.
The \defn{constrained Hamiltonian} is defined by $h_c:= h\circ  \mathcal{P}^* \in  C^\infty(D^*)$.

The dynamics of the system is defined by the flow of the vector field $X_{nh}$ on $D^*$ defined as the derivation
\begin{equation}\label{non-holo-dynamics}
X_{nh}[\varphi] = \{\varphi, h_c \}_{D^*}, \qquad \varphi \in C^\infty(D^*).
\end{equation}
Local expressions for the nonholonomic bracket, for $X_{nh}$, and the corresponding equations of motion are given below.

\subsection*{Linear structure of the nonholonomic bracket}

The momentum phase space $D^*$ is a vector bundle $\tau:D^*\to Q$ and, according to this structure, it is convenient to give special attention  to 
two kinds of functions.  \defn{Linear functions} on $D^*$ are characterised by
being linear when restricted to the fibres. On the other hand  \defn{basic functions}  on $D^*$ only depend on the base point. 
We now review how the  nonholonomic  bracket is determined by
its value on functions that are either basic or linear.

We begin by noting that there is a one-to-one correspondence between the space of  linear functions on $D^*$ and the space of sections $\Gamma (D)$.
 Such correspondence is the following: to a section $Y\in \Gamma(D)$ 
we associate the function $Y^\ell\in C^\infty(D^*)$ given by
\begin{equation*}
Y^\ell(\alpha ) = \alpha (Y( \tau (\alpha))), \qquad \alpha \in D^*.
\end{equation*}
On the other hand, smooth functions on $Q$ are in one-to-one correspondence with  basic functions on $D^*$. 
If $f$ is a basic function on $D^*$ we denote by  $\tilde f\in C^\infty(Q)$ the unique function satisfying    $f= \tilde f\circ \tau$.
In the following proposition, and for the rest of the paper,  $[\cdot , \cdot ]$ denotes the Jacobi-Lie bracket of vector fields.

\begin{proposition} {\em (\cite[Proposition 2.3]{LeMaMa})}
\label{linearity-bracket}
Let   $Y^\ell, Z^\ell $ be the linear functions on $D^*$ corresponding to the sections $Y,Z\in \Gamma(D)$, and let $f, k$ be basic functions
on $D^*$. We have
\begin{equation*}
\begin{split}
\{ Y^\ell, Z^\ell \}_{D^*} = - (\mathcal{P} [Y,Z])^\ell, \qquad
 \{  f , Y^\ell \}_{D^*} = Y[ \tilde f] \circ \tau, \qquad  \{  f , k \}_{D^*}=0,
\end{split}
\end{equation*}
where $f= \tilde f\circ \tau$.
\end{proposition}
\noindent In particular, the proposition shows that the nonholonomic bracket satisfies the following  properties:
\newline $\bullet$ The bracket of linear functions is linear.
\newline  $\bullet$ The bracket of a linear and a basic function is basic.
\newline $\bullet$ The bracket of basic functions vanishes.
\newline
Brackets with these properties often appear  in mechanics and were termed \defn{linear brackets} in \cite{LeMaMa} (see also \cite{Grab}).

\subsection*{Local expressions}

Consider a local basis  $\{e_i\}_{i=1}^r$ of sections of $D$ in an open subset of $Q$ with local coordinates $(\tilde{q}^{a})$, 
$a=1, \dots, n$.
Let
\begin{equation*}
q^{a} := \tilde{q}^{a}\circ \tau, \quad a=1,\dots, n,  \qquad p_i:=e_i^\ell, \quad i=1,\dots, r.
\end{equation*}
Then 
$(q^{a} , p_i)$ is a system of local coordinates on $D^*$  consisting of basic and linear  functions on $D^*$. Proposition~\ref{linearity-bracket}
 implies that the nonholonomic bracket is determined in this coordinate system by the relations
\begin{equation*}
\{q^a, q^b\}_{D^*}=0,  \qquad \{q^{a}, p_i\}_{D^*}=\rho_i^a, \qquad \{p_i, p_j\}_{D^*}=-\sum_{k=1}^rC_{ij}^kp_k,
\end{equation*}
 where the coefficients $\rho_i^a$ and $C_{ij}^k$ depend on $(q^a)$ and are defined by the relations
 \begin{equation*}
e_i= \sum_{a=1}^n \rho_i^a \frac{\partial}{\partial  \tilde{q}^{a}}, \qquad \mathcal{P}[e_i,e_j]=\sum_{k=1}^rC_{ij}^ke_k.
\end{equation*}
In view of~\eqref{non-holo-dynamics}, the equations of motion in these variables take the form
\begin{equation*}
\dot q^a = \sum_{i=1}^r \rho_i^a \frac{\partial h_c}{\partial p_i}, \qquad \dot p_i = - \sum_{a=1}^n \rho_i^a \frac{\partial h_c}{\partial q^a}
-\sum_{k=1}^rC_{ij}^kp_k\frac{\partial h_c}{\partial p_j}.
\end{equation*}
The specific form of the constrained Hamiltonian in these variables is 
\begin{equation*}
h_c(q^a,p_i)=\frac{1}{2}\sum_{i,j=1}^rK^{ij}(q)p_ip_j+U(q),
\end{equation*}
where  $K^{ij}$ are the entries of the inverse matrix of the positive definite matrix with entries $K_{ij}=\llangle e_i, e_j\rrangle$.

It is shown in \cite{vdSM} that the bracket $\{ \cdot, \cdot\}_{D^*}$ satisfies the Jacobi identity if and only if the 
distribution $D$ is integrable and hence the constraints are holonomic.
 Although we have no need in this paper for this fact, we take the opportunity to 
present a coordinate-free proof.

\begin{theorem}\label{Integra-almost-Poisson}
The almost Poisson bracket $\{\cdot, \cdot \}_{D^*}$ satisfies the Jacobi identity
if and only if the distribution $D$ is integrable.
\end{theorem}
\begin{proof}
Throughout the proof, for $\varphi \in C^\infty(D^*)$ we denote  by $X_\varphi$ the  vector field on $D^*$ 
defined as the derivation $X_\varphi [\psi] = \{\psi, \varphi \}_{D^*}$, $\psi  \in C^\infty(D^*)$.

Suppose that $\{\cdot, \cdot\}_{D^*}$ satisfies the Jacobi identity so it is a Poisson bracket. Because of  Frobenius theorem,
in order to show that $D$ is integrable, 
 it is enough to prove that
${\mathcal P}[Y, Z] = [Y, Z]$
for any sections $Y, Z$  of $D$.
Considering that
\[
\{  f , Y^\ell \}_{D^*} = Y[ \tilde f] \circ \tau,
\]
for $f = \tilde{f} \circ \tau$,  we deduce that $X_{Y^\ell}$ is $\tau$-projectable on $Y$. Consequently, 
the Lie bracket $[X_{Y^\ell}, X_{Z^\ell}]$ 
is   $\tau$-projectable on  $[Y, Z]$.
On the other hand, using our assumption that  $\{\cdot, \cdot\}_{D^*}$ is a Poisson bracket we obtain
\[
[{X}_{Y^\ell}, {X}_{Z^\ell}] = -{X}_{\{Y^\ell, Z^\ell\}_{D^*}} = {X}_{{\mathcal P}[Y, Z]^\ell},
\]
which implies that the vector field $[{X}_{Y^\ell}, {X}_{Z^\ell}]$ is $\tau$-projectable on ${\mathcal P}[Y, Z]$.
Therefore,
$
{\mathcal P}[Y, Z] = [Y, Z],
$
and $D$ is integrable. 

Conversely, assume that $D$ is integrable. Then, we have
$
{\mathcal P}[Y, Z] = [Y, Z]
$
for any sections $Y, Z$  of $D$. Using this in Proposition~\ref{linearity-bracket}, shows that the Jacobi identity holds for
basic and linear functions. Namely, 
\[
\{\{\varphi, \psi\}_{D^*}, \mu\}_{D^*} +\{\{\mu, \varphi \}_{D^*}, \psi \}_{D^*}+\{\{ \psi , \mu\}_{D^*}, \varphi\}_{D^*}= 0,
\]
if $\varphi$, $\psi$ and $\mu$ are linear or basic functions on $D^*$. Since the bracket is determined by its value on these
kinds of functions, we conclude that the Jacobi identity holds for general functions on $D^*$.
\end{proof}

\subsection{Nonholonomic systems with symmetries and reduction}
\label{SS:nh-symmetry}

For the purposes of this paper a \defn{nonholonomic system with symmetry} is a nonholonomic system $(Q,D,L)$ together with a Lie group $G$, that acts freely
 and properly on $Q$, and satisfies the following properties:
\begin{enumerate}
\item $G$ acts by isometries on $Q$ and the potential energy $U$ is $G$-invariant,
\item $D$ is $G$ invariant in the sense that $Tg(D_q)=D_{g\cdot q}$ for all $g\in G$.
\end{enumerate}

We shall now give a description of the reduction of the system in terms of almost Poisson structures.
We begin with the following.
\begin{proposition}
\label{p:symmetric-nonho}
Consider a nonholonomic system with symmetry $(Q,D,L)$ with symmetry group $G$. Then $G$ defines a free and proper  action on  $D^*$ that leaves 
the constrained Hamiltonian and the nonholonomic bracket invariant. Moreover, the vector field $X_{nh}$ on $D^*$ that describes the dynamics is
equivariant. 
\end{proposition}
\begin{proof}

Recall that the tangent lift of the action of $G$ on $Q$ is a free and proper action of $G$ on $TQ$ defined by
\[
g \cdot v := (T_qg)(v)\in T_{g\cdot q}Q, 
\]
where $g \in G$, $v \in T_qQ$ and $q \in Q$. The $G$-invariance of $D$ implies that this action restricts to a free and proper
action of $G$ on $D\subset TQ$.

Recall also that the cotangent lift  defines a free and proper action of $G$ 
on $T^*Q$, sending $\alpha \in T^*_qQ$ into the covector $g \cdot \alpha \in T^*_{g\cdot q}Q$ that is defined by
\[
 (g \cdot \alpha)(u) := \alpha(g^{-1}\cdot u),
\]
where  $u \in T_{g\cdot q}Q$. As before, the $G$-invariance of 
 $D$ implies that  this action restricts to a free and proper action of $G$ on $D^*$. Indeed, such action is
 defined by the above formula but with the
 restrictions that  $\alpha \in D^*_q$ and 
 the tangent vector $u\in D_{g\cdot q}$. This proves the first statement of the proposition.

Next, as a consequence of the invariance of $D$ and of the kinetic energy metric,
it follows that both the  projector ${\mathcal P}: TQ \to D$, and the dual morphism ${\mathcal P}^*: D^* \to T^*Q$, are 
$G$-equivariant. Namely,
\begin{equation}
\label{P-star-equivariant}
{\mathcal P}(g\cdot v) = g \cdot {\mathcal P}(v), \qquad {\mathcal P}^*(g\cdot \alpha) = g \cdot {\mathcal P}^*(\alpha),
\end{equation}
for $v \in TQ$, $\alpha \in D^*$.
On the other hand,  our assumptions clearly imply that the Hamiltonian $h$ is also  invariant and therefore  the same is 
true about the constrained Hamiltonian $h_c=h\circ  {\mathcal P}^*$ as claimed.

Now recall that the cotangent lifted action of $G$ on $T^*Q$ preserves the canonical   Poisson bracket $\{\cdot, \cdot\}_{T^*Q}$  (see e.g. \cite{MaRa}). Moreover, 
in virtue of the invariance of $D$, both the canonical inclusion $i_D: D \to TQ$ and the dual projection $i_D^*: T^*Q \to D^*$ are $G$-equivariant.  These observations,
together with \eqref{P-star-equivariant} show that the nonholonomic bracket defined by \eqref{non-holo-bracket} is also $G$-invariant.

Finally, the equivariance of  $X_{nh}$ follows from its definition~\eqref{non-holo-dynamics} and the above observations.
\end{proof}

Denote by $\overline{D^*}:=D^*/G$ the orbit space which, as a consequence of the above proposition, is a smooth manifold, and let $\Pi:D^*\to  \overline{D^*}$
be the orbit projection which is a surjective submersion. The invariance of the nonholonomic bracket proved above implies the  existence of a well-defined almost Poisson bracket 
$\{\cdot, \cdot \}_{\overline{D^*}}$ 
on the reduced space $\overline{D^*}$ defined by the restriction of the nonholonomic bracket to invariant functions on $D^*$. In other words
\begin{equation}\label{bracket-reduced}
\{\bar \varphi  \circ \Pi , \bar \psi  \circ \Pi \}_{D^*} = \{ \bar \varphi   ,  \bar \psi  \}_{\overline{D^*}}  \circ \Pi, \qquad \mbox{for $\bar \varphi   , \bar \psi \in C^\infty(\overline{D^*})$.}
\end{equation}

Note that the orbit space $\overline{D^*}$ is a vector bundle over the \defn{shape space} $S := Q/G$, so the reduced bracket $\{\cdot , \cdot \}_{\overline{D^*}}$ may also be described
by its value on linear and basic functions. In order to give such description, first notice that the space of linear functions on $\overline{D^*}$ may be 
identified with the space $\Gamma(D)^G$ of $G$-equivariant sections of $D$:
\[
\Gamma(D)^G := \{ X \in \Gamma(D) \, :\,  X \mbox{ is $G$-equivariant }\}.
\]
Moreover, if $X \in \Gamma(D)^G$ then $X$ is $\pi$-projectable to a vector field $\tilde{X}$ on $S = Q/G$, where $\pi: Q \to S$ denotes the principal bundle projection. 
The following proposition is a direct consequence of Proposition~\ref{linearity-bracket} and Equation~\eqref{bracket-reduced}.
\begin{proposition}\label{linearity-bracket-red}
Let   $Y^\ell, Z^\ell $ be the linear functions on $\overline{D^*} = D^*/G$ corresponding to the sections $Y,Z\in \Gamma(D)^G$, and let $f, k$ be basic functions on $\overline{D^*} = D^*/G$. Then,
\begin{equation*}
\begin{split}
\{ Y^\ell, Z^\ell \}_{\overline{D^*}} = - (\mathcal{P} [Y,Z])^\ell, \qquad
 \{  f , Y^\ell \}_{\overline{D^*}} = \tilde{Y}[ \tilde f] \circ \overline{\tau}, \qquad  \{  f , k \}_{\overline{D^*}}=0,
\end{split}
\end{equation*}
where $f= \tilde f\circ \overline{\tau}$, $\overline{\tau}: \overline{D^*} \to S$ is the vector bundle projection, and $\tilde Y$ denotes the unique vector field on $S$
that is $\pi$-related to $Y$.
\end{proposition}

On the other hand, the invariance of the constrained Hamiltonian $h_c$, guaranteed by Proposition~\ref{p:symmetric-nonho}, implies the existence of a \defn{reduced Hamiltonian} 
$H\in C^\infty(\overline{D^*})$ 
such that $h_c=H\circ \Pi$. Also, the equivariance of $X_{nh}$ implies the existence of a reduced vector field $\overline{X}_{nh}$ on $\overline{D^*}$, that is 
$\Pi$-related to $X_{nh}$ and describes the reduced dynamics of the system. As one may expect, we have:
\begin{proposition}
The  reduced vector field  $\overline{X}_{nh}$ may be described in an almost Poisson manner with respect to the reduced almost Poisson bracket $\{\cdot, \cdot \}_{\overline{D^*}}$ 
and the reduced Hamiltonian $H\in C^\infty(\overline{D^*})$. In other words
\begin{equation*}
\label{eq:AlmostPoisson-redSpace}
\overline{X}_{nh}[\varphi] =\{\varphi   , H  \}_{\overline{D^*}}, \qquad \mbox{for all $ \varphi   \in C^\infty(\overline{D^*})$.}
\end{equation*}
\end{proposition}
\begin{proof}
Let  $\varphi \in C^{\infty}(\overline{D^*})$. Using \eqref{non-holo-dynamics} and the fact that $X_{nh}$ is $\Pi$-projectable to $\overline{X}_{nh}$, we have
\[
\overline{X}_{nh}[\varphi] \circ \Pi = X_{nh}[ \varphi \circ \Pi ] = \{\varphi \circ \Pi, h_c \}_{D^*} = \{\varphi \circ \Pi, H \circ \Pi\}_{D^*}=\{\varphi, H\}_{\overline{D^*}} \circ \Pi,
\]
where we have used \eqref{bracket-reduced} in the last equality. Equation~\eqref{eq:AlmostPoisson-redSpace} follows from the above relation since
$\Pi$ is surjective.
\end{proof}

\section{Geometry of nonholonomic Chaplygin systems revisited}
\label{S:Chaplygin}

We now come to our main subject of study which are nonholonomic {\em $G$-Chaplygin systems}. Roughly speaking a nonholonomic $G$-Chaplygin system is a nonholonomic system with symmetry group $G$ for which the symmetry directions
 are incompatible with the constraints. An example is a ball that rolls without slipping on a horizontal plane with symmetry group $G=\R^2$ acting by horizontal 
 translations. The system is obviously invariant under a horizontal translation of the origin of the inertial frame. However,
  a pure horizontal translation of the ball that does not involve rolling violates the nonholonomic constraint. 
The precise definition is the following.
\begin{definition}
\label{D:ChaplyginSyst}
A nonholonomic $G$-\defn{Chaplygin system} is a nonholonomic system with symmetry as defined in section~\ref{SS:nh-symmetry} for which the
following splitting is valid for all $q\in Q$:
\begin{equation}
\label{eq:chap-condition}
T_qQ=D_q \oplus ( \g \cdot q),
\end{equation}
where $ \g$ denotes the Lie algebra of $G$ and $\g \cdot q$  the tangent space to the $G$-orbit through $q$.
\end{definition}

\begin{remark} Suppose that the curve $(q(t),\dot q(t))$ is a solution of the nonholonomic system $G$-Chaplygin system 
with the property that $q(t)$ is contained in a $G$-orbit on $Q$ for all $t\in \R$.
The transversality condition \eqref{eq:chap-condition}, and the nonholonomic constraints $\dot q(t)\in D_{q(t)}$, imply that 
 $\dot q(t)=0$ and hence $q(t)=q_0$, and the solution is an equilibrium of the system. This shows that  the only relative equilibria of a nonholonomic
$G$-Chaplygin system are actual equilibria.
\end{remark}

The study of Chaplygin systems goes back to Chaplygin. There are many references in the literature that focus on the geometry of these systems \cite{Stanchenko, Baksa, Koi, BKMM, EhlersKoiller, CaCoLeMa}.
As mentioned in the introduction, the purpose of this paper is to 
 show that the main features of $G$-Chaplygin systems are conveniently encoded in the  {\em gyroscopic tensor} which is a $(1,2)$-tensor field on 
$S$ that  measures the interplay between the kinetic energy metric and the non-integrability of the constraint distribution. 
We also identify some conditions on the gyroscopic tensor that imply measure preservation and Hamiltonisation. The relationship between 
our definition of the gyroscopic tensor and other tensors that have appeared before in 
the  literature is discussed in subsection \ref{SS:history}.

We shall denote by $\pi:Q\to Q/G:=S$ the principal bundle projection, and continue to refer to the base manifold 
$S$ as  the  \defn{shape space}. Note that condition  \eqref{eq:chap-condition}
forces the dimension of $S$
to coincide with the rank $r \geq 2$ of $D$, and the dimension of $G$ to be $n-r$. We will say that the Chaplygin nonholonomic
system has \defn{$r$ degrees of freedom}.

\subsection{The gyroscopic tensor}

In order to give the definition of the gyroscopic tensor, we recall from  Koiller \cite{Koi} that the condition~\eqref{eq:chap-condition} implies that 
the fibres of $D$  may be interpreted as the horizontal spaces of a principal connection on the principal $G$-bundle $\pi : Q\to S$.  Such a principal 
connection defines a \defn{horizontal lift } that to a vector field $Y$ on $S$ assigns the equivariant vector field hor $Y$ on $Q$ taking values on the fibres of $D$ and  
which is $\pi$-related to $Y$.

\begin{definition}
\label{D:defGyrTensor}
 Let $Y, Z \in \mathfrak{X}(S)$. The  \defn{gyroscopic tensor}  $\mathcal{T}$ is defined by assigning to $Y, Z$ the following vector
field on $S$:
\begin{equation}\label{gyroscopic-tensor}
\mathcal{T}(Y,Z)(s)=(T_q\pi) \left ( \mathcal{P} \left [ \mbox{ hor  $Y$} \, , \, \mbox{hor $Z$} \,  \right ] (q)\right ) - [Y,Z](s),
\end{equation}
for $s \in S$, with $q \in Q$ and $\pi(q) = s$, and where $\mathcal{P}:TQ\to D$ denotes the orthogonal projector.
\end{definition}

 We begin by proving that $\mathcal{T}$ is well defined and is indeed a tensor.

\begin{proposition}
\label{P:GyrTensor-well-defined}
The gyroscopic tensor $\mathcal{T}$ is well defined and is a skew-symmetric tensor field of type $(1,2)$ on $S$.
\end{proposition}
\begin{proof}
First we prove  that ${\mathcal T}$ is well defined.
Let $q, q' \in Q$  such that $\pi(q) = \pi(q')$. Then there exists $g \in G$ satisfying
$
q' = g\cdot q.
$
Since  $\mbox{ hor } Y$ and $\mbox{ hor }Z$ are equivariant, the same is true about their  Lie bracket,   and hence
\begin{equation*}
\left [ \mbox{ hor  $Y$} \, , \, \mbox{hor $Z$} \,  \right ] (g \cdot q) = g\cdot  \left [ \mbox{ hor  $Y$} \, , \, \mbox{hor $Z$} \,  \right ] (q).
\end{equation*}
This equation, together with the equivariance of the projector ${\mathcal P}$ shown in Equation~\eqref{P-star-equivariant} above, implies
\[
{\mathcal P} \left [ \mbox{ hor  $Y$} \, , \, \mbox{hor $Z$} \,  \right ](g \cdot q) = g \cdot {\mathcal P} \left [ \mbox{ hor  $Y$} \, , \, \mbox{hor $Z$} \,  \right ](q).
\]
Therefore,
\[
(T_{q'}\pi) \left ( \mathcal{P} \left [ \mbox{ hor  $Y$} \, , \, \mbox{hor $Z$} \,  \right ] (q')\right ) = (T_q\pi) \left ( \mathcal{P} \left [ \mbox{ hor  $Y$} \, , \, \mbox{hor $Z$} \,  \right ] (q)\right ),
\]
and ${\mathcal T}$ is well defined.

Now,  it is clear that ${\mathcal T}$ is  $\R$-bilinear, so, in order to prove that $\mathcal{T}$ is a tensor field, we only need to show that 
\begin{equation}\label{linearity-functions}
{\mathcal T}(fY, Z) = f {\mathcal T}(Y, Z), \; \; \mbox{ for } f \in C^{\infty}(S).
\end{equation}
To prove this first notice that $$ \mbox{ hor  $fY$} = (f\circ \pi)\mbox{hor $Y$},$$
so, using that  $\mbox{hor $Z$}$ is invariant and $\pi$-projectable onto $Z$ together with the standard properties of the Lie bracket, we have
\begin{equation*}
\begin{split}
 [  \mbox{ hor  $fY$}  \, , \, \mbox{hor $Z$}] (q) &=(f\circ \pi) (q) [  \mbox{ hor  $Y$}  \, \, \mbox{hor $Z$}] (q) -  ( \mbox{hor $Z$} )[f\circ \pi](q)  \mbox{ hor  $Y$}(q) \\ &=
f(s)  [  \mbox{ hor  $Y$}  \, , \, \mbox{hor $Z$}] (q) - Z[f](s) \mbox{ hor  $Y$}(q),
\end{split}
\end{equation*}
where $\pi(q)=s$. Since $\mbox{ hor } Y$ is a section of $D$, then ${\mathcal P}(\mbox{ hor }Y) = \mbox{ hor } Y$, and therefore,
\begin{equation*}
\begin{split}
\mathcal{P} [  \mbox{ hor  $fY$}  \, , \, \mbox{hor $Z$}] (q)  =
f(s) \mathcal{P}  [  \mbox{ hor  $Y$}  \, , \, \mbox{hor $Z$}] (q) - Z[f](s) \mbox{ hor  $Y$}(q).
\end{split}
\end{equation*}
Finally, given that  $\mbox{ hor  $Y$}$ is $\pi$-projectable onto $Y$ we obtain
\begin{equation}
\label{eq:tensorproof1}
(T_q\pi)( \mathcal{P} [  \mbox{ hor  $fY$}  \, , \, \mbox{hor $Z$}] (q)) = - Z[f](s) Y(s) + f(s) (  T_q\pi )( \mathcal{P}  [  \mbox{ hor  $Y$}  \, , \, \mbox{hor $Z$}] (q)).
\end{equation}
On the other hand, we have
\begin{equation}
\label{eq:tensorproof2}
[fY,Z](s) = f(s) [Y,Z](s) - Z[f](s) Y(s).
\end{equation}
The proof of \eqref{linearity-functions} follows immediately by substituting equations \eqref{eq:tensorproof1} and \eqref{eq:tensorproof2} into the
 definition \eqref{gyroscopic-tensor} of the gyroscopic tensor $\mathcal T$.
 The skew-symmetry of $\mathcal T$ is obvious.
\end{proof}

We proceed to show that    $\mathcal{T}=0$ if the constraints are holonomic.
\begin{proposition}
The gyroscopic tensor $\mathcal{T}$ vanishes if the constraints are holonomic.
\end{proposition}
\begin{proof}
Let $Y$,  $Z$ be vector fields on $S = Q/G$. If $D$ is integrable, then it is involutive, and hence 
$
[\mbox{ hor } Y, \mbox{ hor }Z] \in \Gamma(D).
$
Thus, 
\[
{\mathcal P}[\mbox{ hor } Y, \mbox{ hor }Z] = [\mbox{ hor } Y, \mbox{ hor }Z].
\]
Moreover, given that the vector fields $\mbox{ hor }Y$ and $\mbox{ hor }Z$ are $\pi$-projectable on $Y$ and $Z$, their
 Lie bracket $[\mbox{ hor } Y, \mbox{ hor }Z]$ is $\pi$-projectable on $[Y, Z]$. Therefore, for any $q \in Q$, we have
\[
(T_q\pi)({\mathcal P}[\mbox{ hor } Y, \mbox{ hor }Z](q)) =  [Y, Z](\pi(q))
\]
which implies that ${\mathcal T} = 0$.
\end{proof}

On the other hand, the vanishing of the gyroscopic tensor $\mathcal T$ does {\em not} imply that the constraints are holonomic. A simple example 
to illustrate this is the motion of a {\em vertical rolling disk}  that rolls without sliding on the plane that we present at the end of this section.
Before doing that, we give local expressions for the gyroscopic tensor.

Let $(s^1, \dots, s^r)$ be local coordinates on $S$. Then $\mathcal{ T} $ is determined by its action on the coordinate vector fields as 
\begin{equation*}
\mathcal{ T} \left ( \frac{\partial }{\partial s^i} \, , \,  \frac{\partial }{\partial s^j}  \right ) =  \sum_{k=1}^r C_{ij}^k (s) \frac{\partial }{\partial s^k},
\end{equation*}
where the coefficients $C_{ij}^k (s)$ are defined by the relations
\begin{equation}
\label{eq:gyrosccoeff}
\mathcal{ P} \left ( \left [ \mbox{hor}_q  \left ( \frac{\partial }{\partial s^i}  \right ) \, , \,   \mbox{hor}_q  \left ( \frac{\partial }{\partial s^j}  \right ) \right ]  \right ) = 
 \sum_{k=1}^r C_{ij}^k (s)  \mbox{hor}_q \left ( \frac{\partial }{\partial s^k} \right ).
\end{equation}
The above relation follows immediately from  Definition~\ref{D:defGyrTensor} since the commutator of the coordinate vector fields vanishes. 
Following~\cite{LGN18}, we refer to   $C_{ij}^k (s)$ as \defn{ the gyroscopic coefficients}. Note that the skew-symmetry of $\mathcal{T}$ implies
that $C_{ij}^k (s)=-C_{ji}^k(s)$.

 We close this section by presenting some of the details
of the calculation that shows that    $\mathcal T =0$ for the vertical rolling disk. In our treatment we follow the notation of \cite{BlochBook}.

\vspace{0.3cm}
\noindent {\bf Example: The vertical rolling disk.} The configuration space for the system is $Q=\R^2\times S^1\times S^1\ni (x,y,\varphi, \theta)$. The coordinates $(x,y)$ and the angle $\varphi$ specify, respectively, the contact point and the orientation
of the disk with respect to an inertial frame $\{e_1, e_2\}$. On the other hand, $\theta$ denotes an internal angle of the disk (see Figure~\ref{F:Disk}).

\begin{figure}[hh]
\centering
\includegraphics[width=10cm]{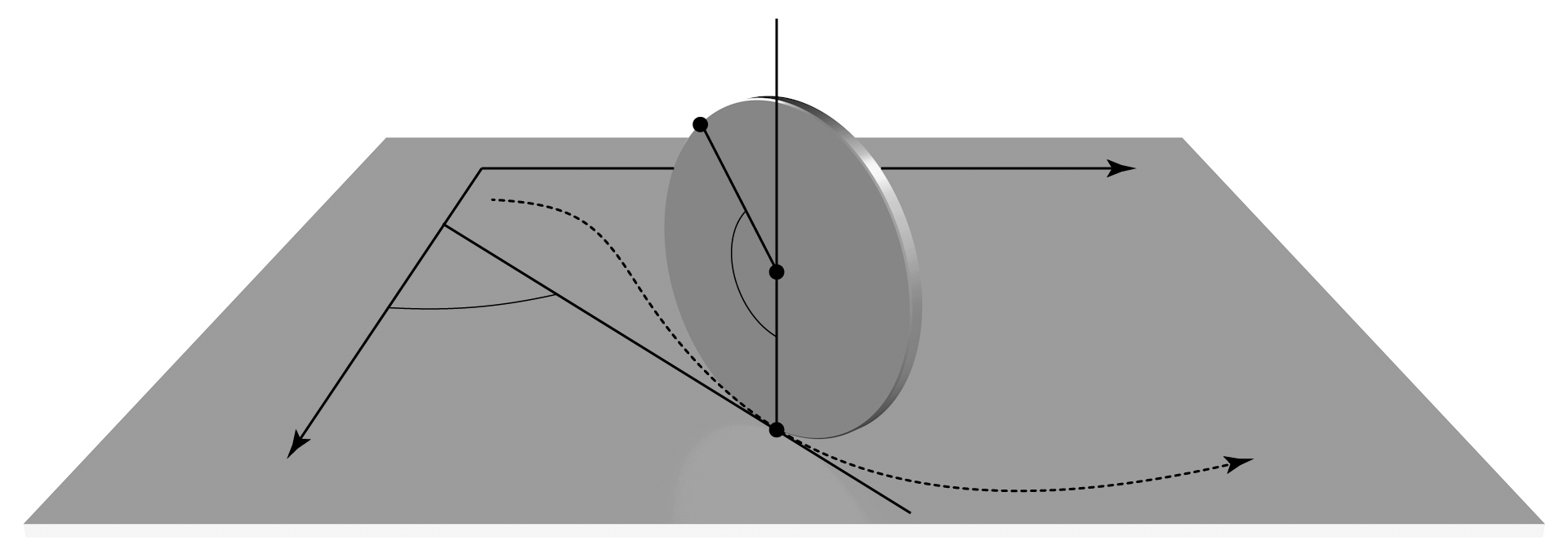}
 \put (-160,48) {$\theta$}  \put (-202,38) {$\varphi$}  \put (-168,15) {$(x,y)$} 
   \put (-242,25) {$e_1$}    \put (-92,63) {$e_2$}
\caption{\small{Vertical rolling disk.}\label{F:Disk} }
\end{figure}
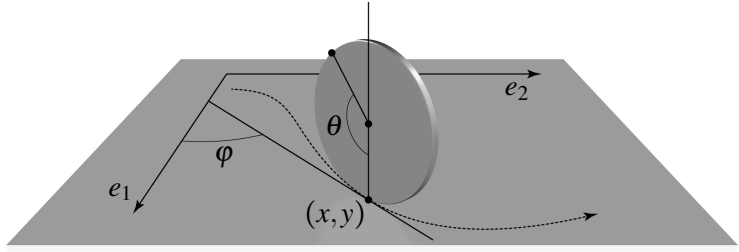

The constraints of rolling without slipping  are
\begin{equation*}
\dot x = R\cos \varphi \dot \theta, \qquad \dot y = R\sin \varphi \dot \theta,
\end{equation*}
where $R$ is the radius of the disk, and hence
\begin{equation}
 \label{eq:Drollingdisk}
D=\mbox{span} \left \{ \partial_\varphi, \partial_\theta + R\cos \varphi \partial_x + R\sin \varphi \partial_y \right \}.
\end{equation}
We assume that the disk is homogeneous so the pure kinetic energy Lagrangian is given by
 \begin{equation}
 \label{eq:Lagrollingdisk}
L=\frac{1}{2} \left ( m(\dot x^2 +\dot y^2)+ I \dot \varphi^2 + J \dot \theta^2  \right ),  
\end{equation}
where $m$ is the mass of the disk and $I$ and $J$ are  the moments of inertia of the disk with respect to the axes that 
pass through the disk's center and
are, respectively,
normal to the plane and normal to the surface of the disk.

The system may be considered as a $G$-Chaplygin system with $G=\R^2$ acting by translations. The shape space is the 2-torus $\mathbb{T}^2$ 
with coordinates $(\varphi, \theta)$ and bundle projection $\pi:Q\to \mathbb{T}^2$ given by $\pi(x,y,\varphi, \theta) =(\varphi, \theta)$. The horizontal lifts of the coordinate vector fields are
\begin{equation*}
\mbox{hor } \partial_\varphi =  \partial_\varphi, \qquad \mbox{hor } \partial_\theta =  \partial_\theta + R\cos \varphi \partial_x + R\sin \varphi \partial_y .
\end{equation*}
Therefore,
\begin{equation*}
[\mbox{ hor }  \partial_\varphi, \mbox{ hor }  \partial_\theta ] = - R\sin \varphi \partial_x + R\cos \varphi \partial_y .
\end{equation*}
It is immediate to check that the above vector field on $Q$ is perpendicular to $D$ given by \eqref{eq:Drollingdisk} with respect to the
Riemannian metric defined by the Lagrangian \eqref{eq:Lagrollingdisk}. It follows that 
$\mathcal{P}[\mbox{ hor }  \partial_\varphi, \mbox{ hor }  \partial_\theta ]=0$
and hence also $\mathcal{T}(  \partial_\varphi,  \partial_\theta)=0$. Therefore $\mathcal{T}$ vanishes identically as claimed.

\subsection{Almost symplectic structure of the reduced dynamics}
\label{SS:Almostsymplectic}

The reduced equations for nonholonomic Chaplygin systems can be formulated in almost symplectic form. Namely, the reduced vector field 
$\overline{X}_{nh}$
describing the reduced dynamics is determined by an equation of the form  ${\bf i}_{\overline{X}_{nh}}\Omega_{nh}=dH$, where $H$ is the 
reduced Hamiltonian
and $\Omega_{nh}$ is a non-degenerate 2-form which is not necessarily closed.  This structure of the equations  seems to have been first noticed  by Stanchenko \cite[Theorem 1]{Stanchenko} in the case of an abelian symmetry
group $G$, and by Cantrijn et al \cite[Equation (17)]{CaCoLeMa}  in the general case. 
This formulation of the equations is useful because the gyroscopic reaction forces that make the system non-Hamiltonian are 
encoded in  the `non-closed' part of $\Omega_{nh}$, and this interpretation 
allows one to give a geometric interpretation of Chaplygin's multiplier method for   Hamiltonisation  (see Section~\ref{Chap-Hamiltonisation} below).

As explained by Ehlers et al in  \cite{EhlersKoiller},  (see also \cite{HochGN}), a construction of the almost symplectic 2-form $\Omega_{nh}$ may be
given utilising the momentum map of the $G$-action and the curvature of the principal connection defined by the constraints. In this section we give an alternative construction
of $\Omega_{nh}$ in terms of the {\em gyroscopic 2-form} $\Omega_\mathcal{T}$, that is defined by~\eqref{eq:defgyrotwoform} below 
using the gyroscopic tensor $\mathcal{T}$
in a way that resembles the definition of the Liouville 1-form on a cotangent bundle. 
The equivalence of the two approaches is proved at the end of Section~\ref{SS:history}.
The main result of this section is Theorem~\ref{T:almostsymplectic}.
 We begin with the following:

\begin{proposition}
The reduced space $\overline{D^*}={D^*}/G$ is naturally identified
with the cotangent bundle $T^*S$ (recall that $S=Q/G$ is the shape space). 
\end{proposition}
\begin{proof}
 There is a vector bundle isomorphism ${\mathcal I}: D/G \to TS$
defined by ${\mathcal I}([v] )= (T_q\pi)(v)$ for  $v \in D_q$.
The inverse morphism is given by ${\mathcal I}^{-1}(u) = [\mbox{hor}_q u]$,
where  $ u \in T_{\pi(q)}S$, and $\mbox{hor}_q: T_{\pi(q)}S \to D_q \subseteq T_qQ$ is the horizontal lift to $D_q$ induced by the principal connection.
The dual isomorphisms, ${\mathcal I}^*: T^*S \to \overline{D^*}$ and $({\mathcal I}^{-1})^*: \overline{D^*} \to T^*S$, define our desired identification
and are given by
\begin{equation}\label{Dual-isomorphisms}
{\mathcal I}^*(\alpha) = [ (T_q^*\pi)(\alpha)_{|D_q}], \qquad ({\mathcal I}^{-1})^*([\beta])(u) = \beta(\mbox{hor}_qu),
\end{equation}
for $\alpha \in T^*_{\pi(q)}S$, $\beta \in D^*_q$ and $u \in T_{\pi(q)}S$.
\end{proof} 

The proposition above allows us to transfer the reduced almost Poisson structure  described by Proposition~\ref{linearity-bracket-red} on $\overline{D^*}$
onto $T^*S$. The resulting bracket on $T^*S$, that will be denoted by $\{\cdot, \cdot\}_{T^*S}$,
  is again linear and the following proposition, whose proof is postponed until the end of this subsection,
   gives its description in terms of the gyroscopic tensor 
$\mathcal{T}$. Note that we continue using the construction outlined in section~\ref{SS:APoissonNonho} for general vector bundles and identify 
  the linear functions on $T^*S$  with vector fields on $S$.

\begin{proposition}\label{bracket-reduced-linear-basic}
Let   $Y^\ell, Z^\ell $ be the linear functions on $T^*S$ corresponding to the vector fields  $Y, \, Z\in  {\frak X}(S)$, and let $f, k$ be basic functions on $T^*S$. Then,
\begin{equation*}
\begin{split}
\{ Y^\ell, Z^\ell \}_{T^*S} = - [Y, Z]^\ell - {\mathcal T}(Y, Z)^\ell, \qquad
 \{  f , Y^\ell \}_{T^*S} = Y[ \tilde f] \circ \tau_S, \qquad  \{  f , k \}_{T^*S}=0,
\end{split}
\end{equation*}
where $f= \tilde f\circ \tau_S$ and $\tau_S: T^*S \to S$ is the canonical projection.
\end{proposition}
Let $(s^1, \dots, s^r)$ be local coordinates on  $S$ and let $(s^1, \dots, s^r, p_1, \dots p_r)$ be the induced bundle coordinates on
$T^*S$ (i.e. an element $\alpha\in T^*S$ is written as $\alpha=\sum p_ids^i$). We have $\left ( \frac{\partial}{\partial s^j} \right )^\ell=p_j$ and therefore the above proposition implies that the  
almost Poisson bracket $\{\cdot, \cdot\}_{T^*S}$ is given locally by 
\begin{equation}
\label{eq:redPBlocal}
\{ p_i , p_j\}_{T^*S} = -\sum_{k=1}^rC_{ij}^k(s)p_k, \qquad
 \{  s^i , p_j \}_{T^*S} = \delta_j^i, \qquad  \{  s^i , s^j \}_{T^*S}=0,
\end{equation}
where $ \delta_j^i$ is the Kronecker delta and $C_{ij}^k(s)$ are the gyroscopic coefficients determined by \eqref{eq:gyrosccoeff}.

Denote by $\Lambda_{T^*S}$ the bivector on $T^*S$ determined by the almost Poisson bracket $\{\cdot, \cdot\}_{T^*S}$, that is,
\[
\Lambda_{T^*S}(d\varphi, d\mu) = \{\varphi, \mu\}_{T^*S},
\]
for $\varphi, \mu \in C^{\infty}(T^*S)$, and let $\Lambda_{T^*S}^\sharp$  be  vector bundle morphism $\Lambda_{T^*S}^\sharp: T^*(T^*S) \to T(T^*S)$
defined by $\Lambda_{T^*S}^\sharp( \beta)=  \Lambda_{T^*S}(\cdot , \beta)$. 
Equations \eqref{eq:redPBlocal} imply that $\Lambda_{T^*S}^\sharp$ has block matrix representation
\begin{equation*}
\Lambda_{T^*S}^\sharp = \begin{pmatrix} 0 & I \\ -I & -\mathcal{C} \end{pmatrix},
\end{equation*}
with respect to the bases  $\{ ds^i, dp_i \}$ of $T_{(s,p)}^*(T^*S)$ and $\{ \frac{\partial}{\partial s^j}, \frac{\partial}{\partial p_j}\}$ of $T_{(s,p)}(T^*S)$.
Here $I$ denotes the $r\times r$ identity matrix and $\mathcal{C}$ the $r\times r$ matrix with entries $\mathcal{C}_{ij}=\sum_{k=1}^rC_{ij}^k(s)p_k$.  
It is clear that the above matrix for $\Lambda_{T^*S}^\sharp$  is invertible. 
As a consequence, there is a unique non-degenerate 2-form $\Omega_{nh}$ on $T^*S$ whose induced bundle 
morphism  $\Omega_{nh}^\flat: T(T^*S) \to T^*(T^*S)$, defined by $\Omega_{nh}^\flat: U \mapsto \Omega_{nh}(U,\cdot)= {\bf i}_{U}\Omega_{nh}$, 
is the inverse of $\Lambda_{T^*S}^\sharp$. Its  matrix representation is
\begin{equation}
\label{eq:om-flat-matrix}
\Omega_{nh}^\flat=(\Lambda_{T^*S}^\sharp )^{-1}=\begin{pmatrix} - \mathcal{C} & -I \\ I & 0 \end{pmatrix},
\end{equation}
and so
\begin{equation*}
 \Omega_{nh}^\flat\left ( \frac{\partial}{\partial s^i}\right ) = d p_i + \sum_{j,k=1}^rC_{ij}^k(s)p_k \, ds^j, \qquad  
 \Omega_{nh}^\flat\left ( \frac{\partial}{\partial p_i}\right ) =-ds^i.
\end{equation*}
Therefore, $\Omega_{nh}$ is locally given by
\begin{equation}
\label{eq:Omnh-local}
\Omega_{nh} = \sum_{k=1}^n ds^i\wedge dp_i + \sum_{i<j}\sum_{k=1}^rC_{ij}^k(s)p_k\, ds^i\wedge ds^j.
\end{equation}
In order to give an intrinsic definition of $\Omega_{nh}$, we start by defining the \defn{gyroscopic  2-form}  $\Omega_{\mathcal T}$ on $T^*S$ as 
follows:
\begin{equation}
\label{eq:defgyrotwoform}
\Omega_{\mathcal T}(\alpha) (U, V) := \alpha \left({\mathcal T}((T_\alpha \tau_S)(U), (T_\alpha \tau_S)(V))\right),
\end{equation}
for $\alpha \in T^*S$ and $U, V \in T_\alpha(T^*S)$, with $\tau_S: T^*S \to S$ the canonical projection. It is straightforward to check that
$\Omega_{\mathcal T}$ is semi-basic and that it has the following local expression in bundle coordinates
\begin{equation}\label{Omega-T-local}
\Omega_{\mathcal T}= \sum_{i<j}\sum_{k=1}^rC_{ij}^k(s)p_k\, ds^i\wedge ds^j.
\end{equation}
Let $\Omega_{can}$ be the canonical symplectic form\footnote{our sign  convention 
is such that locally $\Omega_{can}= \sum_{i=1}^n ds^i\wedge dp_i$.}  on $T^*S$. We define  $\Omega_{nh}$ 
 intrinsically by:
\begin{equation}
\label{eq:defOmnhIntrinsic}
\Omega_{nh}:= \Omega_{can}+\Omega_{\mathcal T},
\end{equation}
so that the local expression~\eqref{eq:Omnh-local} holds.

 We will now formulate the main result of this section. In order to keep the notation simple, we also denote by $H$ and $\overline{X}_{nh}$ 
 the respective pull-backs to $T^*S$ of
 the reduced Hamiltonian $H\in C^\infty(\overline{D^*})$ and  the reduced vector field $\overline{X}_{nh}\in \frak{X}(\overline{D^*})$
  by the isomorphism ${\mathcal I}^*: T^*S \to \overline{D^*}$.
 
 \begin{theorem}
 \label{T:almostsymplectic}
 The 2-form $\Omega_{nh}$ defined by~\eqref{eq:defOmnhIntrinsic} is non-degenerate and characterises the reduced vector field $\overline{X}_{nh}$
 on $T^*S$ uniquely by the relation
  \begin{equation}
  \label{eq:almost-symplectic}
  {\bf i}_{\overline{X}_{nh}}\Omega_{nh}=dH,
\end{equation}
where $H\in C^\infty(T^*S)$ is the reduced Hamiltonian. Moreover, $\Omega_{nh}$ is closed (and hence symplectic) if and
only if the gyroscopic tensor $\mathcal{T}$ vanishes.
 \end{theorem}
 \begin{proof}
 The non-degeneracy of $\Omega_{nh}$ follows from the local expressions given above. In particular from the matrix representation
 \eqref{eq:om-flat-matrix} for $\Omega_{nh}^\flat$. Now, Proposition~\ref{eq:AlmostPoisson-redSpace} together with our identification
 of $\overline{D^*}$ with $T^*S$ implies that $\Lambda_{T^*S}^\sharp(dH)=\overline{X}_{nh}$ and therefore $\Omega_{nh}^\flat(\overline{X}_{nh})=dH$
 which is equivalent to \eqref{eq:almost-symplectic}.
 
 Finally, note that since $\Omega_{can}$ is closed, then $d\Omega_{nh}=d\Omega_\mathcal{T}$. 
Hence, if $\mathcal{T}=0$ then 
  $d\Omega_{nh}=0$. 
   Conversely, suppose that  $d\Omega_{nh}=d\Omega_\mathcal{T}=0$ and let $\gamma\in \Omega^1(S)$. Considering that the vertical lift $\gamma^{\bf v}\in \mathfrak{X}(T^*S)$ (defined by \eqref{vertical-lift-dual}) is
   vertical and $\Omega_\mathcal{T}$ is semi-basic we have  $ {\bf i}_{\gamma^{\bf v}}\Omega_\mathcal{T}=0$ and therefore
\begin{equation*}
0= {\bf i}_{\gamma^{\bf v}}d\Omega_\mathcal{T} = \mathcal{L}_{\gamma^{\bf v}}\Omega_\mathcal{T} = \tau_S^*(\mathcal{T}_\gamma),
\end{equation*}
  where $\mathcal{T}_\gamma$ is the 2-form on $S$ given by $\mathcal{T}_\gamma (s)(u,v) = \langle \gamma (s) , \mathcal{T}(u,v) \rangle$ for $s\in S$ and
$u,v \in T_sS$. Considering that $\tau_S^*$ is injective, the above equation implies that $\mathcal{T}_\gamma =0$ for any 1-form $\gamma$. Hence $\mathcal{T}=0$.

 \end{proof}
 
 Taking into account~\eqref{eq:almost-symplectic}, and the local expression~\eqref{eq:Omnh-local} of $\Omega_{nh}$, leads to the following
 local expressions that determine  the reduced vector field $\overline{X}_{nh}$:
 \begin{equation}
 \label{eq:motionLocal}
\dot s^i = \frac{\partial H}{\partial p_i}, \qquad \dot p_i =- \frac{\partial H}{\partial s^i} -\sum_{j,k=1}^rC_{ij}^kp_k\frac{\partial H}{\partial p_j}, \qquad i=1,\dots, r.
\end{equation}
The above equations differ from the standard Hamilton equations by the presence of the terms proportional to the gyroscopic coefficients
$C_{ij}^k$.
 These terms correspond to gyroscopic forces that take the system outside of the Hamiltonian realm since, in accordance to the above theorem,
 $\Omega_{nh}$ is in general not symplectic.
 
 We note that the  reduced Hamiltonian $H \in C^\infty(T^*S)$ is given in bundle coordinates by
 \begin{equation}
 \label{eq:HamChap-local}
H(s^i,p_i)= \frac{1}{2} \sum_{i,j=1}^k K^{ij}(s)p_ip_j +\overline{U}(s),
\end{equation}
where $\overline{U}\in C^\infty(S)$ is the reduced potential energy induced by the $G$-invariant potential $U\in C^\infty(Q)$
and  $K^{ij}(s)$ are the entries of the inverse matrix of the positive definite matrix with entries
\begin{equation}
\label{eq:reduced-metric-coords}
K_{ij}(s) = \left  \llangle \mbox{ hor } \frac{\partial}{\partial s^i} \, ,  \, \mbox{ hor } \frac{\partial}{\partial s^j}  \right \rrangle ,
\end{equation}
where we recall that  $\llangle \cdot \, ,  \, \cdot   \rrangle$ is the kinetic energy metric on $Q$. One can easily verify
that $K_{ij}$ are well defined functions on the coordinate chart of $S$ by using the $G$-invariance of the kinetic energy
and of the horizontal lift.   In fact, $(K_{ij})$ is the matrix of the coefficients of the Riemannian metric 
$\llangle \cdot \, ,  \, \cdot   \rrangle^{-}$
on $S$ characterised by
\begin{equation}
\label{reduced-metric}
\llangle X, Y \rrangle^{-} \circ \pi_s = \llangle \mbox{ hor }X, \mbox{ hor }Y \rrangle, \; \; \mbox{ for } X, Y \in \frak{X}(S).
\end{equation}
We now use this metric to construct   the tensor field ${\mathcal B}$ of type $(2, 1)$ on $S$ by 
{\em raising an index} of $\mathcal T$. Namely, we define
\begin{equation}\label{T-bar}
\langle {\mathcal B}(\alpha, \beta), X \rangle = \langle \alpha, {\mathcal T}(\beta^{\sharp}, X) \rangle, \; \; \mbox{ for } \alpha, \beta \in \Omega^1(S) \mbox{ and } X \in {\frak X}(S),
\end{equation}
where   $\beta^{\sharp} \in {\frak X}(S)$ denotes the metric dual  
of the $1$-form $\beta\in \Omega^1(S)$ (see \eqref{eq:metric-dual}).
Next, we define   the semi-basic $1$-form $\eta_{\mathcal T}$  on $T^*S$ by
\begin{equation}\label{def-mu-T}
\eta_{\mathcal T}(\alpha)(X) = \langle {\mathcal B}(\alpha, \alpha), (T_\alpha\tau_S)(X)\rangle, \; \; \mbox{ for } \alpha \in T^*S \mbox{ and } X\in T_\alpha(T^*S).
\end{equation}
The 1-form $\eta_{\mathcal T}$ encodes the gyroscopic forces that deviate the vector field $\overline{X}_{nh}$ from being 
Hamiltonian in the manner
that is made precise in  the following proposition that will be useful ahead.
\begin{proposition}\label{Contraction-can-sym}
We have 
\[
{\bf i}_{\overline{X}_{nh}}\Omega_{can} = dH + \eta_{\mathcal T},
\]
where $\eta_{\mathcal T}$ is the $1$-form on $T^*S$ defined by~\eqref{def-mu-T}.
\end{proposition}
\begin{proof}
In view of the almost symplectic formulation ${\bf i}_{\overline{X}_{nh}}\Omega_{nh}=dH$, and the definition 
$\Omega_{nh}=\Omega_{can}+\Omega_{\mathcal T}$,  it suffices to show that 
$\eta_{\mathcal T} = -{\bf i}_{\overline{X}_{nh}}\Omega_{\mathcal T}$.
Locally we have 
\[
\mathcal B = \sum_{i,j,k,l=1}^r C_{ij}^k K^{jl} \left( \frac{\partial}{\partial s^k} \otimes \frac{\partial}{\partial s^l}\right) \otimes ds^{i}.
\]
This implies
that $\eta_{\mathcal T}$ admits the local expression
\begin{equation}\label{mu-T-local}
\eta_{\mathcal T} =  \sum_{i,j,k,l=1}^r  C_{ij}^k K^{jl}p_k p_l \, ds^{i}.
\end{equation}
A direct calculation that uses~\eqref{Omega-T-local},~\eqref{eq:motionLocal} and \eqref{eq:HamChap-local}, shows that the right hand side of this equation coincides with the
local expression for $-{\bf i}_{\overline{X}_{nh}}\Omega_{\mathcal T}$.
\end{proof}

 We finish this section by presenting the following:
 \begin{proof}[Proof of Proposition~\ref{bracket-reduced-linear-basic}]
The  almost-Poisson bracket $\{\cdot, \cdot\}_{T^*S}$ on $T^*S$ induced by the almost-Poisson bracket $\{\cdot, \cdot\}_{\overline{D^*}}$ on $\overline{D^*}$ is given by
\begin{equation}\label{reduced-non-holo-bracket-Cha}
\{\varphi, \mu\}_{T^*S} = \{\varphi \circ ({\mathcal I}^{-1})^*, \mu \circ ({\mathcal I}^{-1})^*\}_{\overline{D^*}} \circ {\mathcal I}^*,
\end{equation}
for $\varphi, \mu \in C^\infty(T^*S)$. On the other hand,  if $Y \in  {\frak X}(S)$ and $\tilde{f} \in C^{\infty}(S)$, then~\eqref{Dual-isomorphisms} implies that
\begin{equation}\label{Change-isomorphism}
Y^\ell \circ ({\mathcal I}^{-1})^* = (\mbox{ hor } Y)^\ell, \qquad (\tilde{f} \circ \tau_S) \circ ({\mathcal I}^{-1})^* = \tilde{f} \circ \overline{\tau},
\end{equation}
where $\mbox{ hor } X$ is interpreted as a section of $D/G$ in virtue of its equivariance and $\overline{\tau}: \overline{D^*} \to S$ denotes the vector bundle projection. Therefore, in view of~\eqref{reduced-non-holo-bracket-Cha},~\eqref{Change-isomorphism} and Proposition \ref{linearity-bracket-red}, 
for $\alpha \in T_{\pi(q)}^*S$ we have
\begin{equation*}
\begin{split}
\{ Y^\ell, Z^\ell \}_{T^*S}(\alpha)& = \{  (\mbox{ hor } Y)^\ell,  (\mbox{ hor } Z)^\ell \}_{\overline{D^*}} ( {\mathcal I}^*(\alpha)) \\
&= - (\mathcal{P} [\mbox{ hor } Y \, ,  \, \mbox{ hor } Z])^\ell( {\mathcal I}^*(\alpha)) \\
&= - \langle  {\mathcal I}^*(\alpha) \,  ,  \, \mathcal{P} [\mbox{ hor } Y \, , \, \mbox{ hor } Z] \rangle \\
&=  - \langle T^*_q\pi (\alpha) \,  ,  \, \mathcal{P} [ \mbox{ hor } Y \, , \, \mbox{ hor } Z] \rangle \\
&=  - \langle \alpha \,  ,  \, T_q\pi \left  ( \mathcal{P} [\mbox{ hor }Y \, , \, \mbox{ hor } Z]  \right ) \rangle \\
&=  - \langle \alpha \,  ,  \,  [Y,Z] + \mathcal{T}(Y,Z)   \rangle \\
&=  \left ( - [Y,Z]^\ell - \mathcal{T}(Y,Z)^\ell \right ) (\alpha),
\end{split}
\end{equation*}
where we have used~\eqref{Dual-isomorphisms} to give an expression for   $ {\mathcal I}^*(\alpha)$ in the third equality. In a similar manner,
but even simpler, we have
\begin{equation*}
\begin{split}
\{ f, Y^\ell \}_{T^*S}& = \{ \tilde{f} \circ \overline{\tau} ,  (\mbox{ hor } Y)^\ell \}_{\overline{D^*}} \circ  {\mathcal I}^* \\
&= Y[\tilde f ]  \circ \overline{\tau}  \circ  {\mathcal I}^*  \\
&= Y[\tilde f ]  \circ \tau_S,
\end{split}
\end{equation*}
and
\begin{equation*}
\begin{split}
\{ f, k \}_{T^*S}& = \{ \tilde{f} \circ \overline{\tau} ,  \tilde{k} \circ \overline{\tau} \}_{\overline{D^*}} \circ  {\mathcal I}^* =0.
\end{split}
\end{equation*}
\end{proof}

\subsection{Existence of a smooth invariant measure }
\label{SS:measures}

One of the most important invariants that a nonholonomic system may have is a smooth volume form. For Chaplygin systems without potential
forces,  a necessary
and sufficient condition for its existence is that a certain 
1-form on $S$, that we will denote by $\Theta$, is exact (see Cantrijn et al \cite[Theorem 7.5]{CaCoLeMa}, and also \cite[Corollary 4.5]{FedGNMa2015}).
 The 1-form $\Theta$ is naturally constructed as the ordinary contraction of the gyroscopic tensor $\mathcal{T}$:
\begin{equation}
\label{eq:1-form-theta-def}
\Theta(Y):= \sum_{j=1}^r \langle X^j \, , \, \mathcal{T}(X_j,Y) \rangle ,
\end{equation}
where   $\{ X_1, \dots X_r\}$ is a local basis of vector fields of $S$,   $\{ X^1, \dots X^r\}$ is the dual basis, and  $\langle \cdot \, ,  \, \cdot \rangle$ 
denotes the pairing of covectors and vectors on $S$. 
It is clear that $\Theta$ is well-defined (globally and independently of the basis). A local expression for $\Theta$ may be obtained by taking $\{\frac{\partial}{\partial s^j}\}$ as the basis of vector fields $\{X_j\}$ in its definition.
In view of the definition of the gyroscopic coefficients, we get $\Theta \left (\frac{\partial}{\partial s^i} \right ) = \sum_{j=1}^rC_{ji}^j(s)$, and,  therefore,
 $\Theta$ is locally given by
\begin{equation}
\label{eq:Theta-local} 
\Theta=  \sum_{i,j=1}^rC_{ji}^j(s)\,ds^i.
\end{equation}

Recall that   the cotangent bundle $T^*S$  is equipped with  the  \defn{Liouville volume form} $\nu$ defined as  $\nu:=\Omega_{can}^r$.
\begin{definition}
 A volume form $\mu$ on $T^*S$ is \defn{basic} if its density with respect to the Liouville volume form $\nu$ is a basic function. Namely if
 \begin{equation*}
\mu = (f\circ \tau_S) \, \nu,
\end{equation*}
for a positive function $f\in C^\infty(S)$.
\end{definition}

The relationship between $\Theta$ and the existence of an invariant measure is given in the following theorem.
\begin{theorem}[Cantrijn et al. \cite{CaCoLeMa}]
\label{T:measures-1-form}
Let $\nu = \Omega_{can}^r$ be  the Liouville volume form on $T^*S$. 
\begin{enumerate}
\item For a Hamiltonian $H=K+U$, the  reduced  equations of motion of a nonholonomic Chaplygin system preserve the basic measure
\begin{equation}
\label{eq:measure}
\mu = \exp (\sigma \circ \tau_S) \, \nu, \qquad \sigma\in C^\infty(S),
\end{equation}
 if and only if $\Theta$ is exact with $\Theta =d\sigma$.
\item In the absence of potential energy, the reduced equations posses a smooth invariant measure if and only if it is  basic (which is  then
characterised by item (i)).
\end{enumerate}
\end{theorem}

\begin{remark} In section \ref{SS:NonhoParticle} below we give an example of a Chaplygin system with non-trivial potential possessing a smooth invariant measure 
 that is not basic. Such example shows that the conclusion of item (ii) may not be extended to systems with potential energy.
\end{remark}

Theorem \ref{T:measures-1-form} was proved in \cite{CaCoLeMa} (see also \cite{FedGNMa2015}) with an alternative
definition of the 1-form $\Theta$. Here we  present an alternative intrinsic proof which is based on the following lemma. In its statement, recall that 
 $\Theta^\sharp \in \mathfrak{X}(S)$ denotes the metric dual of  $\Theta\in \Omega^1(S)$ (see~\eqref{eq:metric-dual}), and
 $(\Theta^\sharp)^\ell\in C^\infty( T^*S)$ is the associated linear function (see~\eqref{eq:linear-function}).
 
\begin{lemma}
\label{l:For-Measure-Conservation}
Let $\mu$ be a general (not necessarily basic) volume form on $T^*S$ given by
$\mu = \exp{(\overline{\sigma})}\nu$, with $\overline{\sigma} \in C^{\infty}(T^*S)$ and where as usual $\nu$ is the Liouville 
measure on $T^*S$.  
Denote by ${\mathcal L}_{\overline{X}_{nh}}$ the Lie derivative operator with respect to $\overline{X}_{nh}$. 
Then
\begin{equation}\label{action-dynamics-volume-form-lemma}
{\mathcal L}_{\overline{X}_{nh}}(\mu) =  (\overline{X}_{nh}[\overline{\sigma}] - (\Theta^\sharp)^\ell) \, \mu.
\end{equation}
\end{lemma}

The proof of this lemma is postponed to the end of the section and we proceed to give the proof of Theorem~\ref{T:measures-1-form}.
\begin{proof}[Proof of Theorem~\ref{T:measures-1-form}]

To prove item $(i)$ consider the basic measure $\mu = \exp( \sigma \circ \tau_S) \nu$  with $\sigma \in C^{\infty}(S)$. 
Using the local expressions  \eqref{eq:motionLocal} and (\ref{eq:HamChap-local}) for $\overline{X}_{nh}$ one shows that 
$\overline{X}_{nh}[ \sigma \circ \tau_S]= (d\sigma^\sharp)^\ell$.  Hence, in this case,
\eqref{action-dynamics-volume-form-lemma} may be written as
\begin{equation*}
{\mathcal L}_{\overline{X}_{nh}}(\mu) = ((d\sigma^\sharp)^\ell - (\Theta^\sharp)^\ell) \, \mu,
\end{equation*}
which implies that ${\mathcal L}_{\overline{X}_{nh}}(\mu)=0$ if and only if $(d\sigma^\sharp)^\ell - (\Theta^\sharp)^\ell=0$ or, equivalently,
$d\sigma=\Theta$.

Next, we will prove item $(ii)$ of the theorem. Suppose that the volume form $\mu = \exp(\overline{\sigma}) \nu$ is invariant under the action of $\overline{X}_{nh}$. Then \eqref{action-dynamics-volume-form-lemma} implies that
$\overline{X}_{nh}[\overline{\sigma}] = (\Theta^\sharp)^\ell$ and therefore, for any 
$\gamma \in \Omega^1(S)$, we have
\begin{equation}\label{gamma-theta-bar-sigma}
(\gamma^{\bf v}[\overline{X}_{nh}[\overline{\sigma}]]) \circ 0 = (\gamma^{\bf v}[(\Theta^\sharp)^\ell]) \circ 0 = \langle \gamma, \Theta^\sharp \rangle \circ \tau_S,
\end{equation}
 where $0: S \to T^*S$ is the zero section of the vector bundle $\tau_S: T^*S \to S$.

On the other hand, using the local equations \eqref{eq:motionLocal} and \eqref{eq:HamChap-local} which determine to the vector field $\overline{X}_{nh}$ and the local expression (\ref{vertical-complete-lift-dual}) of the vector field $\gamma^{\bf v}$, we conclude that
\[
(\gamma^{\bf v}[\overline{X}_{nh}[\overline{\sigma}]]) \circ 0 = \langle \gamma, (d\sigma)^\sharp \rangle \circ \tau_S, \; \; \; \forall \gamma \in \Omega^1(S)
\]
with $\sigma = \overline{\sigma} \circ 0 \in C^{\infty}(S)$. So, from (\ref{gamma-theta-bar-sigma}), we deduce that
\[
\langle \gamma, (d\sigma)^\sharp \rangle \circ \tau_S = \langle \gamma, \Theta^\sharp \rangle \circ \tau_S, \; \; \forall \gamma \in \Omega^1(S)
\]
which implies that $d\sigma = \Theta$. Thus, using item $(i)$, it follows that the basic volume form $\exp({\sigma}\circ \tau_S)\nu$ is invariant under the action of $\overline{X}_{nh}$.
\end{proof}

We now present the following: 
\begin{corollary}[Stanchenko~\cite{Stanchenko}]
\label{C:Noetherianity-measure}
The existence  of a basic invariant measure for a 
 $G$-Chaplygin nonholonomic system is  weakly Noetherian. Namely,
if a   $G$-Chaplygin nonholonomic system preserves a basic measure, then it  continues to 
preserve the same basic measure under the addition of a
$G$-invariant potential.
\end{corollary}
\begin{proof}
The gyroscopic tensor $\mathcal{T}$ only depends on the kinetic energy and not on the potential. Thus, the same is true
for the 1-form $\Theta$. In  virtue of item (i) of Theorem~\ref{T:measures-1-form}, the preservation of a basic measure is 
equivalent to the exactness of $\Theta$,  which holds independently of the potential.
\end{proof}

We finish the section with the proof of Lemma~\ref{l:For-Measure-Conservation}.
\begin{proof}
 Using the basic
properties of the Lie derivative we have
\begin{equation}
\label{eq:aux-proof-volume-intrinsic}
{\mathcal L}_{\overline{X}_{nh}}(\mu)= {\mathcal L}_{\overline{X}_{nh}}(\exp(\overline{\sigma}) \nu) = \exp(\overline{\sigma})(\overline{X}_{nh}[\overline{\sigma}]\nu 
+ r\, ({\mathcal L}_{\overline{X}_{nh}}\Omega_{can}) \wedge \Omega_{can}^{r-1}).
\end{equation}
On the other hand, using Cartan's magic formula  and the fact that $\Omega_{can}$ is closed, we get
\[
{\mathcal L}_{\overline{X}_{nh}}\Omega_{can} = {\bf i}_{\overline{X}_{nh}}(d\Omega_{can}) + d({\bf i}_{\overline{X}_{nh}}\Omega_{can}) = d({\bf i}_{\overline{X}_{nh}}\Omega_{can}),
\]
which, in view of Proposition~\ref{Contraction-can-sym}  allows us to
write~\eqref{eq:aux-proof-volume-intrinsic} as
\begin{equation}\label{action-dynamics-volume-form}
{\mathcal L}_{\overline{X}_{nh}}(\mu) = \exp(\overline{\sigma})(\overline{X}_{nh}[\overline{\sigma}]\nu + r\, d\eta_{\mathcal T} \wedge \Omega_{can}^{r-1}).
\end{equation}
However we claim that 
\begin{equation}
\label{eq:key-id-lemma-measure}
  r \, d\eta_{\mathcal T} \wedge \Omega_{can}^{r-1}=-(\Theta^{\sharp})^\ell \, \nu.
\end{equation}
Note that substitution of~\eqref{eq:key-id-lemma-measure} into~\eqref{action-dynamics-volume-form}  proves the lemma, so it only remains to prove that \eqref{eq:key-id-lemma-measure} holds.

Let  $\{X_i\}$  be  a local basis of $\frak{X}(S)$ and $\{X^{i}\}$  be the  dual basis of $\Omega^1(S)$.
To prove \eqref{eq:key-id-lemma-measure} we will show that the following two identities hold:
\begin{equation}
\label{eq:lemma-intrinsic}
(\Theta^{\sharp})^\ell = -\sum_{i=1}^rd\eta_{\mathcal T}\left(X_i^{*c}, (X^{i})^{\bf v}\right), \qquad  \qquad r d\eta_{\mathcal T} \wedge \Omega_{can}^{r-1} = \sum_{i=1}^r d\eta_{\mathcal T}(X_i^{*c}, (X^{i})^{\bf v}) \nu,
\end{equation}
where the vertical lifts $(X^{i})^{\bf v}$ and the complete lifts $X_i^{*c}$ are respectively defined in \eqref{vertical-lift-dual}
and \eqref{complete-lift-dual}.

Starting from the local expression~\eqref{mu-T-local} of $\eta_{\mathcal T}$ and using (\ref{vertical-complete-lift-dual}), we have
\begin{equation}\label{eta-T-vertical-complete}
\eta_{\mathcal T}(\alpha^{\bf v}) = 0, \qquad \eta_{\mathcal T}(Y^{*c}) = \sum_{i,j,k,l=1}^rY^{i}C_{ij}^kK^{jl}p_kp_l,
\end{equation}
for $\alpha \in \Omega^{1}(S)$ and $Y \in \frak{X}(S)$.
So, if $({\mathcal B}, Y)$ is the tensor of type $(2, 0)$ on $S$ given by
\begin{equation}\label{B-cal-Y}
({\mathcal B}, Y)(\alpha, \beta) = \langle \mathcal{B}(\alpha, \beta), Y\rangle =\langle \alpha, {\mathcal T}(\beta^{\sharp}, Y)\rangle, \mbox{ for } \alpha, \beta \in \Omega^1(S),
\end{equation}
it follows that the quadratic function $({\mathcal B}, Y)^{\mathfrak{q}}$ on $T^*S$ (see \eqref{eq:quad-function}) is just $-\eta_{\mathcal T}(Y^{*c})$.
Thus, using (\ref{Lie-bracket-complete-vertical}) and (\ref{eta-T-vertical-complete}), we deduce that
\[
d\eta_{\mathcal T}(\alpha^{\bf v}, Y^{*c}) = \alpha^{\bf v}(\eta_{\mathcal T}(Y^{*c})) - Y^{*c}(\eta_{\mathcal T}(\alpha^{\bf v})) - \eta_{\mathcal T}[\alpha^{\bf v}, Y^{*c}] =-\alpha^{\bf v}(({\mathcal B}, Y)^{\mathfrak{q}}).
\]
This implies that
\[
d\eta_{\mathcal T}(\alpha^{\bf v}, Y^{*c}) =\sum_{i,j,k,l=1}^r Y^{i}\alpha_{l}(C_{ij}^{k}K^{jl} + C_{ij}^{l}K^{jk})p_k
\]
or, equivalently,
\[
d\eta_{\mathcal T}(\alpha^{\bf v}, Y^{*c}) = -({\mathcal T}(\alpha^{\sharp}, Y))^\ell - (({\mathcal B}, Y)(\alpha, \cdot))^\ell,
\]
where $({\mathcal B}, Y)(\alpha, \cdot)$ is the vector field on $S$ which is characterized by
\[
\langle \beta, ({\mathcal B}, Y)(\alpha, \cdot) \rangle = ({\mathcal B}, Y)(\alpha, \beta), \; \; \mbox{ for } \beta \in \Omega^{1}(S).
\]
Therefore, 
\[
\sum_{i=1}^r d\eta_{\mathcal T}\left( X_i^{*c}, (X^{i})^{\bf v}\right) = \sum_{i=1}^{r} ({\mathcal T}((X^{i})^{\sharp}, X_{i}))^\ell + \sum_{i=1}^{r}(({\mathcal B}, X_i)(X^{i}, \cdot))^\ell.
\]
Now, since the gyroscopic tensor $\mathcal T$ is skew-symmetric, ${\mathcal T}((X^{i})^{\sharp}, X_{i}) = 0$, for all $i$. In fact, if our local basis $\{X_i\}$ is orthonormal then $(X^{i})^{\sharp} = X_i$, for all $i$. Consequently,
\begin{equation}\label{traza-d-eta}
\sum_{i=1}^r d\eta_{\mathcal T}\left( X_i^{*c}, (X^{i})^{\bf v}\right) =  \sum_{i=1}^{r}(({\mathcal B}, X_i)(X^{i}, \cdot))^\ell.
\end{equation}

On the other hand, in view of the local expression~\eqref{eq:Theta-local} of $\Theta$ and using (\ref{B-cal-Y}), we have
\[
\Theta(\alpha^{\sharp}) = - \langle \alpha, \sum_{i=1}^r ({\mathcal B}, X_i)(X^{i}, \cdot) \rangle, \; \; \mbox{ for } \alpha \in \Omega^1(S).
\]
This implies that
\[
\Theta^{\sharp} = -\sum_{i=1}^r ({\mathcal B}, X_i)(X^{i}, \cdot)
\]
and, by (\ref{traza-d-eta}), it follows that 
\[
(\Theta^{\sharp})^\ell = -\sum_{i=1}^r d\eta_{\mathcal T}\left( X_i^{*c}, (X^{i})^{\bf v}\right),
\]
which proves that the first identity in~\eqref{eq:lemma-intrinsic} holds.

Now,~\eqref{vertical-complete-lift-dual} shows that $\{X_i^{*c}, (X^{i})^{\bf v}\}$ is a local basis of $\frak{X}(T^*S)$ and, moreover, 
\begin{equation}\label{Omega-complete-vertical}
\Omega_{can}(Y^{*c}, \alpha^{\bf{v}}) = \alpha(Y) \circ \tau_S, \; \; \mbox{ for } Y\in \frak{X}(S) \mbox{ and } \alpha \in \Omega^1(S). 
\end{equation}
Thus, we deduce that
\begin{equation*}
\begin{split}
(r\, d\eta_{\mathcal T} \wedge \Omega_{can}^{r-1}) & (X_1^{*c}, (X^1)^{\bf v}, \dots, X_r^{*c}, (X^r)^{\bf v})  = \\
&= r \sum_{i=1}^r d\eta_{\mathcal T}(X_i^{*c}, (X^{i})^{\bf v})\Omega_{can}^{r-1}(X_1^{*c}, (X^1)^{\bf v}, \dots, \widehat{X_i^{*c}}, \widehat{(X^{i})^{\bf v}}, \dots, X_r^{*c}, (X^r)^{\bf v})  \\
& =  r! \sum_{i=1}^r d\eta_{\mathcal T}(X_i^{*c}, (X^{i})^{\bf v}).
 \end{split}
\end{equation*}
On the other hand, using (\ref{Omega-complete-vertical}), it also follows that
\[
\nu(X_1^{*c}, (X^1)^{\bf v}, \dots, X_r^{*c}, (X^r)^{\bf v}) = r!
\]
which  shows that the second equation in~\eqref{eq:lemma-intrinsic} also holds.
\end{proof}

\subsection{Chaplygin Hamiltonisation}
Within the class of Chaplygin systems with an invariant measure there is a special subclass whose equations of motion
may  be written in Hamiltonian 
form after a  time reparametrisation $dt =g(s)\, d\tau$, for a positive function $g\in C^\infty(S)$. This observation 
 goes back to Chaplygin who introduced his
  reducing multiplier method  \cite{ChapRedMult}. The   geometric formulation of this procedure was given first by 
Stanchenko~\cite{Stanchenko}. Given that the time reparametrisation corresponds to  the vector field rescaling $\overline{X}_{nh} \mapsto g\overline{X}_{nh}$,  we define:
\begin{definition}
A nonholonomic Chaplygin system is said to be \defn{Hamiltonisable} if there exists a positive function $g\in C^\infty(S)$ and a symplectic
form $\overline{\Omega}$ on $T^*S$ such that the rescaled  vector field $g\overline{X}_{nh}$ 
satisfies
\begin{equation*}
 {\bf i}_{g\overline{X}_{nh}}\overline{\Omega}=dH.
\end{equation*}
In this case we say that the vector field $\overline{X}_{nh}$ is \defn{conformally Hamiltonian} and that the system is  Hamiltonisable  \defn{ with the time reparametrisation} $dt =g(s)\, d\tau$.
\end{definition}

In this paper it will always be the case that the  2-form $\overline{\Omega}=g^{-1}\Omega_{nh}$. So the task is to determine conditions 
that guarantee that $\Omega_{nh}$ is closed after multiplication by a positive function $g^{-1}$.
 If such a function exists 
we shall say that $\Omega_{nh}$ is \defn{conformally symplectic} and the function
$g^{-1}$ will be called the \defn{conformal factor}.  A characterisation of the condition that $\Omega_{nh}$ is
conformally symplectic  is given  in our main Theorem~\ref{T:measure-gyro} ahead in terms of 
the gyroscopic tensor.

\begin{remark} 
\label{rmk:Ham-with-Omega0}
We note that according to~\cite{Stanchenko, EhlersKoiller}, 
one could more generally  have
 $\overline{\Omega} =g^{-1}(\Omega_{nh}+\Omega_0)$ where the degenerate 2-form $\Omega_0$ satisfies 
 ${\bf i}_{\overline{X}_{nh}}\Omega_0=0$. The possibility of achieving a 
 Hamiltonisation  with  $\Omega_0\neq 0$, is not considered in this paper and therefore our approach only leads to sufficient conditions for Hamiltonisation.
 \end{remark}

\begin{remark}
\label{Rmk:MomRescaling}
In many references Chaplygin's Hamiltonisation procedure is presented as a time reparametrisation together with 
a rescaling of the momenta.
From the geometric perspective, the rescaling of the momenta serves to obtain Darboux coordinates for the symplectic form 
$\overline{\Omega}$. This is illustrated in our treatment of the nonholonomic particle in Section~\ref{SS:NonhoParticle}.
\end{remark}

Below we show that a Hamiltonisable Chaplygin system indeed admits an invariant measure. This is a
well-known result appearing in various references, e.g. \cite{FedJov,EhlersKoiller}, and which is a consequence of the following general observation for conformally Hamiltonian systems.
\begin{proposition}\label{conformal-invariant-measure}
 Let $\overline{\Omega}$ a symplectic structure on a manifold $P$ of dimension $2r$, $\overline{X}$ a vector field on $P$ and $g$ a real positive $C^{\infty}$-function on $P$ such that $g\overline{X}$ is Hamiltonian vector field with respect to the symplectic structure $\overline{\Omega}$. Then, $g\overline{\Omega}^r$ is an invariant measure for $\overline{X}$.
\end{proposition}
\begin{proof}  
Since $g\overline{X}$ is a Hamiltonian vector field with respect to the symplectic structure $\overline{\Omega}$, it follows that $\overline{\Omega}^r$ is an invariant measure for $g \overline{X}$. Thus, using the properties of the Lie derivative operator ${\mathcal L}$, we have that
\begin{equation*}
0 = {\mathcal L}_{g\overline{X}}\overline{\Omega}^r = g {\mathcal L}_{\overline{X}}\overline{\Omega}^r + dg \wedge i_{\overline{X}}\overline{\Omega}^r.
\end{equation*}
On the other hand, considering that $dg \wedge \overline{\Omega}^r$ is a $(2r+1)$-form on $P$, it identically vanishes, and hence
\[
0 = i_{\overline{X}}(dg \wedge \overline{\Omega}^r) = \overline{X}[g]\overline{\Omega}^r - dg \wedge i_{\overline{X}}\overline{\Omega}^r.
\]
Therefore, using the properties of the Lie derivative and the above equalities, we have
\[
{\mathcal L}_{\overline{X}}(g \overline{\Omega}^r) = g \,{\mathcal L}_{\overline{X}}\overline{\Omega}^r + \overline{X}[g] \overline{\Omega}^r=0.
\]
\end{proof}

\subsection{ $\phi$-simple Chaplygin systems and Hamiltonisation}
\label{Chap-Hamiltonisation}

We now introduce the notion of a  \emph{$\phi$-simple} Chaplygin system that is central to the  results of our paper.

\begin{definition}
\label{D:sigma-gyro}
 A non-holonomic Chaplygin system is said to be \defn{$\phi$-simple},  if  there exists a function $\phi \in C^\infty(S)$ such that the gyroscopic 
 tensor ${\mathcal T}$ satisfies
\begin{equation}\label{Ham-condition}
{\mathcal T}(Y, Z) =Z[\phi]Y - Y[\phi]Z, 
\end{equation}
for all $ Y, Z \in {\frak X}(S)$.
\end{definition}

\begin{remark}
\label{rmk:phi-simple-Noether}
Considering that the definition of the gyroscopic tensor $\mathcal{T}$ is independent of the potential energy, we conclude that
the notion of $\phi$-simplicity is   weakly Noetherian. Namely,
if a   $G$-Chaplygin nonholonomic system is $\phi$-simple, then it  continues to 
be $\phi$-simple under the addition of a
$G$-invariant potential.
\end{remark}

The following is the main result of the paper:

\begin{theorem}
\label{T:measure-gyro}
\begin{enumerate}
\item[(i)] A nonholonomic Chaplygin system is $\phi$-simple if and only if $\Omega_{nh}$ is conformally symplectic with 
conformal factor  $\exp(\phi \circ \tau_S)$ (i.e. $d(\exp(\phi \circ \tau_S)\Omega_{nh}=0$). 
\item[(ii)] The reduced equations of  motion of a $\phi$-simple non-holonomic Chaplygin system    possess
 the basic invariant measure $\mu = \exp (\sigma \circ \tau_S) \, \nu$, where $\nu$ is the Liouville measure and  $\sigma=(r-1)\phi$.
 \item[(iii)] If $r=2$ then statement (ii) may be inverted: if the reduced equations of motion possess
 the basic invariant measure $\mu = \exp (\phi \circ \tau_S) \, \nu$, then the system is $\phi$-simple.
\end{enumerate}
\end{theorem}

In particular, item $(i)$ implies:
\begin{corollary}\label{l:Hamiltonisation-geometric}
A $\phi$-simple Chaplygin system is Hamiltonisable after the time reparametrisation  $dt =\exp(-\phi (s))\, d\tau$.
\end{corollary}
\begin{remark}
Recently Jovanovi\'c~\cite{Jov2019} proved the existence of a Chaplygin system with $r>2$ degrees of freedom
 that possesses an invariant measure
but does not allow a Hamiltonisation, and hence is not $\phi$-simple. 
This shows that the reciprocal of the statement in item $(ii)$ is not true in general if $r>2$. 
Another example of this instance is the Chaplygin sphere treated as a Chaplygin system with $G=\R^2$. The system has an invariant measure but, as shown
in \cite[Section 3]{EhlersKoiller}, the system does not allow a Hamiltonisation at the $T^*\SO(3)$ level.\footnote{We
mention that the Chaplygin sphere does allow a Hamiltonisation when reduced by the larger group 
$\mbox{SE}(2)$. This was first shown by Borisov and Mamaev \cite{BorMamChap}, and the underlying geometry
of this result was first clarified in \cite{LGN10}. }
\end{remark}

For the proof of Theorem~\ref{T:measure-gyro} we will require the following. 
\begin{lemma}\label{previous-main-theorem}
The following conditions are equivalent:
\begin{enumerate}
\item $\Omega_{nh}$ is conformally symplectic with conformal factor $\exp(\phi \circ \tau_S)$,
\item $d\Omega_{nh}=-d(\phi \circ \tau_S) \wedge \Omega_{nh}$,
\item The gyroscopic 2-form $\Omega_{\mathcal T} =\lambda_S\wedge d(\phi \circ \tau_S)$, where $\lambda_S$ is the Liouville $1$-form on $T^*S$ given by
\begin{equation}\label{Liouville-1-form}
\lambda_S(\alpha)(U) = \alpha((T_\alpha\tau_S)(U)), \; \; \mbox{ for } \alpha \in T^*S \mbox{ and } U\in T_\alpha(T^*S).
\end{equation}
\end{enumerate}
\end{lemma}
\begin{proof}
We have
\begin{equation*}
d(\exp(\phi \circ \tau_S)\Omega_{nh})= \exp(\phi \circ \tau_S)\ \left ( d(\phi \circ \tau_S)\wedge \Omega_{nh} + d\Omega_{nh}  \right ) 
\end{equation*}
which shows the equivalence of $(i)$ and $(ii)$.

$(iii)$ $\implies$ $(i)$ This result was first proved by Stanchenko~\cite[Proposition 2]{Stanchenko}  (see also
Cantrijn et al~\cite[Equation~(18)]{CaCoLeMa}). For the sake of completeness, we present a proof here.
Using that $\Omega_{nh}= \Omega_{can} + \Omega_\mathcal{T}$ and that $d\Omega_{can}=0$, we have
\[
\begin{split}
d\left(\exp(\phi \circ \tau_S)\Omega_{nh}\right) &=  \exp \left (\phi \circ \tau_S \right )
\left (d(\phi \circ \tau_s) \wedge (\Omega_{can} + 
\Omega_{\mathcal T}) + d\Omega_{\mathcal T}  \right ) \\
&=   \exp \left (\phi \circ \tau_S \right )
 \left ( d(\phi \circ \tau_s) \wedge \Omega_{can}  + d\Omega_{\mathcal T}  \right ),
\end{split}
\]
where we have used $d(\phi \circ \tau_s) \wedge  \Omega_{\mathcal T}=0$ in the second equality as is implied by  item $(iii)$. On the other hand, taking the exterior differential 
in the expression  $\Omega_{\mathcal T} =\lambda_S\wedge d(\phi \circ \tau_S)$ and
using the well-known relation $d\lambda_S= -\Omega_{can}$  yields
\[
d\Omega_{\mathcal T} = - d(\phi \circ \tau_S) \wedge \Omega_{can},
\]
which shows that $d(\exp(\phi \circ \tau_S)\Omega_{nh})=0$.

$(ii)$ $\implies$ $(iii)$
Starting from 
$d\Omega_{nh} = -d(\phi \circ \tau_S) \wedge \Omega_{nh}$ and using that 
 $\Omega_{nh} = \Omega_{can} + \Omega_{\mathcal T}$ and  $d\Omega_{can} = 0$, we deduce that
\begin{equation}\label{diff-Omega-T}
d\Omega_{\mathcal T} = -d(\phi \circ \tau_S) \wedge \Omega_{can} - d(\phi \circ \tau_S) \wedge \Omega_{\mathcal T}.
\end{equation}
Now, let $\gamma \in \Omega^1(S)$  and $\gamma^{\bf v} \in \mathfrak{X}(T^*S)$ be its vertical lift 
(see \eqref{vertical-lift-dual}). Then, since $\gamma^{\bf v}$ is a vertical vector field  (see e.g. (\ref{vertical-complete-lift-dual}))  and
the 2-form  $\Omega_{\mathcal T}$ is semi-basic, we obtain the identities
\[
{\bf i}_{\gamma^{\bf v}}(d(\phi \circ \tau_S)) = 0, \qquad i_{\gamma^{\bf v}}\Omega_{\mathcal T} = 0.
\]
Therefore, from (\ref{diff-Omega-T}), we have that
\begin{equation}\label{diff-Omega-T-1}
{\bf i}_{\gamma^{\bf v}}d\Omega_{\mathcal T} = d(\phi \circ \tau_S) \wedge i_{\gamma^{\bf v}}\Omega_{can},
\end{equation}
which in view of (\ref{contraction-vertical-Omega-can}) becomes
\[
{\bf i}_{\gamma^{\bf v}}d\Omega_{\mathcal T} = -d(\phi \circ \tau_S) \wedge \tau_S^*\gamma.
\]
In particular, we obtain that
\begin{equation}\label{diff-Omega-T-2}
({\bf i}_{\gamma^{\bf v}}d\Omega_{\mathcal T}) \circ \gamma = -(d(\phi \circ \tau_S) \wedge \tau_S^*\gamma) \circ \gamma.
\end{equation}
In addition, using the local expressions of the $2$-form $\Omega_{\mathcal T}$, given respectively by (\ref{Omega-T-local}) and (\ref{vertical-complete-lift-dual}),   we can prove that
\begin{equation}\label{diff-Omega-T-3}
({\bf i}_{\gamma^{\bf v}}d\Omega_{\mathcal T}) \circ \gamma = \Omega_{\mathcal T} \circ \gamma.
\end{equation} 
Moreover, the definition (\ref{Liouville-1-form}) of the Liouville $1$-form $\lambda_S$, implies that
\[
\left( \lambda_S(\gamma(s)) \right)(U) = \gamma(s)\left((T_{\gamma(s)}\tau_S)(U)\right) = \left(\tau_S^*\gamma\right)(\gamma(s))(U),
\]
for $s \in S$ and $U\in T_{\gamma(s)}(T^*S)$. In other words, 
\[
\lambda_S \circ \gamma = (\tau_S^*\gamma) \circ \gamma.
\]
Thus, from (\ref{diff-Omega-T-2}) and (\ref{diff-Omega-T-3}), we conclude that
\[
\Omega_{\mathcal T} \circ \gamma = -(d(\phi \circ \tau_S) \wedge \lambda_S) \circ \gamma,
\]
and since the  $1$-form $\gamma$ is arbitrary we obtain 
\[
\Omega_{\mathcal T} = \lambda_S \wedge d(\phi \circ \tau_S),
\]
as required.
\end{proof}

We are now ready to present the proof of Theorem \ref{T:measure-gyro}.

\begin{proof}(of Theorem \ref{T:measure-gyro})
$(i)$ Suppose that the system is $\phi$-simple so~\eqref{Ham-condition} holds. Using this equation in the definition~\eqref{eq:defgyrotwoform} of $\Omega_\mathcal{T}$  we have that, for $\alpha \in T^*S$ and $U, V \in T_\alpha(T^*S)$
 \[
 \begin{split}
\Omega_{\mathcal T}(\alpha)(U, V) &=  d\phi (T_\alpha \tau_S V) \alpha (T_\alpha\tau_S(U))  - d\phi (T_\alpha \tau_S U) \alpha (T_\alpha\tau_S(V)) 
\\
& = d(\phi \circ \tau_S)(\alpha) (V) \lambda_S(\alpha)(U) - d(\phi \circ \tau_S)(\alpha) (U) \lambda_S(\alpha)(V)\\
& = \lambda_S \wedge d(\phi \circ \tau_S)(\alpha)(U,V).
\end{split}
\]
So, it follows that $\Omega_{\mathcal T} = \lambda_S \wedge d(\phi \circ \tau_S)$ and Lemma~\ref{previous-main-theorem} implies that $d(\exp(\phi \circ \tau_S)\Omega_{nh}) = 0$ as required. 

Conversely, assume that $d(\exp(\phi \circ \tau_S)\Omega_{nh}) = 0$. Then, using again Lemma \ref{previous-main-theorem}, we deduce that $\Omega_{\mathcal T} = \lambda_S \wedge d(\phi \circ \tau_S)$. So, from the definition~\eqref{eq:defgyrotwoform} of $\Omega_\mathcal{T}$ we obtain
\[
\begin{split}
&\alpha\left({\mathcal T}(\tau_S(\alpha))((T_\alpha\tau_S)(U), (T_\alpha \tau_S)(V)\right) 
= \\ &\qquad \qquad \qquad \alpha\left(    (T_\alpha \tau_S)(V)[\phi] (T_\alpha \tau_S)(U)- (T_\alpha \tau_S)(U)[\phi] (T_\alpha \tau_S)(V)\right).
\end{split}
\]
Therefore, since $\tau_S$ is a submersion, we conclude that
\[
{\mathcal T}(s)(Y, Z) = Z[\phi]Y - Y[\phi]Z, \; \; \mbox{ for } Y, Z \in T_sS,
\]
that is, the system is $\phi$-simple.

$(ii)$ This is a consequence of item $(i)$ and Proposition~\ref{conformal-invariant-measure} but we present a simple alternative proof.
Suppose that the system is $\phi$-simple so~\eqref{Ham-condition} holds and let $\sigma=(r-1)\phi$. Let $Y\in \frak{X}(S)$, then
 the 1-form $\Theta$ defined by~\eqref{eq:1-form-theta-def}
satisfies
\begin{equation*}
\begin{split}
\Theta(Y)&= \sum_{i=1}^r  \left \langle X^i \, , \, \left (Y[\phi]X_i - X_i[\phi]Y  \right )\right \rangle \\
&=rY[\phi] -\sum_{i=1}^k \langle X_i,Y\rangle X_i[\phi]  \\
&=(r-1)Y[\phi].
\end{split}
\end{equation*}
Therefore $\Theta=d\sigma$ and  Theorem~\ref{T:measures-1-form} implies the preservation of the measure $\mu = \exp (\sigma) \, \nu$. 

$(iii)$ Now suppose that $r=2$ and that the reduced equations of motion preserve the basic measure $\mu = \exp (\phi \circ \tau_S) \, \nu$. Then 
item $(i)$ of Theorem~\ref{T:measures-1-form} implies that $\Theta=d\phi$. In view of the local expression~\eqref{eq:Theta-local}
for $\Theta$ we obtain the following relations between the partial derivatives of $\phi$ and the gyroscopic coefficients:
\begin{equation*}
\frac{\partial \phi}{\partial s^1}=C_{21}^2(s), \qquad \frac{\partial \phi}{\partial s^2}=C_{12}^1(s).
\end{equation*}
Therefore,
\begin{equation*}
\begin{split}
{\mathcal T}\left ( \frac{\partial}{\partial s^1} , \frac{\partial}{\partial s^2} \right ) &= C_{12}^1(s) \frac{\partial}{\partial s^1}+ C_{12}^2(s) \frac{\partial}{\partial s^2}
\\
&= \frac{\partial \phi}{\partial s^2}\frac{\partial}{\partial s^1} - \frac{\partial \phi}{\partial s^1}\frac{\partial}{\partial s^2}.
\end{split}
\end{equation*}
The above expression shows that the $\phi$-simplicity relation~\eqref{Ham-condition} holds  for the basis 
$\{\frac{\partial}{\partial s^1},\frac{\partial}{\partial s^2}\}$, and hence, by linearity, for
 general $Y, Z\in \frak{X}(S)$.
\end{proof}

\begin{corollary}\label{C:ChapThm}
\begin{enumerate}
\item[(i)] A purely kinetic, nonholonomic Chaplygin system with 2 degrees of freedom possesses an invariant measure if and only if 
  it is $\phi$-simple.
  \item[(ii)] [Chaplygin's Reducing Multiplier Theorem \cite{ChapRedMult}]
 If a  Chaplygin system with 2 degrees of freedom preserves 
the basic measure $\mu = \exp (\sigma \circ \tau_S) \, \nu$, then
the system is  Hamiltonisable after the time reparametrisation  $dt =\exp(-\sigma (s))\, d\tau$. 
\item[(iii)] The Hamiltonisation of a $\phi$-simple   Chaplygin system by the  time reparametrisation \\* $dt =\exp(-\sigma (s))\, d\tau$
is weakly Noetherian. Namely, the same   time reparametrisation Hamiltonises the system under the addition  of an arbitrary $G$-invariant potential.
\end{enumerate}
\end{corollary}
\begin{proof}
$(i)$ is a direct consequence of item $(ii)$ of Theorem~\ref{T:measures-1-form} and item $(iii)$ of Theorem~\ref{T:measure-gyro}. Item $(ii)$
follows from item $(i)$ and Corollary \ref{l:Hamiltonisation-geometric}. Finally, item $(iii)$ follows from  Remark~\ref{rmk:phi-simple-Noether}.
\end{proof}

We now show that the sufficient
 conditions for Hamiltonisation that were recently obtained in Garc\'ia-Naranjo~\cite{LGN18} are equivalent to  $\phi$-simplicity.
For this matter we recall the so-called hypothesis   \ref{hyp} from this reference:
\begin{customhyp}{(H)}
\label{hyp}
The gyroscopic coefficients $C_{ij}^k$ written in the coordinates $(s^1, \dots, s^r)$ satisfy:
\begin{equation*}
C_{ij}^k= 0, \quad \mbox{for} \quad k\neq i\neq j\neq k, \qquad C_{ij}^j=C_{ik}^k \quad \mbox{for all } \quad j, k\neq i.
\end{equation*}
\end{customhyp}

It is shown in~\cite{LGN18} that~\ref{hyp} is an intrinsic condition (independent of the choice of coordinates). 
Moreover, in this reference it is also shown that if  \ref{hyp} holds, and  the basic measure  $\mu=\exp (\sigma)\nu$ is preserved by the
reduced flow, then the system is Hamiltonisable with the time reparametrisation $dt =\exp(\sigma/(1-r))\, d\tau$. Proposition~\ref{P:H-Hamilt} below shows that these two hypothesis taken together are  equivalent to 
the condition that the system is $\phi$-simple with $\phi =\sigma/(r-1)$, so the Hamiltonisation result of~\cite{LGN18} is 
a particular consequence of  Corollary~\ref{l:Hamiltonisation-geometric}.

Before presenting  Proposition~\ref{P:H-Hamilt}  and its proof, we note
 that   \ref{hyp} is equivalent to the existence of a 1-form $\beta$ on $S$ such that the gyroscopic tensor satisfies
${\mathcal T}(Y, Z) =\beta(Z)Y - \beta(Y) Z$, for vector fields $Y,Z\in \frak{X}(S)$. In this case, using its  definition~\eqref{eq:1-form-theta-def},
it is easy to show that the 1-form
 $\Theta =(r-1) \beta$. Therefore,  the condition   \ref{hyp}  may be reformulated as:
\begin{customhyp}{(H')}
\label{hypp}
The gyroscopic tensor satisfies
\begin{equation*}
{\mathcal T}(Y, Z) = \frac{1}{r-1}\left ( \Theta(Z)Y - \Theta(Y) Z \right ),
\end{equation*}
for vector fields $Y,Z\in \frak{X}(S)$.
\end{customhyp}

\begin{proposition}
\label{P:H-Hamilt}
A Chaplygin system is $\phi$-simple  
if and only if {\em  \ref{hyp}} holds and  the reduced equations of motion preserve 
the invariant measure $\mu=\exp (\sigma)\nu$ with $\sigma =(r-1)\phi$ (where, as usual, $\nu$ is the Liouville measure in $T^*S$).
\end{proposition}

\begin{proof}
Suppose that the Chaplygin system under consideration is $\phi$-simple. Then, item $(ii)$ of Theorem~\ref{T:measure-gyro}
implies that the measure $\mu=\exp (\sigma)\nu$ in the statement of the proposition is preserved by the flow of the reduced system.
Moreover, because of  item $(i)$ of Theorem~\ref{T:measures-1-form},  the invariance of $\mu$  implies that  $\Theta=d\sigma =(r-1)d\phi$.
Substituting $d\phi= \frac{1}{r-1}\Theta$ in~\eqref{Ham-condition} shows that   \ref{hypp}  and hence also    \ref{hyp} holds.

Conversely, using again  item $(i)$ of Theorem~\ref{T:measures-1-form}, the invariance of $\mu$ implies $\Theta = d\sigma$.
So   \ref{hypp}  implies that the system is  $\phi$-simple with $\phi=\sigma/(r-1)$.
\end{proof}

We finish this section with the following remark concerning the work of Hochgerner \cite{Hoch}.

\begin{remark}
Following \cite{LiMa} (see Proposition 16.5 in Chapter 1), the differential of $\Omega_{nh}$ may be 
written  uniquely  as
\[
d\Omega_{nh} = \psi + \frac{1}{r-1} (\delta \Omega_{nh}) \wedge \Omega_{nh},
\]
where $\psi$ is an ``effective $3$-form" with respect to $\Omega_{nh}$ and $\delta$ is the almost symplectic codifferential associated with $\Omega_{nh}$. This is Lepage's decomposition of the $3$-form $d\Omega_{nh}$ (for more details, see \cite{LiMa}).  

In a setup that is more general than that of Chaplygin systems, Hochgerner \cite[Theorem 2.3]{Hoch} gives sufficient conditions for Hamiltonisation by requiring  that 
 $\psi=0$, that the system is kinetic, and there exists an invariant measure.
Under these assumptions one can prove 
that the codifferential $\delta \Omega_{nh}= -d(\phi \circ \tau_S)$ for a certain $\phi\in C^\infty (S)$ and Lemma~\ref{previous-main-theorem}
 together with Theorem \ref{T:measure-gyro}  imply that the system is $\phi$-simple.  Hence, in the context of 
 Chaplygin systems,  the  Hamiltonisation 
 criteria of \cite{Hoch} are always satisfied by $\phi$-simple systems.
\end{remark}

\subsection{Relation of the gyroscopic tensor with previous constructions in the literature}
\label{SS:history}

In this section we show that  the gyroscopic tensor ${\mathcal T}$ appears in the previous works of Koiller~\cite{Koi} 
and Cantrijn et al~\cite{CaCoLeMa}.  The occurrence of  ${\mathcal T}$ in these works is in the geometric study of Chaplygin systems using
an affine connection approach.

Let ${\frak g}$ be the Lie algebra of $G$ and denote by $\xi_Q \in \frak{X}(Q)$ the infinitesimal generator of the $G$-action associated with $\xi \in {\frak g}$.
As is well known,  the \defn{curvature form} of the principal connection is the map   $\Psi: TQ \times_Q TQ \to {\frak g}$, characterised by the condition
\begin{equation}\label{curvature-connection}
(\Psi (\mbox{ hor}_q  \, Y, \mbox{ hor}_q \, Z))_Q(q) = (\mbox{ hor }[Y, Z] - [\mbox{ hor }Y, \mbox{ hor }Z])(q),
\end{equation}
for $Y, Z$ vector fields on $S$.
The following proposition gives an expression for the  gyroscopic tensor $\mathcal{T}$ in terms of $\Psi$.
\begin{proposition}\label{gyroscopic-curvature}
If $Y, Z$ are vector fields on $S$ and $q \in Q$ then
\[
{\mathcal T}(Y, Z)(\pi(q)) = -(T_q\pi)({\mathcal P}( (\Psi (\mbox{\em hor}_q  \, Y, \mbox{\em hor}_q \, Z))_Q(q))).
\]
\end{proposition}
\begin{proof}
Using (\ref{gyroscopic-tensor}), we have 
\begin{equation*}
\begin{split}
{\mathcal T}(Y, Z)(\pi(q)) = (T_q\pi)({\mathcal P}[\mbox{ hor }Y, \mbox{ hor }Z] - \mbox{ hor }[Y, Z])(q).
\end{split}
\end{equation*}
But  $ \mbox{hor}\, [Y, Z] = {\mathcal P}(\mbox{hor}\, [Y, Z] )$ since $\mbox{ hor }[Y, Z]$ is a section of $D$, so using the characterisation
 \eqref{curvature-connection} of $\Psi$, 
the result follows. 
\end{proof}
Consider now the $(0,3)$ tensor field ${\mathcal K}$ on $S$ defined by
\begin{equation*}
{\mathcal K}(X,Y,Z) = \llangle \mbox{ hor } X , \mbox{ hor }  {\mathcal T}(Y, Z) \rrangle ,
\end{equation*}
for vector fields  $X,Y,Z\in \frak{X}(S)$. Then Proposition~\ref{gyroscopic-curvature} shows that 
\begin{equation}\label{Tensor-K}
\mathcal K (X,Y,Z) (\pi(q)) =-\llangle \mbox{ hor}_q\, X(q), (\Psi(\mbox{ hor}_q Y, \mbox{ hor}_q \,Z))_Q(q) \rrangle_q .
\end{equation}
This implies that $-{\mathcal K}$ coincides with the tensor $\tilde K$ in~\cite[Page 337]{CaCoLeMa}. As explained in this reference, this tensor is 
induced by the so-called {\em metric connection} tensor $K$, of type $(0,3)$ and defined on $Q$, that was first introduced by Koiller~\cite[Equation (3.14a)]{Koi}.

The above observation implies that, up to a sign, the gyroscopic tensor $\mathcal{T}$ coincides with the tensor field denoted by $C$ in 
 \cite[Proposition 8.5]{Koi} and~\cite[Page 337]{CaCoLeMa}.

\begin{remark}Another relation between our constructions and \cite{Koi, CaCoLeMa} involves the 
  tensor field $\mathcal{B}$ of type $(2,1)$ on $S$ defined by \eqref{T-bar}. To see this, note that~\eqref{reduced-metric} and \eqref{T-bar}
  imply 
 \[
 \langle \mathcal{B}(\alpha, \beta), X \rangle (\pi(q))= \llangle \mbox{ hor }\alpha^{\sharp}, \mbox{ hor }{\mathcal T}(\beta^{\sharp}, X) \rrangle (q), \; \mbox{ for } q \in Q, \alpha, \beta \in \Omega^1(S) \mbox{ and } X \in {\frak X}(S).
 \]
So, using~\eqref{Tensor-K}, Proposition~\ref{gyroscopic-curvature}, and the fact that the tensor ${\mathcal K}=-\tilde K$ is skew-symmetric in the last two arguments, we deduce that
\[
 \langle {\mathcal B}(\alpha, \beta), X \rangle = -{\mathcal K}(\alpha^{\sharp}, \beta^{\sharp}, X) =  \tilde{K}(\alpha^{\sharp}, \beta^{\sharp}, X).
 \]
 This relation implies that
 \[
 B(\alpha^{\sharp}, \beta^{\sharp}) = {\mathcal B}(\alpha, \beta)^{\sharp},
 \]
 where $B$ is the tensor of type $(1, 2)$ on $S$ considered in ~\cite[Page 337]{CaCoLeMa} and that was first introduced by Koiller in \cite[Equation (8.9)]{Koi}. So, in the terminology of the Riemannian geometry (see, for instance, \cite{On}), $ {\mathcal B}$ 
 and $B$ are metrically equivalent.
 \end{remark}

\subsection*{$\Omega_ {\mathcal T}$ coincides with minus the ``$\langle J, K\rangle $"  term}

 A number of references in nonholonomic Chaplygin systems (e.g. \cite{HochGN,BalseiroFdz,KoillerRubber,Hoch})
  follow the construction in  \cite{KoRiEh,EhlersKoiller} and write
 the almost symplectic structure $\Omega_{nh}$ on $T^*S$ as\footnote{Actually   \cite{EhlersKoiller} write $\Omega_{nh} = \Omega_{can}+\langle J, K\rangle$. The difference in sign is 
 due to the convention on the canonical 2-form on $T^*S$. In   \cite{KoRiEh,EhlersKoiller} it is taken as $d\lambda_S=dp\wedge ds$ while
 we take it as $-d\lambda_S=ds\wedge dp$.}
 \begin{equation}
 \label{eq:nhJK}
\Omega_{nh} = \Omega_{can}-\langle J, K\rangle,
\end{equation}
where ``$\langle J, K\rangle$" is a semi-basic 2-form on $T^*S$ obtained by pairing the momentum map and the curvature
of the principal connection $D$. Here we recall the construction of ``$\langle J, K\rangle$" and show that  it coincides with  
$-\Omega_\mathcal{T}$ defined by \eqref{eq:defgyrotwoform}. This shows that our definition of  $\Omega_{nh}$ in \eqref{eq:defOmnhIntrinsic} 
is consistent with \eqref{eq:nhJK}. 
Since we have introduced the curvature in \eqref{curvature-connection} with the symbol $\Psi$, we
write $\langle J, \Psi \rangle$ instead of $\langle J, K\rangle$ from now on.

We begin by noticing that the definition (\ref{eq:defgyrotwoform}) of the gyroscopic $2$-form $\Omega_{\mathcal T}$ together with Proposition \ref{gyroscopic-curvature} lead to the expression
 \begin{equation}\label{gyroscopic-2-form-1}
\Omega_{\mathcal T}(\alpha)(U, V) = -\alpha \left ( (T_q\pi \circ {\mathcal P}) (\Psi ( \mbox{ hor}_q  \, (T_\alpha\tau_S)(U), \mbox{ hor}_q \, (T_\alpha\tau_S)(V))_Q(q)) \right ). 
\end{equation}

On the other hand, we recall that the  momentum map $J: T^*Q \to {\frak g}^*$ is  given by
\begin{equation}\label{momentum-map}
\langle J(\tilde{\alpha}), \xi \rangle = \langle \tilde{\alpha}, \xi_Q(q) \rangle, \qquad \mbox{ for } \tilde{\alpha} \in T_q^*Q \mbox{ and } \xi \in \frak{g}.
\end{equation} 
References  \cite{KoRiEh,EhlersKoiller} define the action of the 
2-form $\langle J, \Psi \rangle$ on the vectors $U,V\in T_\alpha (T^*S)$, by
\begin{equation}\label{J-K}
\langle J, \Psi \rangle (\alpha)(U, V) = \langle J(\tilde{\alpha}), \Psi(\mbox{ hor}_q \, (T_\alpha\tau_S)(U), \mbox{ hor}_q \, (T_\alpha\tau_S)(V)) \rangle,
\end{equation}
where $q\in Q$ is any\footnote{The construction is independent of the choice of $q$ since the $\mbox{Ad}^*$-equivariance of 
$J$ is cancelled with the  $\mbox{Ad}$-equivariance of $\Psi$, see \cite{KoRiEh,EhlersKoiller}} point satisfying $\pi(\tau_S(\alpha))=q$, and $\tilde{\alpha} \in T^*_qQ$ is defined through the ``clockwise diagram''  \cite[Diagram (3.11)]{EhlersKoiller}. In our notation this is
\begin{equation}\label{alpha-tilde}
\mbox{$\tilde{\alpha} = (\flat_q \circ \, \mbox{hor}_q \circ \sharp_{\pi(q)})(\alpha)$},
\end{equation}
where $\flat_q: T_qQ \to T^*_qQ$ and $\sharp_{\pi(q)}: T^*_{\pi(q)}S \to T_{\pi(q)}S$ are the linear isomorphisms induced by the Riemannian metrics $\llangle \cdot, \cdot \rrangle $ and $\llangle \cdot, \cdot \rrangle^{-} $ on $Q$ and $S$, respectively (see \eqref{eq:metric-dual} and \eqref{eq:musical-isomorphisms}).
In view of (\ref{momentum-map}) and (\ref{alpha-tilde}), we rewrite  (\ref{J-K}) as 
\begin{equation}\label{J-K-1}
\mbox{$\langle J, \Psi \rangle (\alpha)(U, V) =  \langle \alpha, (\flat_q \circ \, \mbox{hor}_q \circ \sharp_{\pi(q)})^*(\Psi(\mbox{ hor}_q \, (T_\alpha\tau_S)(U), \mbox{ hor}_q \, (T_\alpha\tau_S)(V)))_Q(q)\rangle$},
\end{equation}
where $(\flat_q \circ \mbox{ hor}_q \circ \sharp_{\pi(q)})^*: T_qQ \to T_{\pi(q)}S$ is the dual morphism of the linear map $\flat_q \circ \mbox{ hor}_q \circ \sharp_{\pi(q)}: T_{\pi(q)}^*S \to T_q^*Q$.
Below we will prove that
\begin{equation}
\label{eq:auxJK}
\mbox{$(\flat_q \circ \mbox{ hor}_q \circ \sharp_{\pi(q)})^* =  \left . T_q\pi \circ {\mathcal P} \right |_{T_qQ}.$}
\end{equation}
Therefore, up to a sign, the expression
  (\ref{J-K-1}) equals the right hand side of (\ref{gyroscopic-2-form-1}) proving that $\langle J, \Psi \rangle=
-\Omega_\mathcal{T}$ as claimed.

In order to prove that~\eqref{eq:auxJK} indeed holds,
let  $v \in T_qQ$ and $\alpha \in T_{\pi(q)}^*S$. Considering that  
\[
\mbox{$(\mbox{ hor}_q \circ \sharp_{\pi (q)})(\alpha) \in D_q$}\] we have
\[
\mbox{$\langle \alpha, (\flat_q \circ \mbox{ hor}_q \circ \sharp_{\pi(q)})^*(v)\rangle = \llangle v, (\mbox{ hor}_q \circ \sharp_{\pi(q)})(\alpha) \rrangle = \llangle {\mathcal P}(v), (\mbox{ hor}_q \circ \sharp_{\pi(q)})(\alpha) \rrangle.$}
\]
Finally, using that $\mbox{ hor}_q((T_q\pi)({\mathcal P}v)) = {\mathcal P}(v)$ and the definition (\ref{reduced-metric}) of the Riemannian metric $\llangle \cdot, \cdot \rrangle^{-}$ on $S$, we obtain 
\begin{equation*}
\begin{split}
\mbox{$\langle \alpha, (\flat_q \circ \mbox{ hor}_q \circ \sharp_{\pi(q)})^*(v)\rangle$} & = \llangle \mbox{hor}_q((T_q\pi) ({\mathcal P}v)), \, \mbox{hor}_q(\mbox{$\sharp_{\pi(q)}$}\alpha) \rrangle  \\  &= \mbox{ $\llangle (T_q\pi)({\mathcal P}v), \sharp_{\pi(q)}(\alpha) \rrangle^{-}$} \\
&=  \langle \alpha, T_q\pi({\mathcal P}(v))\rangle,
\end{split}
\end{equation*}
which shows that $(\flat_q \circ \mbox{ hor}_q \circ \sharp_{\pi(q)})^*(v) = T_q\pi({\mathcal P}(v))$, and therefore \eqref{eq:auxJK} indeed holds.

\section{Examples} 
\label{S:Example}

We present two examples. The first one is  the nonholonomic particle considered by Bates and \'Sniatycki~\cite{BS93} with a slightly  more general kinetic energy,
that we include  to illustrate the
geometric  constructions of Section~\ref{S:Chaplygin}   in a toy system and
to prove that, in the presence of a potential, there may exist Chaplygin systems possessing only an invariant
measure  that
is not basic.
The second example is more involved and treats the multi-dimensional version of the Veselova problem considered by 
Fedorov and Jovanovi\'c~\cite{FedJov, FedJov2}. We show that the system is $\phi$-simple and hence, the 
remarkable Hamiltonisation of the problem obtained in~\cite{FedJov, FedJov2} may be understood as a consequence of 
Corollary~\ref{l:Hamiltonisation-geometric}. 

\subsection{The nonholonomic particle}
\label{SS:NonhoParticle}

The configuration space of the system is $Q=\R^3$ with coordinates $(x,y,z)$. The Lagrangian $L$ and the nonholonomic
constraint are given by
\begin{equation*}
L=\frac{1}{2}(\dot x^2+\dot y^2 +\dot z^2)+a\dot y \dot z -U(x,y), \qquad \dot z=y\dot x,
\end{equation*}
for a potential $U\in C^\infty(\R^2)$ and a constant real number $|a|<1$. The nonholonomic constraint defines the
  constraint distribution $D=\mbox{span}\{\partial_x+y\partial_z, \partial_y\}$.

The symmetry group is $G=\R$ acting by translations of $z$. It is easy to check that both $L$ and $D$ are invariant under
the lifted action to $TQ$. Notice that $\mbox{span} \{ \partial_z\} \oplus D$ is the total   tangent space to $Q=\R^3$ at every point,
so the condition~\eqref{eq:chap-condition} is satisfied and the system is $G$-Chaplygin (with $r=2$ degrees of freedom).

The shape
space is $S=\R^3/\R=\R^2$ with coordinates $(x,y)$ and the orbit projection $\pi:Q\to S$ is $\pi(x,y,z)=(x,y)$. 
We now proceed to compute the gyroscopic tensor.
The horizontal lift of the coordinate vector fields on $S$ is 
\begin{equation*}
\mbox{hor} \,  \partial_x =\partial_x+y\partial_z, \qquad \mbox{hor} \,  \partial_y =\partial_y.
\end{equation*}
And so, $[\mbox{hor} \,  \partial_x , \mbox{hor} \,  \partial_y ]= -\partial_z$. 
On the other hand,  the orthogonal complement
of $D$ with respect to the  kinetic energy metric is $$D^\perp = \mbox{span} \{(1-a^2)y \partial_x+a\partial_y- \partial_z \},$$ so we have the orthogonal decomposition
\begin{equation*}
[\mbox{hor} \,  \partial_x , \mbox{hor} \,  \partial_y ] =\frac{-(1-a^2)y}{1+(1-a^2)y^2}\left ( \partial_x+y\partial_z \right )  
+ \frac{-a}{1+(1-a^2)y^2}  \partial_y +
 \frac{1}{1+(1-a^2)y^2}\left ((1-a^2)y \partial_x+a\partial_y- \partial_z \right ),
\end{equation*}
 which implies $$\mathcal{P} [\mbox{hor} \,  \partial_x , \mbox{hor} \,  \partial_y ] = \frac{-(1-a^2)y}{1+(1-a^2)y^2}\left ( \partial_x+y\partial_z \right )  
+ \frac{-a}{1+(1-a^2)y^2}  \partial_y.$$
Hence, the gyroscopic tensor defined by~\eqref{gyroscopic-tensor} is determined by its action on the basis vectors by:
\begin{equation}
\label{eq:gyr-tensor-nhparticle}
\mathcal{T} (  \partial_x,  \partial_y) =   \frac{-(1-a^2)y}{1+(1-a^2)y^2}\partial_x + \frac{-a}{1+(1-a^2)y^2}  \partial_y .
\end{equation}

The 1-form $\Theta$ defined by \eqref{eq:1-form-theta-def} is given by
\begin{equation*}
\Theta = \frac{a\, dx - (1-a^2)y \, dy}{1+(1-a^2)y^2},
\end{equation*}
and, for $|a|<1$, it is exact if and only if $a=0$. Therefore, in accordance with item (i) of Theorem~\ref{T:measures-1-form},
 the system possesses a basic invariant measure if
and only if $a=0$.

For $a=0$ the  expression \eqref{eq:gyr-tensor-nhparticle} may be rewritten as 
$\mathcal{T} (  \partial_x,  \partial_y) =   \frac{\partial \phi}{\partial y}\partial_x - \frac{\partial \phi}{\partial x}\partial_y$,
where $\phi(y)=-\frac{1}{2}\ln (1+y^2)$,  so the system is $\phi$-simple (Definition~\ref{D:sigma-gyro}). We conclude from 
Theorem~\ref{T:measure-gyro} and Corollary~\ref{l:Hamiltonisation-geometric} that the reduced equations of motion on $T^*\R^2$ preserve the measure
$\mu= (1+y^2)^{-1/2}dx\wedge dy\wedge dp_x\wedge dp_y$, and become Hamiltonian after the time reparametrisation $dt =(1+y^2)^{1/2}\, d\tau$.

The above conclusions were obtained without writing the reduced equations of motion. For completeness, we now derive
them in their almost Hamiltonian form for general $a$. Using~\eqref{eq:HamChap-local} the reduced Hamiltonian $H\in C^\infty(T^*\R^2)$ is computed to be
\begin{equation*}
H(x,y,p_x,p_y)=\frac{p_x^2+ (1+y^2)p_y^2 -2ay p_xp_y}{2(1+(1-a^2)y^2)}  + U(x,y).
\end{equation*}
On the other hand, using the definition~\eqref{eq:defgyrotwoform} we compute the value of the semi-basic gyrosocopic 2-form 
$\Omega_\mathcal{T}$  on the vectors $ \partial_x, \partial_y$ to be given  by
\begin{equation*}
\begin{split}
\Omega_\mathcal{T}(x,y,p_x,p_y) \left ( \partial_x, \partial_y \right )  = (p_x \, dx + p_y \, dy) \left (\mathcal{T} (  \partial_x,  \partial_y) \right ) 
 =- \frac{(1-a^2)yp_x+ap_y}{1+(1-a^2)y^2}.
\end{split}
\end{equation*}
It follows that $\Omega_\mathcal{T} = - \left ( \frac{(1-a^2)yp_x+ap_y}{1+(1-a^2)y^2}\right ) \, dx\wedge dy$ and hence
 $$\Omega_{nh} = dx\wedge dp_x+ dy\wedge dp_y   - \left ( \frac{(1-a^2)yp_x+ap_y}{1+(1-a^2)y^2}\right ) \, dx\wedge dy.$$
Using the above expressions, one may verify that the vector field $\overline{X}_{nh}$, determined by 
${\bf i}_{\overline{X}_{nh}}\Omega_{nh}=dH$, defines the following reduced equations on $T^*\R^2$: 
\begin{equation}
\label{eq:motion-nhparticle}
\begin{split}
&\dot x = \frac{p_x-ayp_y}{1+(1-a^2)y^2}, \qquad    \dot y =  \frac{-ay p_x+(1+y^2)p_y}{1+(1-a^2)y^2}  \\
& \dot p_x =  -\frac{\partial U}{\partial x }+ \frac{((1-a^2)y p_x+ ap_y)(-ayp_x+(1+y^2)p_y)}{(1+(1-a^2)y^2)^2}, \qquad \dot p_y = -\frac{\partial U}{\partial y }.
\end{split}
\end{equation}

For $a=0$ one can check directly that the above equations preserve the measure $\mu= (1+y^2)^{-1/2}dx\wedge dy\wedge dp_x\wedge dp_y$ given before and that the 2-form $\overline \Omega =(1+y^2)^{-1/2}\Omega_{nh}$ is closed. Moreover, the coordinates $(x,y,\tilde p_x, \tilde p_y)$ with
\begin{equation*}
\tilde p_x = (1+y^2)^{-1/2}p_x, \qquad \tilde p_y=(1+y^2)^{-1/2}p_y, 
\end{equation*}
satisfy $\overline \Omega=dx\wedge d\tilde p_x+ dy\wedge d\tilde p_y$, i.e. they are Darboux coordinates for $\overline \Omega$.
This rescaling of the momenta is natural since these variables are proportional to the velocities and we 
have introduced the  time reparametrisation $dt =(1+y^2)^{1/2}\, d\tau$. This rescaling
 is often encountered in the literature as an ingredient of Chaplygin's Hamiltonisation method (see Remark~\ref{Rmk:MomRescaling}).

Finally, we notice that for more general $|a|<1$ and the special potential
\begin{equation*}
U_a(x,y)=\frac{1}{4}\ln (1+ (1-a^2)y^2),
\end{equation*}
the equations \eqref{eq:motion-nhparticle} possess the non-basic invariant measure
\begin{equation*}
\eta_a= \exp \left (ax+p_y^2 \right ) \, dx \wedge dy \wedge dp_x \wedge dp_y.
\end{equation*}
This example shows that in the presence of a potential the exactness of  the 1-form $\Theta$ is not a necessary condition for the existence of a general 
smooth  invariant
volume form. In this case, it is only  a
necessary condition for the existence of a {\em basic}  invariant measure  (see Theorem~\ref{T:measures-1-form}).
\begin{remark}
For $a=0$ the system possesses the invariant measures $\mu$ and $\eta_0$. 
There is no contradiction since  we have $\eta_0 = F(y,p_y) \mu$ with $F(y, p_y)= \exp \left (p_y^2 \right ) \left (1+ y^2 \right )^{1/2} $ which is a first integral of  \eqref{eq:motion-nhparticle} when $a=0$. \end{remark}

\subsection{The multi-dimensional Veselova system}
\label{SS:Veselova}

The Veselova system, introduced by Veselova~\cite{Veselova},  concerns the motion of  a rigid body that 
rotates under its own inertia and  is subject
to a  nonholonomic constraint that enforces the projection of the angular velocity to  an axis that is fixed in {\em space} to vanish  (see also \cite{Veselov-Veselova-1988}). A multi-dimensional version of the system was  considered by Fedorov and Kozlov~\cite{FedKoz}, and later
by  Fedorov and  Jovanovi\'c~\cite{FedJov, FedJov2}. In these two papers the authors show that, for a {\em special family of inertia tensors},
the system allows a Hamiltonisation via Chaplygin's multiplier method. Apparently, this was the first time that 
 Chaplygin's method was successfully applied to obtain the Hamiltonisation of a nonholonomic Chaplygin system with an arbitrary number  $r$ of degrees of freedom.
Other examples having this property have since been reported in the literature   \cite{ JovaRubber, Jova18, LGN18, LGNTAM-2019}.
  
In this section we show that the gyroscopic tensor of the multi-dimensional Veselova problem  considered 
by Fedorov and Jovanovi\'c~\cite{FedJov, FedJov2} is $\phi$-simple.   In this manner, we show that the remarkable result of~\cite{FedJov, FedJov2} falls under
 the umbrella of Theorem~\ref{T:measure-gyro}  and
Corollary~\ref{l:Hamiltonisation-geometric}. Our description of the problem is kept brief. Readers
who are not familiar with the system may wish to consult~\cite{FedKoz, FedJov,FedJov2, FGNM18} for more details.

The configuration space of the system is $Q=\SO(n)$. An element $g\in \SO(n)$ specifies the attitude of the multi-dimensional rigid body by relating
a frame that is fixed in the body  with an inertial  frame that is fixed in space. The angular velocity in the 
{\em body} frame is the skew-symmetric matrix
\begin{equation}
\Omega =g^{-1}\dot g\in \so(n).
\end{equation}
The Lagrangian  $L:T\SO(n)\to \R$ is the kinetic minus the potential energy. In the left trivialisation of $T\SO(n)$ it is given by
\begin{equation}
\label{eq:Lagrangian-Veselova}
L(g,\Omega) =\frac{1}{2}( \I \Omega, \Omega  )_\kappa-U(g).
\end{equation}
In the above expression $\I:\so(n)\to \so(n)$ is the inertia tensor and $( \cdot, \cdot)_\kappa$ is the Killing metric in $\so(n)$:
\begin{equation*}
(\xi, \eta)_\kappa =-\frac{1}{2}\tr(\xi \eta).
\end{equation*}
We will assume that the potential energy $U:\SO(n)\to \R$ is invariant under the $\SO(n-1)$ action indicated below. 

Following~\cite{FedJov, FedJov2}, we assume that there exists
 a diagonal matrix $A=\mbox{diag}(a_1, \dots , a_n)$, with positive entries, such that the inertia tensor satisfies:
\begin{equation}
\label{eq:inertia-tensor}
\I(u\wedge v)= (Au)\wedge (A v), 
\end{equation}
for $u,v\in \R^n$, where $u\wedge v =uv^T-vu^T$.

\begin{remark}
The  Hamiltonisation of several other multi-dimensional nonholonomic systems relies on the assumption that the
inertia tensor satisfies~\eqref{eq:inertia-tensor}~\cite{JovaRubber,Jovan,Jova18,Gajic}. Interestingly, this condition always holds if 
$n=3$, but, for $n\geq 4$, it is 
 generally inconsistent
with the standard `physical' considerations of multi-dimensional rigid body dynamics (see the discussion in~\cite{FGNM18}). 
\end{remark}

The nonholonomic constraints are simpler to write  in terms of the angular velocity as seen in  the {\em space} frame:
\begin{equation}
\omega =\Ad_g \Omega =\dot g g^{-1}\in \so(n).
\end{equation}
The constraints require that the  following entries of $\omega$  vanish during the motion: 
\begin{equation}
\label{eq:Ves-constraints}
\omega_{ij}=0, \qquad 1\leq i,j\leq n-1.
\end{equation}

Denote by  $e_n:=(0,\dots, 0,1)\in \R^n$. As was first explained  in~\cite{FedJov},  the problem is an $\SO(n-1)$-Chaplygin system with $r=n-1$ degrees of freedom, where 
\begin{equation*}
\SO(n-1)=\{ h\in \SO(n) \, : \, h^{-1}e_n =e_n \},
\end{equation*}
 acts on $Q=\SO(n)$ by left multiplication\footnote{we need to assume that the potential $U$ is invariant under this action.}. 
The shape space $S=\SO(n)/\SO(n-1)=\Ss^{n-1}$ and the corresponding bundle map is
\begin{equation}
\pi: \SO(n) \to \Ss^{n-1}, \qquad \pi(g)=\gamma(g):=g^{-1}e_n,
\end{equation}
where $\Ss^{n-1}$ is realised as 
\begin{equation*}
\Ss^{n-1}=\{ \gamma =(\gamma_1, \dots, \gamma_n) \in \R^n \, : \,\gamma_1^2+\dots +\gamma_n^2=1 \}.
\end{equation*}

The rest of this section is dedicated to the proof of the following theorem. In its statement, and throughout, $(\cdot , \cdot)$ denotes the Euclidean scalar product in $\R^n$.
\begin{theorem}
\label{T:Veselova-phi-simple}
The multidimensional Veselova system with special inertia tensor~\eqref{eq:inertia-tensor} is $\phi$-simple with 
$\phi: \Ss^{n-1}\to \R$ given by $\phi(\gamma)=-\frac{1}{2} \ln (A\gamma,\gamma)$. 
\end{theorem}

 It follows from 
Theorem~\ref{T:measure-gyro} and Corollary~\ref{l:Hamiltonisation-geometric} that the reduced equations of motion on $T^*\Ss^{n-1}$ preserve the measure
$\mu=  (A\gamma,\gamma)^{(2-n)/2} \, \nu$, and become Hamiltonian after the time reparametrisation $dt =(A\gamma,\gamma)^{1/2}\, d\tau$. 
This recovers  the results of Fedorov and Jovanovi\'c~\cite{FedJov, FedJov2} on the Hamiltonisation of the problem.

We will prove Theorem~\ref{T:Veselova-phi-simple} by computing the gyroscopic tensor in local coordinates in $\Ss^{n-1}$. Namely, let  $(s_1, \dots, s_{n-1})$ 
be the coordinates on the northern hemisphere $\Ss_+^{n-1}:=\{(\gamma_1, \dots, \gamma_n)\in \Ss^n \, : \, \gamma_n>0\}$ given by:
\begin{equation}
\gamma_1=s_1,\quad \dots \quad , \gamma_{n-1}=s_{n-1}, \quad  \gamma_n=\sqrt{1-s_1^2-\dots -s_{n-1}^2}. 
\end{equation}
In terms of the canonical vectors $e_1, \dots, e_n$ in the ambient space $\R^n$, we have
\begin{equation}
\label{eq:ds_i-example}
\frac{\partial}{\partial s_i}=e_i -\frac{\gamma_i}{\gamma_n} e_n, \qquad i=1, \dots, n-1.
\end{equation}
 For the rest of the section, we identify $T_g\SO(n)=\so(n)$ via the left trivialisation. The following proposition gives the form of the horizontal lift of the 
coordinate vector fields. To simplify notation, for the rest of the section we denote $X_i(g):=\mbox{ hor}_g\left ( \frac{\partial}{\partial s_i} \right )\in \so(n)$.
\begin{proposition}
Let $\gamma\in \Ss_+^{n-1}$ and $g\in \pi^{-1}(\gamma)$, (i.e. $\gamma=g^{-1}e_n$). The horizontal lift
\begin{equation}
\label{eq:defXi}
\mbox{\em hor}_g\left ( \frac{\partial}{\partial s_i} \right ) =:X_i(g) = \gamma\wedge \left ( e_i  -\frac{\gamma_i}{\gamma_n}e_n \right ), \qquad i=1,\dots, n-1.
\end{equation}
\end{proposition}
\begin{proof}
The nonholonomic  constraints~\eqref{eq:Ves-constraints} imply $\omega=e_n\wedge a$ for a vector $a\in \R^n$ that may be assumed to be
perpendicular to $e_n$. 
Hence, $$\Omega=\mbox{Ad}_{g^{-1}}(e_n\wedge a)=(g^{-1}e_n)\wedge ( g^{-1}a)=\gamma\wedge v,$$
where $v=g^{-1}a$ is perpendicular to $\gamma$. On the other hand, differentiating $\gamma=g^{-1}e_n$ gives $\dot \gamma =-\Omega \gamma$.
Whence,  $\dot \gamma = -(\gamma \wedge v)\gamma=v$ and we conclude that $\Omega=\gamma \wedge \dot \gamma$.
This implies that the horizontal lift of the tangent vector $\dot \gamma \in T_\gamma \Ss^{n-1}$ to $T_g\SO(n)$, with $\gamma=g^{-1}e_n$, is the vector 
$\Omega = \gamma\wedge \dot \gamma$ (in the left trivialisation). The result then follows from~\eqref{eq:ds_i-example}.
\end{proof}

The Lie brackets of the vector fields $X_i$ defined by~\eqref{eq:defXi} may be computed directly. To simplify the reading of this section, 
the details of the calculation are  outlined in the proof of the following lemma given in  
 Appendix~\ref{app:lemmaA1}.
\begin{lemma} 
\label{l:commXi}
We have
\begin{equation*}
[X_i,X_j]= \left ( e_i  -\frac{\gamma_i}{\gamma_n}e_n \right ) \wedge  \left ( e_j  -\frac{\gamma_j}{\gamma_n}e_n \right ), \qquad i,j=1,\dots, n-1.
\end{equation*}
\end{lemma}

According to Definition~\ref{gyroscopic-tensor} of the gyroscopic tensor, in order to compute $\mathcal{T}$ on the coordinate vector fields,
we should compute the orthogonal projection $\mathcal{P}[X_i,X_j]$  onto the constraint distribution with respect to the kinetic energy metric $\llangle \cdot , \cdot \rrangle$.
Note that, in view of the kinetic energy of the Lagrangian~\eqref{eq:Lagrangian-Veselova}, for $Y_1(g), Y_2(g)\in \so(n)$,  vector fields on $\SO(n)$, we have
\begin{equation*}
\llangle Y_1, Y_2 \rrangle_g=(\I (Y_1(g)),Y_2(g))_\kappa .
\end{equation*}

The proof of the following lemma is also a calculation that is postponed to Appendix~~\ref{app:lemmaA2}. Note that the formulae given below involve the 
entries of the matrix $A$, so we are relying on the crucial assumption~\eqref{eq:inertia-tensor} on the inertia tensor.
\begin{lemma}
\label{l:MetricXi}
For $i,j,k,l\in \{1, \dots, n-1\}$ we have
\begin{equation}
\label{eq:Kij}
\llangle X_k, X_l\rrangle =(A\gamma,\gamma) \left (a_l\delta_{kl} + \frac{a_n \gamma_k\gamma_l}{\gamma_n^2} \right ) -\gamma_k\gamma_l(a_n-a_k)(a_n-a_l),
\end{equation}
and
\begin{equation}
\label{eq:XiXjXl}
\llangle [X_i, X_j ] , X_l \rrangle  = \frac{a_n\gamma_i\gamma_j\gamma_l}{\gamma_n^2}(a_i-a_j)+a_j\gamma_i(a_i-a_n)\delta_{jl}-a_i\gamma_j(a_j-a_n)\delta_{il},
\end{equation}
where $\delta_{ij}$ is the Kronecker delta.
\end{lemma}

The following lemma gives an explicit expression for the gyroscopic coefficients $C_{ij}^k$ written in our coordinates. Its proof relies on the previous lemma.
\begin{lemma}
\label{L:Gyro-coeff-Veselova}
For $i,j,k\in \{1, \dots, n-1\}$ we have
\begin{equation}
\label{eq:Cijk}
C_{ij}^k  = \frac{-\gamma_j(a_j-a_n)\delta_{ik}+\gamma_i(a_i-a_n) \delta_{jk}}{(A\gamma,\gamma)}.
\end{equation}
\end{lemma}

\begin{proof} Using $\llangle \mathcal{P} [X_i, X_j ] , X_l \rrangle = \llangle [X_i, X_j ] , X_l \rrangle$, Equation~\eqref{eq:gyrosccoeff}, and the invertibility
of the matrix with coefficients $K_{kl}:=\llangle X_k, X_l\rrangle$, it follows that the gyroscopic coefficients $C_{ij}^k$ are characterised by the relations
\begin{equation}
\label{eq:proofeqn}
\sum_{k=1}^{n-1}C_{ij}^k\llangle X_k, X_l\rrangle=\llangle [X_i, X_j ] , X_l \rrangle, \qquad i,j,k,l\in \{1, \dots, n-1\}.
\end{equation}
We shall prove that the coefficients given by~\eqref{eq:Cijk} satisfy these equations. Starting with~\eqref{eq:Cijk} we compute:
\begin{equation}
\label{eq:auxlemmaCijk1}
\begin{split}
&\sum_{k=1}^{n-1}C_{ij}^k \gamma_l\gamma_k(a_n-a_l)(a_n-a_k) \\ & \qquad \qquad  =
\frac{\gamma_l(a_n-a_l)}{(A\gamma,\gamma)} \sum_{k=1}^{n-1} \left (-\gamma_j(a_j-a_n)\delta_{ik}+\gamma_i(a_i-a_n) \delta_{jk} \right )\gamma_k(a_n-a_k) \\
& \qquad \qquad  =\frac{\gamma_l(a_n-a_l)}{(A\gamma,\gamma)}\left (-\gamma_i\gamma_j(a_j-a_n)(a_n-a_i)+\gamma_i\gamma_j(a_i-a_n)(a_n-a_j) \right )=0.
\end{split}
\end{equation}
Similarly,
\begin{equation}
\label{eq:auxlemmaCijk2}
\begin{split}
\sum_{k=1}^{n-1}C_{ij}^k (A\gamma,\gamma)a_l\delta_{kl} &   = \sum_{k=1}^{n-1} \left (-\gamma_j(a_j-a_n)\delta_{ik}+\gamma_i(a_i-a_n) \delta_{jk} \right )a_l\delta_{kl} \\
& =-a_l\gamma_j(a_j-a_n)\delta_{il}+a_l\gamma_i(a_i-a_n)\delta_{jl}.
\end{split}
\end{equation}
And also,
\begin{equation}
\label{eq:auxlemmaCijk3}
\begin{split}
\sum_{k=1}^{n-1}C_{ij}^k (A\gamma,\gamma) \left ( \frac{a_n\gamma_l\gamma_k}{\gamma_n^2} \right ) &   =  \frac{a_n\gamma_l}{\gamma_n^2} \sum_{k=1}^{n-1} \left (-\gamma_j(a_j-a_n)\delta_{ik}+\gamma_i(a_i-a_n) \delta_{jk} \right )\gamma_k \\
& = \frac{a_n\gamma_i\gamma_j\gamma_l}{\gamma_n^2}(a_i-a_j).
\end{split}
\end{equation}
Combining  the expressions \eqref{eq:Kij} for $\llangle X_k, X_l \rrangle$  and \eqref{eq:XiXjXl} for $\llangle [X_i, X_j ] , X_l \rrangle$,
with the Equations \eqref{eq:auxlemmaCijk1}, \eqref{eq:auxlemmaCijk2} and \eqref{eq:auxlemmaCijk3} obtained above, shows that~\eqref{eq:proofeqn} holds.
\end{proof}

We are now ready to present:
\begin{proof}[Proof of Theorem~\ref{T:Veselova-phi-simple}]
Lemma~\ref{L:Gyro-coeff-Veselova} implies
\begin{equation}
\label{eq:auxProofVesGyroT}
\mathcal{T} \left (\frac{\partial}{\partial s^i} ,\frac{\partial}{\partial s^j} \right ) = \frac{-\gamma_j(a_j-a_n)}{(A\gamma,\gamma)}\frac{\partial}{\partial s^i} +
\frac{\gamma_i(a_i-a_n)}{(A\gamma,\gamma)}\frac{\partial}{\partial s^j}, \qquad 1\leq i,j\leq n-1.
\end{equation}
Considering that
\begin{equation*}
\begin{split}
(A\gamma,\gamma)&=a_1s_1^2 +\dots + a_{n-1}s_{n-1}^2+a_n(1-s_1^2-\dots - s_{n-1}^2) \\ &=a_n +(a_1-a_n)s_1^2 + \dots + (a_{n-1}-a_n)s_{n-1}^2,
\end{split}
\end{equation*}
we have
\begin{equation*}
\frac{\partial}{\partial s^k}   \left ( -\frac{1}{2} \ln (A\gamma,\gamma)\right )= \frac{-s_k(a_k-a_n)}{(A\gamma,\gamma)} = \frac{-\gamma_k(a_k- a_n)}{(A\gamma,\gamma)},
\qquad \mbox{for  all $1\leq k \leq n-1$.}
\end{equation*}
 Therefore, Equation~\eqref{eq:auxProofVesGyroT} may be rewritten as
 \begin{equation*}
\label{eq:gyro-T-Veselova}
\mathcal{T} \left (\frac{\partial}{\partial s^i} ,\frac{\partial}{\partial s^j} \right ) = \frac{\partial \phi}{\partial s^j}\frac{\partial}{\partial s^i} -
\frac{\partial \phi}{\partial s^i}\frac{\partial}{\partial s^j}, \qquad 1\leq i,j\leq n-1,
\end{equation*}
with $\phi(\gamma)=-\frac{1}{2} \ln (A\gamma,\gamma)$. The proof of Theorem~\ref{T:Veselova-phi-simple} follows easily from the 
above expression and  the tensorial properties of $\mathcal{T}$, and the fact that $\Ss^{n-1}$ may be covered with coordinate charts on its 
different hemispheres,
similar to the one that we have considered for $\Ss^{n-1}_+$, and formulae analogous to~\eqref{eq:auxProofVesGyroT} hold on
 each of them. 
\end{proof}

\section{Future work}
\label{S:FutureWork}

This paper shows that every  $\phi$-simple Chaplygin systems admits a Hamiltonisation, and
that a key example like the multidimensional  Veselova problem is $\phi$-simple. 
However,  in virtue of Remark \ref{rmk:Ham-with-Omega0}, being $\phi$-simple is not 
a necessary condition for Hamiltonisation. Therefore the question
 remains open to find examples of Hamiltonisable Chaplygin systems which are not $\phi$-simple.
Also, it would be desirable  to give a geometric  characterisation, in terms of the kinetic energy metric and the constraint distribution, of the necessary conditions for a Chaplygin system to admit a Hamiltonisation.

\appendix

\section{Geometric preliminaries}
\label{App:Intrinsic}

\subsection*{Linear functions induced by sections of a vector bundle}
Let $M$ be a manifold and $\tau: E\to M$  a vector bundle. For a section $Y\in \Gamma(E)$ we denote by $Y^\ell\in C^\infty(E^*)$ the 
function defined by
\begin{equation}
\label{eq:linear-function}
Y^\ell (\alpha)= \alpha (Y(\tau (\alpha))).
\end{equation}

\subsection*{Metric musical isomorphisms}

Let $(M, \llangle \cdot \, ,  \, \cdot   \rrangle)$ be a Riemannian manifold. For $\beta \in \Omega^1(M)$ and   $Y\in \mathfrak{X}(M)$  we define the  \defn{metric duals}   $\beta^{\sharp}  \in {\frak X}(M)$
and $Y^\flat \in \Omega^1(M)$ by the conditions
\begin{equation}\label{eq:metric-dual}
\beta(X) = \llangle \beta^{\sharp}, X \rrangle, \quad Y^\flat (X) = \llangle Y, X \rrangle, \quad \mbox{ for } X\in {\frak X}(M).
\end{equation}
The   isomorphisms  $\beta \mapsto \beta^\sharp$ and $Y\mapsto Y^\flat$  at the point $m\in M$ will be correspondingly denoted  by 
\begin{equation}
\label{eq:musical-isomorphisms}
{\sharp}_m:T_m^*M\to T_mM, \qquad {\flat}_m:T_mM\to T_m^*M.
\end{equation}

\subsection*{ Vertical and complete lifts of  $1$-forms and vector fields on a  cotangent bundle}

Let  $\tau_S:T^*S \to S$ denote the canonical projection.  
If $\gamma\in \Omega^1(S)$ then the \defn{vertical lift} $\gamma^{\bf v}$ is the vertical vector field on $T^*S$ defined by
\begin{equation}\label{vertical-lift-dual}
\gamma^{\bf v}(\alpha) =\left . \frac{d}{dt} \right |_{t=0}(\alpha + t \gamma(\tau_S(\alpha))), \; \; \mbox{ for } \alpha \in T^*S.
\end{equation}

On the other hand,  if $Y\in \mathfrak{X}(S)$, then it induces the linear function $Y^\ell$ on $T^*S$ (defined by \eqref{eq:linear-function}) whose Hamiltonian vector field
with respect to $\Omega_{can}$ is the  \defn{complete lift} of $Y$ denoted by $Y^{*c}\in \frak{X}(T^*S)$. Namely,   $Y^{*c}$
is characterised by 
\begin{equation}\label{complete-lift-dual}
{\bf i}_{Y^{*c}}\Omega_{can} = dY^\ell.
\end{equation}

Let  $(s^1, \dots, s^r, p_1, \dots p_r)$ denote fibered coordinates of  $T^*S$
such that the canonical 2-form  has local expression $\Omega_{can}= \sum_{i=1}^rds^i \wedge dp_i$. If  $\gamma$  and  $Y$ are locally given by
\[
 \gamma =\sum_{i=1}^r \gamma_i \, ds^{i}, \qquad  Y = \sum_{i=1}^rY^i\frac{\partial}{\partial s^{i}},  
\]
then, the  definitions given above  imply
\begin{equation}\label{vertical-complete-lift-dual}
\gamma^{\bf v} =\sum_{i=1}^j \gamma_i\frac{\partial}{\partial p_{i}}, \qquad Y^{*c} =\sum_{i=1}^r\left ( 
 Y^i\frac{\partial}{\partial s^{i}} -\sum_{j=1}^r \frac{\partial Y^{i}}{\partial s^{j}}p_j\frac{\partial}{\partial p_i} \right ).
\end{equation}
 In particular we have  
\begin{equation}\label{contraction-vertical-Omega-can}
{\bf i}_{\gamma^{\bf v}}\Omega_{can} = -\tau_S^*(\gamma).
\end{equation}

For reference, we also 
note that the Lie brackets of the previous vector fields are 
\begin{equation}\label{Lie-bracket-complete-vertical}
[Y^{*c}, Z^{*c}] = [Y, Z]^{*c}, \qquad  [Y^{*c}, \gamma^{\bf v}] = ({\mathcal L}_Y\gamma)^{\bf v}, \qquad [\gamma^{\bf v}, \beta^{\bf v}] = 0,
\end{equation}
where $Y, Z\in \frak{X}(S)$ and $\gamma, \beta \in \Omega^1(S)$.

\subsection*{Quadratic function associated to a $(2, 0)$ tensor field}

Let $\Xi$ be a tensor field on $S$ of type $(2, 0)$. We denote by $\Xi^{\mathfrak{q}}$ the quadratic function  on $T^*S$ defined by
\begin{equation}
\label{eq:quad-function}
\Xi^{\mathfrak{q}}(\alpha) = \Xi(\alpha, \alpha).
\end{equation}
Locally, if 
\[
\Xi =\sum_{i,j=1}^r \Xi^{ij}(s) \frac{\partial}{\partial s^{i}}\otimes\frac{\partial}{\partial s^{j}} \qquad \mbox{ then } \qquad \Xi^{\mathfrak{q}} =\sum_{j,k=1}^r \Xi^{jk}(s)p_j p_k.
\]
\section{Proofs of Lemmas~\ref{l:commXi} and~\ref{l:MetricXi}.}
\label{app:proofs-lemmas}

\subsection{Proof of Lemma~\ref{l:commXi}.}
\label{app:lemmaA1}

The proof is a long calculation. We present some preliminary lemmas that contain intermediate steps of it. Recall that we work with the
left trivialisation of $T\SO(n)$ so vector fields $Y\in \frak{X}(\SO(n))$ are interpreted as functions $Y:\SO(n)\to \so(n)$.

\begin{lemma}
\label{l:left-eij}
Consider the left invariant vector field $e_i\wedge e_j$ on $\SO(n)$ and denote by $g_{jk}:\SO(n)\to \R$ the function
that returns the $j$-$k$ entry of $g\in \SO(n)$. We have
\begin{equation*}
e_i\wedge e_j\, [g_{kl}] =g_{ki}\delta_{jl} - g_{kj}\delta_{il}.
\end{equation*}
\end{lemma}
\begin{proof}
\begin{equation*}
\begin{split}
e_i\wedge e_j\, [g_{kl}] & = \left . \frac{d}{dt} \right |_{t=0}\left ( g \exp \left  ( t(e_i\wedge e_j) \right )  \right )_{kl} = 
\left (g (e_i\wedge e_j)  \right )_{kl}  \\
&= \left ( (g_{1i}, \dots , g_{ni} )^Te_j^T - (g_{1j}, \dots , g_{nj} )^Te_i^T \right )_{kl} 
=  \begin{array}{c}  \begin{array}{ccccc} &  i  \; \; \; \; \; \;\;  & & j&  \end{array}  \\ \left ( \begin{array}{ccccc} & -g_{1j} & & g_{1i} & \\ 0 & \vdots & 0 & \vdots & 0 \\ & -g_{nj} & & g_{ni} &    \end{array} \right )_{kl} \end{array}.
\end{split}
\end{equation*}
\end{proof}

Now recall that  $\gamma=g^{-1}e_n$ so that $\gamma_i=g_{ni}$ for $i=1,\dots, n$.

\begin{lemma}\label{l:left-q-eij}
Consider the $\SO(n-1)$-equivariant vector field $\gamma \wedge e_j$ on $\SO(n)$. We have
\begin{equation*}
\gamma \wedge e_j\, [\gamma_k] =  \delta_{jk}-\gamma_j\gamma_k.
\end{equation*}
\end{lemma}
\begin{proof}
If $j\neq k$ then, writing $\gamma=\sum_{i=1}^n \gamma_i \, e_i$, we have
\begin{equation*}
\begin{split}
\gamma \wedge e_j \, [\gamma_k]& = \sum_{i=1}^n \gamma_i \, e_i \wedge e_j \, [g_{nk}]=\sum_{i=1}^n  \gamma_i (-\delta_{ik} g_{nj}) =-\gamma_k\gamma_j,
 \end{split}
\end{equation*}
where we have used Lemma~\ref{l:left-eij}.
Using again  Lemma~\ref{l:left-eij},
\begin{equation*}
\begin{split}
\gamma \wedge e_k \, [\gamma_k]& = \sum_{i=1}^n \gamma_i \, e_i \wedge e_k\, [g_{nk}]=\sum_{i=1}^n  \gamma_i(g_{ni} - \delta_{ik} g_{nk})=
\left ( \sum_{i=1}^n 
 \gamma_{i}^2 \right )-\gamma_k^2=
1-\gamma_k^2.
 \end{split}
\end{equation*}
\end{proof}

\begin{lemma}
\label{l:diff-quotient}
Consider the $\SO(n-1)$-equivariant vector field $\gamma \wedge e_j$ on $\SO(n)$.  Along the open subset of $\SO(n)$ where $\gamma_n\neq 0$, we have
\begin{equation*}
\gamma \wedge e_j\, \left [\frac{\gamma_k}{\gamma_n} \right ] =\begin{cases}  0 \quad \mbox{if} \quad j\neq k \neq n \neq j, \\ \frac{1}{\gamma_n} \quad  \mbox{if} \quad j=k \neq n, \\
-\frac{\gamma_k}{\gamma_n^2}  \quad \mbox{if} \quad j = n \neq k.
\end{cases}
\end{equation*}
\end{lemma}
\begin{proof}
If $ j\neq k \neq n \neq j$ then 
\begin{equation*}
\begin{split}
\gamma \wedge e_j \, \left [\frac{\gamma_k}{\gamma_n} \right ]& = \frac{1}{\gamma_n} \gamma \wedge e_j\, \left [ \gamma_k \right ] -  \frac{\gamma_k}{\gamma_n^2}
 \gamma \wedge e_j\, \left [ \gamma_n \right ] = -\frac{\gamma_k\gamma_j}{\gamma_n}-  \frac{\gamma_k}{\gamma_n^2}(-\gamma_j\gamma_n)=0,
 \end{split}
\end{equation*}
where we have used Lemma~\ref{l:left-q-eij}. Similarly, using again  Lemma~\ref{l:left-q-eij} and assuming $k \neq n$,
\begin{equation*}
\begin{split}
\gamma \wedge e_k \, \left [\frac{\gamma_k}{\gamma_n} \right ]& = \frac{1}{\gamma_n} \gamma \wedge e_k \, \left [ \gamma_k \right ] -  \frac{\gamma_k}{\gamma_n^2}
 \gamma \wedge e_k\, \left [ \gamma_n \right ] = \frac{1}{\gamma_n}(1-\gamma_k^2)-  \frac{\gamma_k}{\gamma_n^2}(-\gamma_k\gamma_n) =  \frac{1}{\gamma_n}.
  \end{split}
\end{equation*}
Finally, assuming again $k \neq n$, and using once more Lemma~\ref{l:left-q-eij},
\begin{equation*}
\begin{split}
\gamma \wedge e_n \, \left [\frac{\gamma_k}{\gamma_n} \right ]& = \frac{1}{\gamma_n} \gamma \wedge e_n \, \left [ \gamma_k \right ] -  \frac{\gamma_k}{\gamma_n^2}
 \gamma \wedge e_n \, \left [ \gamma_n \right ] = \frac{1}{\gamma_n}(-\gamma_k\gamma_n)-  \frac{\gamma_k}{\gamma_n^2}(1-\gamma_n^2) =  -  \frac{\gamma_k}{\gamma_n^2}.
  \end{split}
\end{equation*}

\end{proof}

\begin{lemma} 
\label{l:comm-aux}
Consider the $\SO(n-1)$-equivariant vector fields $\gamma \wedge e_i$ and $\gamma \wedge e_j$ on $\SO(n)$. We have
\begin{equation*}
\left [ \, \gamma \wedge e_i \, , \, \gamma \wedge e_j   \, \right ] = e_i \wedge e_j.
\end{equation*}
\end{lemma}
\begin{proof} This is obvious if $i=j$, so assume $i\neq j$. We have
\begin{equation}
\begin{split}
\label{eq:AuxLemmaComm1}
 \left [ \, \gamma \wedge e_i \, , \, \gamma \wedge e_j   \, \right ] &=  \left [ \, \sum_{k=1}^n \gamma_k \, e_k \wedge e_i  \, , \, \sum_{l=1}^n
  \gamma_l  \, e_l \wedge e_j    \, \right ] =  \sum_{k,l=1}^n \left (\gamma_k\gamma_l \left [ \, e_k \wedge e_i  \, , \,  e_l \wedge e_j  \, \right ] \right ) 
    \\
  & \qquad \qquad\qquad +  \sum_{l=1}^n \left ( 
 \gamma\wedge e_i  \left [ \gamma_l  \right ]  \right ) \,  e_l \wedge e_j  -\sum_{k=1}^n (  \gamma \wedge e_j 
\left [ \gamma_k  \right ] )\, e_k \wedge e_i .
   \end{split}
\end{equation}
Using Lemma~\ref{l:left-q-eij}  we have 
\begin{equation}
\label{eq:AuxLemmaComm2}
\begin{split}
\sum_{l=1}^n \left ( 
\gamma\wedge e_i  \left [ \gamma_l  \right ]  \right )  \, e_l \wedge e_j   & =
  \sum_{l=1}^n \left ( \delta_{il}
-\gamma_i\gamma_l \right )  \, e_l \wedge e_j  =  e_i \wedge e_j -\gamma_i \, \gamma \wedge e_j ,
  \end{split}
\end{equation}
and by the same reasoning
\begin{equation}
\label{eq:AuxLemmaComm3}
\begin{split}
\sum_{k=1}^n (  \gamma \wedge e_j 
\left [ \gamma_k  \right ] ) \, e_k \wedge e_i    = e_j \wedge e_i-\gamma_j  \, \gamma \wedge e_i.
  \end{split}
\end{equation}

On the other hand, using that the Lie bracket of left invariant vector fields is determined by the Lie bracket of their generators in the Lie algebra, and
 computing the matrix commutator gives:
\begin{equation*}
\left [ \, e_k \wedge e_i  \, , \,  e_l \wedge e_j    \, \right ] =\delta_{kl}  e_j \wedge e_i +\delta_{kj}  e_i \wedge e_l
+\delta_{li} e_k \wedge e_j,
\end{equation*}
where we have used our assumption that $i\neq j$. So we can simplify
\begin{equation}
\label{eq:AuxLemmaComm4}
\begin{split}
  \sum_{k,l=1}^n \left (\gamma_k\gamma_l \left [ \, e_k \wedge e_i  \, , \,  e_l \wedge e_j  \, \right ] \right )
 & = e_j \wedge e_i \sum_{k=1}^n\gamma_k^2 +\gamma_j \sum_{l=1}^n\gamma_l \, e_i \wedge e_l   +\gamma_i  \sum_{k=1}^n\gamma_k \, e_k \wedge e_j  \\
 &=  e_j \wedge e_i + \gamma_j \, e_i \wedge \gamma +\gamma_i \, \gamma \wedge e_j. 
   \end{split}
\end{equation}
Substituting \eqref{eq:AuxLemmaComm2}, \eqref{eq:AuxLemmaComm3} and \eqref{eq:AuxLemmaComm4} into \eqref{eq:AuxLemmaComm1} proves the 
result.

\end{proof}

We are finally ready to present a proof of Lemma~\ref{l:commXi}.
\begin{proof}[Proof of Lemma~\ref{l:commXi}]
 The result is obvious if $i=j$ so we assume $i\neq j$. Using the standard properties of Lie brackets of vector fields we have:
\begin{equation*}
\begin{split}
[X_i,X_j]
 & =[ \gamma\wedge e_i ,\gamma\wedge e_j] -\frac{\gamma_j}{\gamma_n} [ \gamma\wedge e_i ,\gamma\wedge e_n] -\frac{\gamma_i}{\gamma_n} [ \gamma\wedge e_n ,\gamma\wedge e_j] \\
& \qquad + \left (  -\gamma\wedge e_i \left [  \frac{\gamma_j}{\gamma_n} \right ] +\gamma\wedge e_j \left [  \frac{\gamma_i}{\gamma_n} \right ] + \frac{\gamma_i}{\gamma_n} \gamma\wedge e_n \left [  \frac{\gamma_j}{\gamma_n} \right ] 
- \frac{\gamma_j}{\gamma_n} \gamma\wedge e_n \left [  \frac{\gamma_i}{\gamma_n} \right ]\right )\, \gamma\wedge e_n.
    \end{split}
\end{equation*}
In view of Lemmas~\ref{l:diff-quotient} and~\ref{l:comm-aux}, this simplifies to 
\begin{equation*}
\begin{split}
[X_i,X_j]
 & = e_i \wedge e_j -\frac{\gamma_j}{\gamma_n}  e_i \wedge e_n -\frac{\gamma_i}{\gamma_n}  e_n \wedge e_j 
 = \left ( e_i  -\frac{\gamma_i}{\gamma_n}e_n \right ) \wedge  \left ( e_j  -\frac{\gamma_j}{\gamma_n}e_n \right ).    \end{split}
\end{equation*}
\end{proof}

\subsection{Proof of Lemma~\ref{l:MetricXi}.}
\label{app:lemmaA2}
For the  calculations in this section, we use the following:
\begin{lemma}
\label{l:ComputeKilling}
If $u_1,v_1,u_2,v_2\in \R^n$ then
\begin{equation*}
(u_1\wedge v_1,u_2\wedge v_2)_\kappa= (u_1,u_2)(v_1,v_2)-(u_1,v_2)(u_2,v_1).
\end{equation*}
\end{lemma}
\begin{proof}
It is an elementary  calculation that uses the properties of the trace.
\end{proof}

We begin by proving~\eqref{eq:Kij}. 
Using the crucial hypothesis \eqref{eq:inertia-tensor} on the inertia tensor we have:
\begin{equation*}
\begin{split}
\llangle X_k , X_l\rrangle &= \left ( \I \left ( \gamma\wedge \left ( e_k -\frac{\gamma_k}{\gamma_n} e_n \right ) \right ) \, , \, \gamma\wedge \left ( e_l -\frac{\gamma_l}{\gamma_n} e_n \right ) \right )_\kappa  \\
&=  \left ( \left ( A\gamma\wedge \left ( a_ke_k -\frac{a_n\gamma_k}{\gamma_n} e_n \right ) \right ) \, , \, \gamma\wedge \left ( e_l -\frac{\gamma_l}{\gamma_n} e_n \right ) \right )_\kappa.
\end{split}
\end{equation*}
Using Lemma~\ref{l:ComputeKilling} this simplifies to 
\begin{equation*}
\llangle X_k , X_l\rrangle=(A\gamma,\gamma)\left ( a_k \delta_{kl}+\frac{a_n\gamma_k\gamma_l}{\gamma_n^2} \right ) -(a_l\gamma_l  - a_n\gamma_l) (a_k\gamma_k - a_n\gamma_k),
\end{equation*}
which is equivalent to \eqref{eq:Kij}.

Next, using Lemma~\ref{l:commXi} and the hypothesis \eqref{eq:inertia-tensor} on the inertia tensor, we have:
\begin{equation*}
\begin{split}
\llangle [X_i, X_j ] , X_l \rrangle &= \left ( \I \left (  \left ( e_i  -\frac{\gamma_i}{\gamma_n}e_n \right ) \wedge  \left ( e_j  -\frac{\gamma_j}{\gamma_n}e_n \right ) \right )
 \, , \, \gamma\wedge \left ( e_l -\frac{\gamma_l}{\gamma_n} e_n \right ) \right )_\kappa  \\
&=  \left (  \left ( a_ie_i  -\frac{a_n\gamma_i}{\gamma_n}e_n \right ) \wedge  \left ( a_je_j  -\frac{a_n\gamma_j}{\gamma_n}e_n \right ) 
 \, , \, \gamma\wedge \left ( e_l -\frac{\gamma_l}{\gamma_n} e_n \right ) \right )_\kappa.
\end{split}
\end{equation*}
Using Lemma~\ref{l:ComputeKilling} this simplifies to 
\begin{equation*}
\llangle [X_i, X_j ] , X_l \rrangle =(a_i\gamma_i -a_n\gamma_i)\left ( a_j \delta_{jl} + \frac{a_n\gamma_j\gamma_l}{\gamma_n^2} \right ) - (a_j\gamma_j -a_n\gamma_j)\left ( a_i \delta_{il} + \frac{a_n\gamma_i\gamma_l}{\gamma_n^2} \right ) ,
\end{equation*}
which upon rearrangement is equivalent to~\eqref{eq:XiXjXl}.

\vspace{1cm}

\noindent \textbf{Acknowledgements:} We are thankful to the anonymous referees for their
suggestions that helped us to improve this paper. LGN acknowledges the Alexander von Humboldt Foundation  for a  Georg Forster Experienced Researcher Fellowship   
  that funded a research visit to TU Berlin where part of this work was done. JCM acknowledges the partial support by European Union (Feder) grant MTM 2015-64166-C2-2P and PGC2018-098265-B-C32. The authors are thankful to R. Ch\'avez-Tovar for his help to produce Figure~\ref{F:Disk}.

\vskip 1cm

\noindent LGN: Departamento de Matem\'aticas y Mec\'anica, 
IIMAS-UNAM.
Apdo. Postal 20-126, Col. San \'Angel,
Mexico City, 01000,  Mexico. luis@mym.iimas.unam.mx. \\
 JCM: ULL-CSIC, Geometr{\'\i}a Diferencial y Mec\'anica Geom\'etrica, Departamento de Matem\'aticas, Estad\'istica e Investigaci\'on Operativa, 
 Facultad de Ciencias, Universidad de La Laguna, Tenerife, Spain. jcmarrer@ull.edu.es.
 

\begin{thebibliography}{99}
\let\\, \newcommand{\by}[1]{\textsc{\ignorespaces #1}\\}
  \newcommand{\title}[1]{\textsl{\ignorespaces #1}\\}
  \newcommand{\vol}[1]{{\bf{\ignorespaces #1}}}
  \newcommand{\info}[1]{\textrm{\ignorespaces #1}.}

\small  \setlength{\parskip}{0pt}
  

\bibitem{Baksa} \by{Bak{\v s}a, A.} On geometrization of motion of some nonholonomic systems, \emph{Mat. Vesnik} \vol{12} (1975),
233--244 (in Serbo-Croatian). English translation in \emph{Theor. Appl. Mech.} \vol{44} (2017), 133--139.


\bibitem{BalseiroGN} \by{Balseiro, P.\  and L.C.~Garc\'ia-Naranjo} 
Gauge transformations, twisted Poisson brackets and Hamiltonization of nonholonomic systems. 
    {\em  Arch. Rat. Mech. Anal.} {\bf 205} (2012), no. 1, 267--310.

\bibitem{BalseiroFdz} \by{Balseiro, P.\ and O.E.~Fernandez}  Reduction of nonholonomic systems in two stages and Hamiltonization.
{\em  Nonlinearity} {\bf 28} (2015) 2873. 

\bibitem{BS93}
	\by{Bates L., and J. Sniatycki}
	Nonholonomic reduction, \emph{Rep. Math. Phys.} \vol{32} (1993), 99--115.

\bibitem{Blackall}  \by{Blackall, C. J.} On volume integral invariants of non-holonomic dynamical systems. 
\emph{Amer. J. Math.} \vol{63} (1941), 155--168.

\bibitem{BKMM} \by{Bloch, A.M.,  P.S.~Krishnaprasad, J.E.~Marsden and R.M.~Murray} 
Nonholonomic mechanical systems with symmetry. \emph{Arch. Ration. Mech. Anal.} \vol{136} (1996), 21--99.

\bibitem{BlochBook}  \by{Bloch, A.M.}
\emph{Nonholonomic mechanics and control.}
$2^{nd}$ edition. Interdisciplinary Applied Mathematics, 24. Springer, New York, 2015. 

 \bibitem{BorMamChap} \by{Borisov A.~V. and   I.~S. Mamaev} Chaplygin's Ball Rolling Problem Is Hamiltonian. {\em Math. Notes}, (2001),
{\bf 70}, 793--795. 
 
  \bibitem{BolsBorMam2015}  \by{ Bolsinov, A. V.,  A. V. Borisov and  I. S. Mamaev} Geometrisation of Chaplygin's
 Reducing Multiplier theorem \emph{Nonlinearity}, \vol{28}  (2015),  2307--2318.
 

 \bibitem{BorMamHam}  \by{Borisov   A.V. and   I.S.~Mamaev}  Isomorphism and Hamilton Representation of Some Non-holonomic Systems, {\em Siberian Math. J.}, {\bf 48} (2007),  33--45 See also: arXiv: nlin.-SI/0509036 v. 1 (Sept. 21, 2005).
 
 


 \bibitem{CaCoLeMa} \by{Cantrijn F, Cort\'es J., de Le\'on M. and D. Mart\'in de Diego} On the geometry of generalized Chaplygin systems. \emph{Math. Proc.
 Cambridge Philos. Soc.} \vol{132} (2002),  323--351.



\bibitem{ChapRedMult} \by{Chaplygin, S.A.} On the theory of the motion of nonholonomic systems. The Reducing-Multiplier Theorem. \emph{Regul. Chaotic Dyn.} \vol{13}, 369--376 (2008) [Translated from \textit{Matematicheski\v{i} Sbornik} (Russian) \vol{28} (1911), by A. V. Getling]



\bibitem{EhlersKoiller} \by{Ehlers, K., J. Koiller, R. Montgomery  and P.M. Rios}
 Nonholonomic  Systems via Moving Frames: Cartan Equivalence and Chaplygin
Hamiltonization.  in \emph{The breath of Symplectic and Poisson Geometry},
Progress in Mathematics Vol.\ \vol{232} (2004), 75--120.


\bibitem{Fasso2} \by{Fass\`o, F., A.~Giacobbe and N.~Sansonetto}  
Gauge conservation laws and the momentum equation in nonholonomic mechanics. \emph{Rep.\ Math.\ Phys.} \vol{62} (2008), 345--367.


\bibitem{FassoJGM} \by{Fass\`o, F.  and N.~ Sansonetto}  
An elemental overview of the nonholonomic Noether theorem, \emph{ Int.\ J.\ Geom.\ Methods Mod.\ Phys.} \vol{6} (2009), 1343--1355.



\bibitem{Fasso3} \by{Fass\`o, F., Giacobbe~A. and N.~Sansonetto}  
Linear weakly Noetherian constants of motion are horizontal gauge momenta, 
\emph{J.\ Geom.\ Mech.} \vol{4} (2012), 129--136.





\bibitem{FGNM18} \by{Fass\`o, F.,  Garc{\'ia}-Naranjo L. C., and  Montaldi J.} Integrability and dynamics of the $n$-dimensional symmetric Veselova top. 
{\em J. Nonlinear Sci.}  \vol{29}, (2019)  1205--1246.


\bibitem{FedKoz} \by{Fedorov, Y. N., and V.V.~Kozlov} Various aspects of $n$-dimensional rigid body dynamics.
\emph{Amer.\ Math.\ Soc.\ Transl.\ (2)} \vol{168} (1995), 141--171.

\bibitem{FedJov} \by{Fedorov, Y. N. and B. Jovanovi\'c} Nonholonomic LR systems as generalized Chaplygin systems with an invariant measure and flows on homogeneous spaces. 
\emph{J.\ Nonlinear Sci.} \vol{14} (2004), 341--381. 

\bibitem{FedJov2} \by{Fedorov, Y. N. and B. Jovanovi\'c} Hamiltonization of the generalized Veselova LR system. 
\emph{Regul.\ Chaot.\ Dyn.} \vol{14} (2009), 495--505.

\bibitem{FedGNMa2015} \by{Fedorov Y.~N.,   Garc\'ia-Naranjo L.~C. and J.~C.~Marrero} Unimodularity and preservation of volumes in nonholonomic mechanics.
 {\em J. Nonlinear Sci.} \vol{25} (2015), 203--246.


\bibitem{Gajic} \by{Gaji\'c~B. and B.~Jovanovi\'c.} Nonholonomic connections, time reparametrizations,
and integrability of the rolling ball over a sphere.  {\em Nonlinearity} {\bf 32} (2019), 1675--1694.  

\bibitem{LGN10}  \by{Garc\'ia-Naranjo, L.C.} Reduction of almost Poisson brackets and Hamiltonization of the Chaplygin sphere. {\em Discrete Contin. Dyn. Syst. Ser. S}  {\bf 3} (2010), 37--60. 


\bibitem{LGN18}  \by{Garc\'ia-Naranjo, L.C.} 
Generalisation of Chaplygin's Reducing Multiplier  Theorem with an application to multi-dimensional nonholonomic
dynamics. {\em J. Phys. A: Math. Theor.}  {\bf 52} (2019) 205203 (16pp). 

\bibitem{LGNTAM-2019} \by{Garc\'ia-Naranjo L.~C.} {Hamiltonisation, measure preservation and first integrals of the 
multi-dimensional rubber Routh sphere.   {\em Theoretical and
Applied Mechanics}, {\bf 46} (2019) 65--88.}

\bibitem{Grab} \by{Grabowski, J., de Le\'on, M., Marrero, J. C. and D. Mart\'in de Diego} Nonholonomic constraints: a new viewpoint. {\em J. Math. Phys.} {\bf 50} (2009),  013520, 17 pp.





\bibitem{HochGN} \by{Hochgerner~S. and L.~C.Garc\'ia-Naranjo}  $G$-Chaplygin systems with internal symmetries, truncation, and an (almost) symplectic view of Chaplygin's ball. {\em J. Geom. Mech.} {\bf 1} (2009), 35--53.

\bibitem{Hoch}\by{Hochgerner~S.} Chaplygin systems associated to Cartan decompositions of semi-simple Lie groups. {\em Differential Geometry and its Applications} {\bf 28} (2010) 436--453



\bibitem{Ibort} \by{Ibort, A.,  de Le\'on M.,  Marrero J. C.   and D.  Mart\'in de Diego} Dirac Brackets in Constrained Dynamics, \emph{Fortschr. Phys.}  \vol{47}   (1999), 459--492.





\bibitem{Iliev1985} \by{Iliev, I.} 1985. On the conditions for the existence of the reducing Chaplygin factor. \emph{J. Appl.
Math. Mech.} \vol{49} (1985), 295--301.


\bibitem{JovaRubber} \by{Jovanovi\'c, B.}  LR and L+R systems. \emph{J. Phys. A}  {\bf 42} (2009),  18 pp.

\bibitem{Jovan} \by{Jovanovi\'c, B.} Hamiltonization and integrability of the Chaplygin sphere in $\R^n$. \emph{J. Nonlinear Sci.} \vol{20} (2010), 569--593.

\bibitem{Jova18} \by{Jovanovi\'c, B.} Rolling balls over spheres in $\R^n$. \emph{Nonlinearity}, \vol{31}  (2018),  4006--4031.

\bibitem{Jov2019}  \by{Jovanovi\'c, B.}  Note on a ball rolling over a sphere: integrable Chaplygin system with an invariant measure 
without Chaplygin Hamiltonization. 
\emph{Theoretical and Applied Mechanics}, {\bf 46} (2019) 65--88.


\bibitem{KeSh} \by{Kentaro, Y. and I.~Shigeru} \emph{Tangent and cotangent bundles: differential geometry.} Pure and Applied Mathematics, No. 16. Marcel Dekker, Inc., New York, 1973.

\bibitem{Koi} \by{Koiller, J.} 
 Reduction of some classical nonholonomic systems with symmetry.
\emph{Arch. Ration. Mech. Anal.} \vol{118} (1992), 113--148.

\bibitem{KoRiEh} \by{Koiller, J., Rios P.P.M. and  K.M.~Ehlers} Moving frames for cotangent bundles. \emph{Rep. Math. Phys.} \vol{49} (2002), 225--238.

\bibitem{KoillerRubber} \by{Koiller J. and K. Ehlers} Rubber rolling over a sphere. \emph{Regul.\ Chaot.\ Dyn.} \vol{12} (2006), 127--152.



\bibitem{LeRo} \by{de Le\'on, M. and P.~R.~Rodrigues} \emph{Methods of differential geometry in analytical mechanics.} North-Holland Mathematics Studies, 158. North-Holland Publishing Co., Amsterdam, 1989.

\bibitem{LeMaMa} \by{de Le\'on, M., Marrero, J. C. and D. Mart{\'\i}n de Diego} Linear almost Poisson structures and Hamilton-Jacobi equation. Applications to nonholonomic mechanics. \emph{J. Geom. Mech.} \vol{2} (2010), 159--198.

\bibitem{LeMa} \by{de Le\'on, M. and D. Mart{\'\i}n de Diego} On the geometry of nonholonomic Lagrangian systems. \emph{J. Math. Phys.} \vol{
37} (1996), 3389--3414

\bibitem{LiMa} \by{Libermann, P. and  C.~M.~Marle, } \emph{Symplectic geometry and analytical mechanics.} Translated from the French by Bertram Eugene Schwarzbach. Mathematics and its Applications, 35. D. Reidel Publishing Co., Dordrecht, 1987.

\bibitem{Marle} \by{Marle, Ch.M.} Various approaches
to conservative and nonconservative nonholonomic systems. Proc. Workshop on 
Non-Holonomic Constraints in Dynamics (Calgary, August 26-29, 1997) \emph{Rep. Math.
Phys.} \vol{42} (1998) 211--229.


\bibitem{MaRa} \by{Marsden J.E. and T.S.~Ratiu} \title{Introduction to Mechanics with Symmetry} 
\info{Texts in Applied Mathematics  \vol{17} Springer-Verlag 1994}

\bibitem{NeiFu} \by{Neimark Y. and N.~A.~Fufaev} {\em Dynamics of Nonholonomic Systems.} Amer. Math. Soc. Translations 33, 1972

\bibitem{On} \by{O'Neill, Barrett} \title{Semi-Riemannian geometry. With applications to relativity} \info{Pure and Applied Mathematics, \vol{103}. Academic Press, Inc. [Harcourt Brace Jovanovich, Publishers], New York, 1983}


\bibitem{Stanchenko} \by{Stanchenko, S.} Nonholonomic Chaplygin systems. \emph{Prikl. Mat. Mekh.} \vol{53},16--23;
English trans.: \emph{J. Appl. Math. Mech.} \vol{53} (1989),  11--17.



\bibitem{vdSM} \by{van der Schaft, A.J., and  B.M.~Maschke} On the Hamiltonian formulation of nonholonomic mechanical systems.
\emph{Reports on Math.\ Phys.} \vol{34} (1994),  225--233.

\bibitem{VF} \by{Vershik,~A.~M. and L.~D.~Fadeev } Lagrangian mechanics in invariant
form, {\em Selecta Math. Sov.} {\bf 1}, 339--350, 1981.


\bibitem{Veselov-Veselova-1988} 
\by{Veselov, A.P.\ and L.E. Veselova} 
Integrable Nonholonomic Systems on Lie Groups. {\em Mat. Notes} \vol{44} (5-6)
(1988), 810-819. [Russian original in \textit{Mat. Zametki} \vol{44} (1988), no. 5, 604--619.]

\bibitem{Veselova} \by{Veselova, L.}
\textit{New cases of integrability of the equations of motion of a rigid body in the presence of a nonholonomic constraint.} 
Geometry, Differential Equations, and
Mechanics (in Russian), Moscow State Univ. (1986), pp. 64-68.



\bibitem{ZeBo} \by{Zenkov~D.V.\ and  A.~M.~Bloch} Invariant measures of nonholonomic
flows with internal degrees of freedom {\em Nonlinearity} \vol{16} (2003), 1793--1807.




\end{thebibliography}
\end{document}